\documentclass[12pt]{article}
\usepackage{fullpage}

\usepackage{amsmath}
\usepackage{amsthm}
\usepackage{amsfonts}
\usepackage{amssymb}
\usepackage{subfigure}
\usepackage{epsfig}
\usepackage{algorithm,algorithmic}
\usepackage{graphicx}
\usepackage{color}
\usepackage{cite}

\def\F{{\mathcal{F}}}

\def\I{{\mathcal{I}}}
\def\J{{\mathcal{J}}}
\def\L{{\mathcal{L}}}
\def\M{{\mathcal{M}}}
\def\S{{\mathcal{S}}}
\def\X{{\mathcal{X}}}
\def\Y{{\mathcal{Y}}}
\def\O{{\mathcal{O}}}

\def\RR{{\mathbb{R}}}

\def\x{{\bf x}}

\def\u{{\bf u}}

\def\0{{\bf 0}}

\def\bnabla{\boldsymbol{\nabla}}

\def\bsigma{\boldsymbol{\sigma}}

\def\Dpartial#1#2{ {\partial #1 \over \partial #2} }

\def\Bmp#1{ \begin{minipage}{#1} }
\def\Emp{ \end{minipage} }

\newcommand{\argmin}{\operatorname{argmin}}
\newtheorem{theorem}{Theorem}[section]

\begin{document}

\title{Optimal Reconstruction of Material Properties
       in Complex Multiphysics Phenomena}
\author{Vladislav Bukshtynov$^{*}$ and Bartosz Protas$^{\dag}$}
\date{}

\maketitle

\begin{center}
$^{*}$School of Computational Science and Engineering, McMaster University \\
1280 Main Street West, Hamilton, Ontario, CANADA L8S 4K1\\
e-mail: bukshtu@math.mcmaster.ca
\vspace*{0.5cm} \\
$^{\dag}$Department of Mathematics and Statistics, McMaster University\\
1280 Main Street West, Hamilton, Ontario, CANADA L8S 4K1 \\
e-mail: bprotas@mcmaster.ca
\end{center}

\begin{abstract}
  We develop an optimization--based approach to the problem of
  reconstructing \newline temperature--dependent material properties
  in complex thermo--fluid systems described by the equations for the
  conservation of mass, momentum and energy. Our goal is to estimate
  the temperature dependence of the viscosity coefficient in the
  momentum equation based on some noisy temperature measurements,
  where the temperature is governed by a separate energy equation. We
  show that an elegant and computationally efficient solution of this
  inverse problem is obtained by formulating it as a PDE--constrained
  optimization problem which can be solved with a gradient--based
  descent method. A key element of the proposed approach, the cost
  functional gradients are characterized by mathematical structure
  quite different than in typical problems of PDE--constrained
  optimization and are expressed in terms of integrals defined over
  the level sets of the temperature field. Advanced techniques of
  integration on manifolds are required to evaluate numerically such
  gradients, and we systematically compare three different methods.
  As a model system we consider a two--dimensional unsteady flow in a
  lid--driven cavity with heat transfer, and present a number of
  computational tests to validate our approach and illustrate its
  performance. \\


\end{abstract}

\section{Introduction}
\label{sec:intro}

In this work we propose and validate a computational approach to the
reconstruction of material properties in complex multiphysics
phenomena based on incomplete and possibly noisy measurements. The
material properties we are interested in here are the transport
coefficients characterizing diffusion processes such as the viscosity
or the thermal conductivity, and we focus on problems in which these
coefficients depend on the state variables in the system. By the
``multiphysics'' aspect we mean situations in which the material
property used in one conservation equation is a function of a state
variable governed by a different conservation equation, e.g.,
reconstruction of the temperature dependence of the viscosity
coefficient used in the momentum equation, where the temperature is
governed by a separate energy equation, which is the specific model
problem investigated in this study. This research is motivated by
questions arising in the computational analysis and optimization of
advanced welding processes which involves modelling complex alloys
in the liquid phase at high temperatures \cite{vplg09a}.

Inverse problems of parameter estimation for partial differential
equations (PDEs) have received significant attention in the
literature, both regarding theoretical \cite{i06} and practical
aspects \cite{t05}. However, most of the attention focused on problems
in which the material properties are functions of the space variable
(i.e., the independent variable in the problem). Such problems are, at
least in principle, relatively well understood and represent the
foundation of, for example, numerous imaging techniques in medicine
\cite{ns02} and earth sciences \cite{gd96}.  The problem considered
here is in fact different, in that the material properties are sought
as functions of the state (dependent) variables in the system which
gives rise to a number of computational challenges absent in the
``classical'' parameter estimation problem. Other than the seminal
work of Chavent and Lemonnier \cite{chl74}, our earlier study
\cite{bvp10} concerning a simplified model problem and a few
investigations of fully discrete formulations surveyed in
\cite{bvp10}, there does not seem to be much literature concerning
computational methods for this type of parameter estimation problems.
One way to solve such inverse problems is to formulate them as
optimization problems and this is the approach we will follow focusing
on the ``optimize--then--discretize'' paradigm in which the optimality
conditions are formulated at the continuous (PDE) level and only then
discretized. The goal of this investigation is to extend the approach
formulated in \cite{bvp10} for a simple model to a realistic
multiphysics problem involving time--dependent fluid flow in a
two--dimensional (2D) domain. As will be shown below, a number of
computational difficulties will need to be overcome in order to
achieve this goal.

As a key contribution of this work, we address a number of
computational challenges related to accurate and efficient evaluation
of cost functional gradients which are critical to the implementation
of the proposed approach. More specifically, these gradients are given
in terms of integrals of expressions involving state and adjoint
variables defined on a grid over contours given by the level sets of
the temperature field.  A number of techniques have been proposed for
the numerical evaluation of integrals defined over manifolds defined
by level--set functions. Some of them rely on regularized Dirac delta
and Heaviside functions \cite{ZahTorn10, EngTornTs04}, or
discretization of the Dirac delta function \cite{Smereka2006, Mayo84,
  Beale08}.  Similar approaches, based on approximations of the Dirac
delta functions obtained using the level--set function and its
gradient, were developed by Towers \cite{Tow09, Tow07}.  The family of
geometric approaches proposed by Min and Gibou in \cite{MinGib07,
  MinGib08} relies on a decomposition of the entire domain into
simplices. We emphasize that the problem discussed in this work is in
fact more complicated, as the computation of our cost functional
gradients requires evaluation of the corresponding integrals for the
level--set values spanning the entire state space of interest, hence
there are also additional issues related to the discretization of the
state space which were outside the scope of references
\cite{EngTornTs04,ZahTorn10,Mayo84,Smereka2006,Beale08,MinGib07,MinGib08,
  Tow09,Tow07}. Thus, in order to address these questions and assess
the different trade--offs in the choice of the numerical parameters we
will compare the computational performance of three different methods
for the evaluation of cost functional gradients.

The structure of this paper is as follows: in the next Section we
formulate our model problem, in the following Section we cast the
problem of parameter estimation as an optimization problem, an
adjoint--based gradient--descent algorithm is formulated in Section
\ref{sec:grad_adj}, in Section \ref{sec:reg} we outline some
regularization strategies needed in the presence of measurement
noise, whereas in Section \ref{sec:grad_eval} we analyze three
different numerical approaches to the evaluation of the cost
functional gradients; extensive computational results are presented
in Section \ref{sec:results} with discussion and conclusions
deferred to Section \ref{sec:final}.

\section{Model Problem}
\label{sec:estim_model_problem}

Let $\Omega \subset \RR^d$, $d=2,3$, be the spatial domain on which
our model problem is formulated. To fix attention, but without loss of
generality, in the present investigation we focus on the problem of a
reconstruction of the temperature dependence $\mu = \mu(T)$ of the
viscosity coefficient $\mu\: : \: \RR \rightarrow \RR^+$ in the
momentum equation (Navier-Stokes equation), where the temperature $T$
is governed by a separate energy equation
\begin{subequations}
\label{eq:coupled_PDEs}
\begin{alignat}{2}
\partial_t \u + (\u \cdot \bnabla) \u + \bnabla p - \bnabla \cdot
\left[ \mu(T) [ \bnabla \u + (\bnabla \u)^T] \right] & = & 0
\qquad & \textrm{in} \ \Omega,
\label{eq:coupled_NS} \\
\bnabla \cdot \u & = & 0 \qquad & \textrm{in} \ \Omega,
\label{eq:coupled_incompress} \\
\partial_t T + (\u \cdot \bnabla) T - \bnabla \cdot
[ k \bnabla T] & = & 0 \qquad & \textrm{in} \ \Omega,
\label{eq:coupled_heat}
\end{alignat}
\end{subequations}
subject to appropriate Dirichlet (or Neumann) boundary and
initial conditions
\begin{subequations}
\label{eq:coupled_BC}
\begin{alignat}{2}
\u & = \u_B \qquad && \textrm{on} \ \partial\Omega, \label{eq:BC_u} \\
T & = T_B \qquad && \textrm{on} \ \partial\Omega, \label{eq:BC_T} \\
\u(\cdot,0) & = \u_0, \ T(\cdot,0) = T_0 \qquad && \textrm{in} \ \Omega. \label{eq:IC_u_T}
\end{alignat}
\end{subequations}
The specific inverse problem we address in this investigation is
formulated as follows. Given a set of time--dependent ``measurements''
$\{\tilde T_i(t) \}_{i=1}^M$ of the state variable (temperature)
$T$ at a number of points $\{\x_i\}_{i=1}^M$ in the domain $\Omega$
(or along the boundary $\partial \Omega$) and obtained within the time
window $t \in [0, t_f]$, we seek to reconstruct the constitutive
relation $\mu = \mu(T)$ such that solutions of problem
\eqref{eq:coupled_PDEs}-\eqref{eq:coupled_BC} obtained with this
reconstructed function will fit best the available measurements.

In regard to reconstruction of constitutive relations in general, it
is important that such relations be consistent with the second
principle of thermodynamics \cite{m89}. There exist two mathematical
formalisms, one due to Coleman and Noll \cite{cn63} and another one
due to Liu \cite{l72}, developed to ensure in a very general setting
that a given form of the constitutive relation does not violate the
second principle of thermodynamics. In continuous thermodynamical and
mechanical systems this principle is expressed in terms of the
Clausius--Duhem inequality \cite{tpcm08a} which in the case of our
present model problem \eqref{eq:coupled_PDEs}--\eqref{eq:coupled_BC}
reduces to the statement that $\mu(T) > 0$ for all possible values of
$T$.

In our discussion below we will also need definitions of the following
intervals, cf.~Figure \ref{fig:domains}:
\begin{itemize}
\item $[T_{\alpha}, T_{\beta}] \triangleq [\min_{\x \in
    \overline{\Omega}} T(\x), \max_{\x \in \overline{\Omega}} T(\x)]$
  which represents the temperature range spanned by the solution of problem
  \eqref{eq:coupled_PDEs}; thus, following \cite{k03}, we will refer
  to the interval $\I \triangleq [T_{\alpha}, T_{\beta}]$ as the {\em
    identifiability interval},

\item $\L \triangleq [T_a, T_b]$, where $T_a \le T_{\alpha}$ and $T_b
  \ge T_{\beta}$; this will be the temperature interval on which we will seek to
  obtain a reconstruction of the material property; we note that in
  general the interval $\L$ will be larger than the identifiability
  interval, i.e., $\I \subseteq \L$, and

\item $\M \triangleq [\min_{1 \leq i \leq M} \min_{0 < t \leq t_f}
  \tilde T_i(t), \max_{1 \leq i \leq M} \max_{0 < t \leq t_f} \tilde T_i(t)]$
  which defines the temperature range spanned by the measurements
  $\{\tilde{T}_i \}_{i=1}^M$; this interval is always contained
  the identifiability interval $\I$, i.e., $\M \subseteq \I$; we will refer
  to the interval $\M$ as the {\em measurement span}.
\end{itemize}

\begin{figure}
\begin{center}
\includegraphics[width=1.0\textwidth]{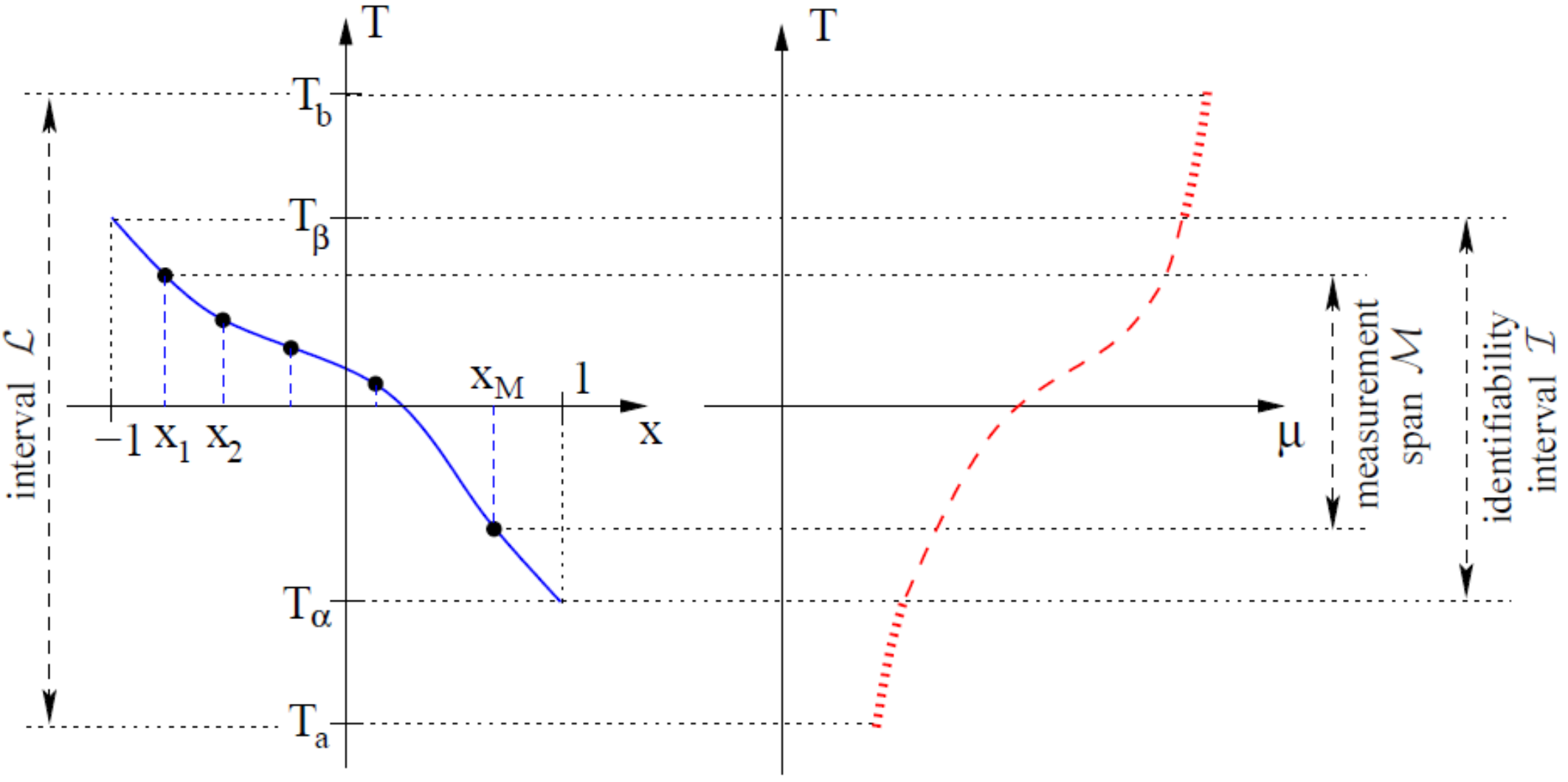}
\end{center}
\caption{Schematic showing (left) the solution $T(t_0,x)$ at some
  fixed time $t_0$ and (right) the corresponding constitutive relation
  $\mu(T)$ defined over their respective domains, i.e., $\Omega =
  (-1,1)$ and the identifiability region $\I$.  The thick dotted line
  represents an extension of the constitutive relation $\mu(T)$ from
  $\I$ to the interval $\L$. In the Figure on the right the horizontal
  axis is to be interpreted as the ordinate.}
\label{fig:domains}
\end{figure}

\section{Parameter Estimation as an Optimization Problem}
\label{sec:param}

It is assumed that the constitutive relations $\mu(T)$ are
differentiable functions of the state variable (temperature) and
belong to the following set
\begin{equation}
\S_{\mu} = \{\mu(T) \ \textrm{piecewise} \ C^1 \ \textrm{on} \ \L;
\ 0< m_{\mu} \leq \mu(T) \leq M_{\mu} < \infty,\ \forall \,T \in \L\},
\label{eq:S_mu}
\end{equation}
where $m_{\mu}, M_{\mu} \in \RR^+$. We will also assume that the set
$\S_{\mu}$ consisting of constitutive relations $\mu(T)$ defined on $\L$ is
embedded in a Hilbert (function) space $\X$ to be specified below.
Solving our parameter estimation problem is therefore equivalent to
finding a solution to the operator equation
\begin{equation}
\F(\mu) = T,
\label{eq:F}
\end{equation}
where $\F \, : \, \S_{\mu} \rightarrow \left( L_2([0,t_f]) \right)^M$ is
the map from the constitutive relations to the measurements. An
approach commonly used to solve such problems consists in
reformulating them as least--squares minimization problems which in
the present case can be done by defining the cost functional $\J \: :
\: \X \rightarrow \RR$ as
\begin{equation}
\J(\mu) \triangleq \frac{1}{2} \int_0^{t_f} \sum_{i=1}^M \left[
T(\tau,\x_i;\mu) - \tilde{T}_i(\tau) \right]^2 d\tau,
\label{eq:J0}
\end{equation}
where the dependence of the temperature field $T(\cdot;\mu)$ on the
form of the constitutive relation $\mu = \mu(T)$ is given by governing
system \eqref{eq:coupled_PDEs}--\eqref{eq:coupled_BC}.
The optimal reconstruction $\hat{\mu}$ is obtained
as a minimizer of cost functional \eqref{eq:J0}, i.e.,
\begin{equation}
\hat{\mu} = \underset{\mu \in \S_{\mu}}{\argmin}\,\J(\mu).
\label{eq:minJ}
\end{equation}
We recall that the constitutive property is required to satisfy the
positivity condition $\mu(T) > 0$ for all $T \in \L$,
cf.~\eqref{eq:S_mu}. Therefore, the optimal reconstruction $\hat{\mu}$
should in fact be obtained as an {\em inequality}--constrained
minimizer of cost functional \eqref{eq:J0}, i.e.,
\begin{equation}
\hat{\mu} = \underset{\stackrel{\mu \in \X,}{\mu(T)>0, \ T \in \L}}{\argmin} \J(\mu).
\label{eq:min}
\end{equation}
We add that in problems involving constitutive relations depending on
several state variables the inequality constraint $\mu(T) > 0$ will be
replaced with a more general form of the Clausius--Duhem inequality
\cite{tpcm08a}. Different computational approaches for converting
inequality--constrained optimization problems to unconstrained
formulations are surveyed in \cite{r06, v02}. Here we follow a
straightforward approach based on the so-called ``slack'' variable
\cite{bv04}. We define a new function $\theta(T)\: : \: \RR \rightarrow \RR$
such that
\begin{equation}
\mu(T) = \theta^2(T) + m_{\mu},
\label{eq:slack}
\end{equation}
where $m_{\mu}$ is a lower bound for $\mu(T)$, cf.~\eqref{eq:S_mu}.
This change of variables allows us to transform the {\em
  inequality}--constrained optimization problem \eqref{eq:min} to a
new unconstrained one
\begin{equation}
\hat{\theta} = \underset{\theta \in \X}{\argmin}\,\J(\theta),
\label{eq:minJslack}
\end{equation}
where the constraint $\mu(T) > 0$ is satisfied automatically when
minimization is performed with respect to the new variable $\theta(T)$.
In view of \eqref{eq:S_mu}, we note that the new
optimization variable $\theta$ belongs to the following set
\begin{equation}
\S_{\theta} = \{\theta(T) \ \textrm{piecewise} \ C^1 \ \textrm{on} \ \L; \
| \theta(T) | < \sqrt{M_{\mu}-m_{\mu}},\ \forall \,T \in \L\}.
\label{eq:S_theta}
\end{equation}
The governing PDE system \eqref{eq:coupled_PDEs} can thus be rewritten
in the form
\begin{subequations}
\label{eq:coupled_PDEs_slack}
\begin{alignat}{2}
\partial_t \u + (\u \cdot \bnabla) \u + \bnabla p - \bnabla \cdot
\left[ (\theta^2(T) + m_{\mu}) [ \bnabla \u + (\bnabla \u)^T] \right] & = & 0
\qquad & \textrm{in} \ \Omega,
\label{eq:coupled_NS_slack} \\
\bnabla \cdot \u & = & 0 \qquad & \textrm{in} \ \Omega,
\label{eq:coupled_incompress_slack} \\
\partial_t T + (\u \cdot \bnabla) T - \bnabla \cdot
[ k \bnabla T] & = & 0 \qquad & \textrm{in} \ \Omega,
\label{eq:coupled_heat_slack}
\end{alignat}
\end{subequations}
subject to Dirichlet boundary and initial conditions \eqref{eq:coupled_BC}.
The new problem \eqref{eq:minJslack} requires redefining cost functional
\eqref{eq:J0} in terms of the new variable
\begin{equation}
\J(\theta) = \frac{1}{2} \int_0^{t_f} \sum_{i=1}^M \left[ T(\tau,\x_i;\theta)
- \tilde{T}_i(\tau) \right]^2 \, d\tau.
\label{eq:J0_slack}
\end{equation}
Problem \eqref{eq:minJslack} is characterized by the first--order
optimality condition which requires the G\^ateaux differential of cost
functional \eqref{eq:J0_slack}, defined
as $\J'(\theta;\theta') \triangleq
\lim_{\epsilon\rightarrow 0} \epsilon^{-1} [\J(\theta+\epsilon \theta') -
\J(\theta)]$, to vanish for all perturbations $\theta' \in \X$ \cite{l69}, i.e.,
\begin{equation}
\forall_{\theta'\in\X} \ \ \J'(\hat{\theta};\theta') = 0.
\label{eq:dJ}
\end{equation}
The (local) optimizer $\hat{\theta}$ can be computed with the following
gradient descent algorithm as $\hat{\theta} = \lim_{n\rightarrow
\infty} \theta^{(n)}$, where	
\begin{equation}
\left\{
\begin{alignedat}{2}
&\theta^{(n+1)} && = \theta^{(n)} - \tau^{(n)} \bnabla_{\theta}\J(\theta^{(n)}), \qquad n=1,\dots, \\
&\theta^{(1)}   && = \theta_0,
\end{alignedat}
\right.
\label{eq:desc}
\end{equation}
in which $\bnabla_{\theta}\J(\theta)$ represents the gradient of cost
functional $\J(\theta)$ with respect to the control variable $\theta$
(we will adopt the convention that a subscript on the operator
$\bnabla$ will be used when differentiation is performed with respect
to variables other than $\x$), $\tau^{(n)}$ is the length of the step
along the descent direction at the $n$--th iteration, whereas
$\theta_0 = \sqrt{\mu_0 - m_{\mu}}$ is the initial guess taken, for
instance, corresponding to a constant $\mu_0$, or some other
appropriate initial approximation. For the sake of clarity,
formulation \eqref{eq:desc} represents the steepest--descent
algorithm, however, in practice one typically uses more advanced
minimization techniques, such as the conjugate gradient method, or one
of the quasi--Newton techniques \cite{nw00}. We note that, since
minimization problem \eqref{eq:minJslack} is in general nonconvex,
condition \eqref{eq:dJ} characterizes only a {\em local}, rather than
{\em global}, minimizer.

The key ingredient of minimization algorithm \eqref{eq:desc} is
computation of the cost functional gradient
$\bnabla_{\theta}\J(\theta)$. We emphasize that, since $\theta =
\theta(T)$ is a continuous variable, the gradient
$\bnabla_{\theta}\J(\theta)$ represents in fact the
infinite--dimensional sensitivity of $\J(\theta)$ to perturbations of
$\theta(T)$.  This gradient can be determined based on suitably
defined {\em adjoint variables} (Lagrange multipliers) obtained from
the solution of the corresponding {\em adjoint system}. Since this
derivation differs in a number of imported technical details from
analogous derivations in ``standard'' PDE--constrained optimization
problems, it will be reviewed in Section \ref{sec:grad_adj}. The
expression for the gradient is then validated for consistency in
Section \ref{sec:kappa}.

\section{Cost Functional Gradients via Adjoint--based Analysis}
\label{sec:grad_adj}

Since the new variable $\theta(T)$ belongs to set $\S_{\theta}$,
cf.~\eqref{eq:S_theta}, we will seek to reconstruct $\theta(T)$
as elements of the Sobolev space $H^1(\L)$, so that the gradient
$\bnabla_{\theta}\J$ will need to be obtained with respect to the
corresponding inner product. However, in order to make the derivation
procedure easier to follow, we will first obtain an expression for
the gradient in the space $L_2(\L)$, and only then will obtain the
Sobolev gradients which will be eventually used in the solution of
optimization problem \eqref{eq:minJslack}. In all these steps our
transformations will be formal. We begin by computing the directional
(G\^ateaux) differential of cost functional \eqref{eq:J0_slack} which yields
\begin{equation}
\J'(\theta;\theta') = \int_0^{t_f} \sum_{i=1}^M [T(\tau,\x_i;\theta) - \tilde{T}_i(\tau)] \,
T'(\tau,\x_i; \theta, \theta') \, d\tau,
\label{eq:dJ2_coupled}
\end{equation}
where the perturbation variable $T'(\theta, \theta')$ satisfies the
perturbation system obtained from
\eqref{eq:coupled_PDEs}--\eqref{eq:coupled_BC}. Next, we invoke the
Riesz representation theorem \cite{b77} for the directional
differential $\J'(\theta;\cdot)$ which yields
\begin{equation}
\J'(\theta; \theta') =
\Big\langle \bnabla_{\theta}\J, \theta' \Big\rangle_{\X},
\label{eq:riesz}
\end{equation}
where $\langle\cdot,\cdot\rangle_{\X}$ represents an inner product in
the Hilbert space $\X$ (we will first set $\X = L_2(\L)$ and
afterwards change this to $\X = H^1(\L)$). We note that the expression
on the right--hand side (RHS) in \eqref{eq:dJ2_coupled} is not
consistent with Riesz representation \eqref{eq:riesz}, since, as will
be shown below, the perturbation variable $\theta'$ is hidden in the
system defining $T'(\theta,\theta')$. However, this expression can be
transformed to Riesz form \eqref{eq:riesz} with the help of a
suitably--defined adjoint variable, a result which is stated in
Theorem \ref{thm:riesz} below. The main aspect in which this
derivation differs from the standard adjoint analysis \cite{g03} is
that the inner product in Riesz identity \eqref{eq:riesz} is defined
using the state variable (temperature) as the integration variable,
whereas the variational formulation of
\eqref{eq:coupled_PDEs}--\eqref{eq:coupled_BC} is defined using
integration with respect to the independent variables ($\x$ and $t$).

\begin{theorem}
  Let $\Omega$ be a sufficiently regular open bounded domain and
  $\theta'\in \X = H^1(\L)$. We assume that the solutions $\u$ and $T$
  of system \eqref{eq:coupled_PDEs}--\eqref{eq:coupled_BC} are
  sufficiently smooth. Then, the Riesz representation of directional
  differential \eqref{eq:dJ2_coupled} has the form
\begin{equation}
\J'(\theta;\theta') =
-2 \int_{-\infty}^{\infty} \int_{\Omega} \delta(T(\x)-s) \, \theta(s) \,
\left[ \int^{t_f}_0 [\bnabla \u + (\bnabla \u)^T] : \bnabla \u^*  \, d\tau \right ]
\, \theta'(s) \, d\x \, ds,
\label{eq:rrgT}
\end{equation}
where $\delta(\cdot)$ denotes Dirac delta function and the adjoint state
$\{ \u^*, T^* \}$ is defined as the solution of the system
\begin{subequations}
\label{eq:adjoint_coupled}
\begin{alignat}{2}
- \partial_t \u^* - (\u \cdot \bnabla) \u^* - \bnabla \cdot \bsigma^* +
\u^* \cdot (\bnabla \u)^T + T^* \bnabla T & = 0
\qquad & \mathrm{in} \ \Omega, \label{eq:adjoint_coupled_a} \\
\bnabla \cdot \u^* & = 0 \qquad & \mathrm{in} \ \Omega, \label{eq:adjoint_coupled_b} \\
- \partial_t T^* - (\u \cdot \bnabla) T^* - \bnabla \cdot
[k \bnabla T^*] + 2\theta(T)\,\dfrac{d\theta}{dT}(T) \,[\bnabla \u + (\bnabla \u)^T] : \bnabla^* \u
\nonumber &\\
= \sum_{i=1}^M [T(\x_i;\theta) - \tilde{T}_i] \delta (\x - \x_i)&
\qquad & \mathrm{in} \ \Omega, \label{eq:adjoint_coupled_c}
\end{alignat}
\end{subequations}
where $\bsigma^* \triangleq -p^*\I + (\theta^2(T) + m_{\mu})
\left[ \bnabla \u^* + (\bnabla \u^*)^T \right]$,
with the following boundary and terminal conditions
\begin{equation}
\label{eq:BC_adjoint}
\begin{aligned}
\u^* & = 0 \qquad && \mathrm{on} \ \partial\Omega,\\
T^* & = 0 \qquad && \mathrm{on} \ \partial\Omega,\\
\u^*(\cdot; t_f) & = 0, \ T^*(\cdot; t_f) = 0 \qquad && \mathrm{in} \ \Omega.
\end{aligned}
\end{equation}
\label{thm:riesz}
\end{theorem}

\begin{proof}
We will denote the stress tensor $\bsigma \triangleq -p\I + (\theta^2(T) +
m_{\mu}) \left[ \bnabla \u + (\bnabla \u)^T] \right]$ and rewrite the
governing system \eqref{eq:coupled_PDEs_slack} as
\begin{equation}
\label{eq:coupled_PDEs_sigma}
\begin{aligned}
\partial_t \u + (\u \cdot \bnabla) \u - \bnabla \cdot \bsigma & = & 0
\qquad & \textrm{in} \ \Omega,\\
\bnabla \cdot \u & = & 0 \qquad & \textrm{in} \ \Omega,\\
\partial_t T + (\u \cdot \bnabla) T - \bnabla \cdot
[ k \bnabla T] & = & 0 \qquad & \textrm{in} \ \Omega
\end{aligned}
\end{equation}
with the boundary and initial conditions \eqref{eq:coupled_BC}.
Perturbing the state variables $\u$, $p$ and $T$, which are functions
of time and space, we get
\begin{equation}
\label{eq:states_perturbed}
\begin{aligned}
\u & = \u_0 + \epsilon \u' + \O(\epsilon^2), \\
p & = p_0 + \epsilon p' + \O(\epsilon^2), \\
T & = T_0 + \epsilon T' + \O(\epsilon^2),
\end{aligned}
\end{equation}
so that the corresponding expansion of the constitutive relation
$\theta(T)$ will have the following form
\begin{equation}
\label{eq:slack_perturbed}
\theta(T) = \theta_0(T) + \epsilon \theta'(T) + \O(\epsilon^2)
= \theta_0(T_0) + \epsilon \dfrac{d\theta}{dT}(T_0) \, T'(\theta_0;\theta') + \epsilon \theta'(T_0) + \O(\epsilon^2),
\end{equation}
where the subscript ``$0$'' is used to denote the unperturbed (reference) material property and the state variable,
whereas the prime denotes the corresponding perturbations. We also have
\begin{equation}
\label{eq:slack_2_perturbed}
\theta^2(T) = \theta^2_0(T_0) + 2 \epsilon \theta_0 (T_0) \, \dfrac{d\theta}{dT}(T_0) \, T'(\theta_0;\theta') +
2 \epsilon \theta_0 (T_0) \theta'(T_0) + \O(\epsilon^2).
\end{equation}
Substituting \eqref{eq:states_perturbed} and \eqref{eq:slack_2_perturbed}
into \eqref{eq:coupled_PDEs_sigma}, collecting terms corresponding to
$\epsilon$ in different powers and denoting
\begin{equation*}
\begin{aligned}{}
\hat \bsigma & \triangleq
-p'\I + (\theta^2(T) + m_{\mu}) \left[ \bnabla \u' + (\bnabla \u')^T \right], \\
\tilde \bsigma & \triangleq \left( 2\theta(T) \, \dfrac{d\theta}{dT}(T) \, T'(\theta_0;\theta') + 2\theta(T) \, \theta'(T)
\right) \left[ \bnabla \u + (\bnabla \u)^T \right],
\end{aligned}
\end{equation*}
we obtain the perturbation (sensitivity) system corresponding to
\eqref{eq:coupled_PDEs_slack}
\begin{subequations}
\label{eq:perturb}
\begin{alignat}{2}
\partial_t \u' + (\u' \cdot \bnabla) \u + (\u \cdot \bnabla) \u' -
\bnabla \cdot ( \hat \bsigma + \tilde \bsigma )  & = & 0
\qquad & \textrm{in} \ \Omega, \label{eq:perturb_a} \\
\bnabla \cdot \u' & = & 0 \qquad & \textrm{in} \ \Omega,  \label{eq:perturb_b} \\
\partial_t T' + (\u' \cdot \bnabla) T + (\u \cdot \bnabla) T' - \bnabla \cdot
[ k \bnabla T'] & = & 0 \qquad & \textrm{in} \ \Omega  \label{eq:perturb_c}
\end{alignat}
\end{subequations}
with the following boundary and initial conditions
\begin{subequations}
\label{eq:BC_perturb}
\begin{alignat}{2}
\u' & = 0 \qquad && \textrm{on} \ \partial\Omega,\\
T' & = 0 \qquad && \textrm{on} \ \partial\Omega,\\
\u'(\cdot,0) & = 0, \ T'(\cdot,0) = 0 \qquad && \textrm{in} \ \Omega.
\end{alignat}
\end{subequations}
Then, integrating equation \eqref{eq:perturb_a} against $\u^*$,
equation \eqref{eq:perturb_b} against $p^*$, and equation
\eqref{eq:perturb_c} against $T^*$ over the space domain $\Omega$ and
time $[0,t_f]$, integrating by parts
and factorizing $\u'$, $T'$ and $p'$, we arrive at the following relation
\begin{equation}
\label{eq:adjoint_final}
\begin{aligned}
&\int^{t_f}_0 \int_{\Omega} \left[ -\partial_t \u^* + \u^* \cdot (\bnabla \u)^T - (\u \cdot \bnabla) \u^*
-\bnabla \cdot \bsigma^* + T^* \bnabla T \right] \cdot \u' \, d\x \, d\tau\\
- &\int^{t_f}_0 \int_{\Omega} (\bnabla \cdot \u^*) p' \, d\x \, d\tau\\
+ &\int^{t_f}_0 \int_{\Omega} \left[- \partial_t T^* - (\u \cdot \bnabla) T^* - \bnabla \cdot (k \bnabla T^*)
+ 2\theta(T)\dfrac{d\theta}{dT}(T) [\bnabla \u + (\bnabla \u)^T] : \bnabla \u^* \right] T' \, d\x \, d\tau\\
+ &\int^{t_f}_0 \int_{\Omega} 2\theta(T)\theta'(T) [\bnabla \u + (\bnabla \u)^T] : \bnabla \u^* \, d\x \, d\tau = 0.
\end{aligned}
\end{equation}
We now require that the adjoint variables $\u^*$, $p^*$ and $T^*$
satisfy system \eqref{eq:adjoint_coupled}--\eqref{eq:BC_adjoint}. We
also note that owing to the judicious choice of the RHS term in
\eqref{eq:adjoint_coupled_c}, the last term in relation
\eqref{eq:adjoint_final} is in fact equal to the directional
differential $\J'(\theta;\theta')$, so that we have
\begin{equation}
\J'(\theta;\theta') =
-2 \int^{t_f}_0 \int_{\Omega} \theta(T(\x,\tau))\theta'(T(\x,\tau)) \,
[\bnabla \u(\x,\tau) + (\bnabla \u(\x,\tau))^T]
: \bnabla \u^*(\x,\tau) \,d\x \, d\tau,
\label{eq:rrgd}
\end{equation}
where, for emphasis, we indicated the integration variables as
arguments of the state and adjoint variables. We note that this
expression is still not in Riesz form \eqref{eq:riesz}, where
integration must be performed with respect to the state variable
(temperature $T$). Thus, we proceed to express for any given function
$f(T)$ its pointwise evaluation at $T(\x)$ through the following
integral transform. It is defined using a ``change--of--variable''
operator, denoted $\Pi$, such that for given functions $f \, : \, \RR
\rightarrow \RR$ and $T \, : \, \Omega \rightarrow \RR$, we have
\begin{equation}
f(T(\x)) = \int_{-\infty}^{+\infty} \delta (T(\x) - s) f(s) \, ds \triangleq (\Pi f)(\x).
\label{eq:change_vars}
\end{equation}
Using this transform to express $f(T(\x)) =
\theta(T(\x,\tau))\theta'(T(\x,\tau))$ in \eqref{eq:rrgd} and changing
the order of integration (Fubini's Theorem), we obtain expression
\eqref{eq:rrgT} which is the required Riesz representation
\eqref{eq:riesz} of directional differential \eqref{eq:dJ2_coupled}.
\end{proof}

We remark that we were able to prove an analogous result using a
simpler approach based on the Kirchhoff transform in \cite{bvp10},
where both the constitutive relation and the state variable were
governed by the same equation (i.e., the problem was not of the
``multiphysics'' type).

With the Riesz representation established in \eqref{eq:rrgT}, we now
proceed to identify expressions for the cost functional gradient
$\bnabla_{\theta} \J$ according to \eqref{eq:riesz} using different
spaces $\X$. While this is not the gradient that we will use in actual
computations, we analyze first the ``simplest'' case when $\X = L_2(\L)$,
i.e., the space of functions square integrable on $[T_a,T_b]$, as it
already offers some interesting insights into the structure of the problem.
The $L_2$ gradient of the cost functional hence takes the form
\begin{equation}
\label{eq:grad_T_coupled}
\bnabla_{\theta}^{L_2} \J(s) = -2 \int_0^{t_f} \int_{\Omega} \delta (T(\x) - s)
\, \theta(s) \, [\bnabla \u + (\bnabla \u)^T] : \bnabla \u^* \ d\x \, d\tau.
\end{equation}
As was discussed at length in \cite{bvp10}, the $L_2$ gradients are
not suitable for the reconstruction of material properties in the
present problem, because in addition to lacking necessary smoothness,
they are not defined outside the identifiability region (other than
perhaps through a trivial extension with zero). Given the regularity
required of the constitutive relations, cf.~\eqref{eq:S_theta}, the cost
functional gradients should be elements of the Sobolev space $H^1(\L)$
of functions with square--integrable derivatives on $\L$. Using
\eqref{eq:riesz}, now with $\X = H^1(\L)$, we
obtain
\begin{equation}
\begin{aligned}
\J'(\theta; \theta') & = \Big\langle \bnabla_{\theta}^{L_2} \J, \theta' \Big\rangle_{L_2(\L)} =
\Big\langle\bnabla_{\theta}^{H^1} \J, \theta' \Big\rangle _{H^1(\L)} \\
& = \int_{T_{a}}^{T_{b}} \left[ (\bnabla_{\theta}^{H^1} \J) \, \theta' +
\ell^2 \frac{d (\bnabla_{\theta}^{H^1}\J)}{ds} \frac{d\theta'}{ds} \right] \, ds
\end{aligned}
\label{eq:dJ4}
\end{equation}
in which $\ell \in \RR$ is a parameter with the meaning of a
``temperature--scale'' [we note that the $L_2$ inner product is
recovered by setting $\ell=0$ in \eqref{eq:dJ4}]. Performing
integration by parts with the assumption that the Sobolev gradient
$\bnabla_{\theta}^{H^1} \J$ satisfies the homogeneous Neumann boundary
conditions at $T = T_a, T_b$, and noting that relation \eqref{eq:dJ4}
must be satisfied for any arbitrary $\theta'$, we conclude that the
Sobolev gradient can be determined as a solution of the following
inhomogeneous elliptic boundary--value problem where the state
variable (temperature) acts as the independent variable
\begin{subequations}
\label{eq:helm}
\begin{alignat}{2}
\bnabla_{\theta}^{H^1} \J - \ell^2 \frac{d^2}{ds^2}
\bnabla_{\theta}^{H^1} \J & = \bnabla_{\theta}^{L_2} \J
\qquad && \textrm{on} \ (T_{a},T_{b}), \label{eq:helm_a} \\
\frac{d}{ds} \bnabla_{\theta}^{H^1} \J & = 0 && \textrm{for} \ s=T_{a}, T_{b}. \label{eq:helm_b}
\end{alignat}
\end{subequations}
We recall that by changing the value of the temperature--scale
parameter $\ell$ we can control the smoothness of the gradient
$\bnabla_{\theta}^{H^1} \J(\theta)$, and therefore also the relative
smoothness of the resulting reconstruction of $\theta(T)$, and hence
also the regularity of $\mu(T)$.  More specifically, as was shown in
\cite{pbh04}, extracting cost functional gradients in the Sobolev
spaces $H^p$, $p>0$, is equivalent to applying a low--pass filter to
the $L_2$ gradient with the quantity $\ell$ representing the
``cut-off'' scale. There are also other ways of defining the Sobolev
gradients in the present problem which result in gradients
characterized by a different behavior outside the identifiability
region. These approaches were thoroughly investigated in \cite{bvp10},
and since they typically lead to inferior results, they will not be
considered here. We finally conclude that iterative reconstruction of
the constitutive relation $\mu(T)$ involves the following computations
\begin{enumerate}
\item
solution of direct problem \eqref{eq:coupled_PDEs_slack} with
boundary and initial conditions \eqref{eq:coupled_BC},

\item
solution of adjoint problem \eqref{eq:adjoint_coupled}--\eqref{eq:BC_adjoint},

\item
evaluation of expression \eqref{eq:grad_T_coupled} for the cost functional gradient,

\item
computation of the smoothed Sobolev gradient by solving \eqref{eq:helm}.
\end{enumerate}
While steps (1), (2) and (4) are fairly straightforward, step (3) is
not and will be thoroughly investigated in Section \ref{sec:grad_eval}.

As we discussed in detail in \cite{bvp10}, while the Sobolev gradient
$\bnabla_{\theta}^{H^1} \J$ may be defined on an arbitrary interval
$\L \supset \I$, the actual sensitivity information is essentially
available only on the identifiability interval $\I$ (see Figure
\ref{fig:domains}). In other words, extension of the gradient outside
$\I$ via \eqref{eq:helm} does not generate new sensitivity
information.  Since, as demonstrated by our computational results
reported in \cite{bvp10}, such techniques are not capable of
accurately reconstructing the relation $\mu(T)$ on an interval $\L$
much larger than the identifiability region $\I$, here we mention a
different method to ``extend'' the identifiability region, so that the
relation $\mu(T)$ can be reconstructed on a larger interval. This can
be done in a straightforward manner by choosing appropriate
time--dependent boundary conditions for temperature $T_B$ in
\eqref{eq:BC_T} which will result in a suitable identifiability region
$\I$ and the measurement span $\M$.  Computational results
illustrating the performance of our approach with different
identifiability regions obtained using this method will be presented
in Section \ref{sec:estim}. We remark that extending the
identifiability region in this way is not possible in
time--independent problems where an iterative approach has to be used
involving solution of a sequence of reconstruction problems on shifted
identifiability regions \cite{bvp10}.

\section{Reconstruction in the Presence of Measurement Noise}
\label{sec:reg}

In this Section we discuss the important issue of reconstruction in
the presence of noise in the measurements. As can be expected based on
the general properties of parameter estimation problems \cite{t05}, and
as will be confirmed in Section \ref{sec:noise}, incorporation of
random noise into the measurements leads to an instability in the form
of small--scale oscillations appearing in the reconstructed
constitutive relations. In the optimization framework a standard
approach to mitigate this problem is Tikhonov regularization
\cite{ehn96} in which original cost functional \eqref{eq:J0_slack} is
replaced with a regularized expression of the form
\begin{equation}
\J_{\lambda}(\theta) \triangleq \J(\theta) +
\dfrac{\lambda}{2} \big\| \theta - \bar{\theta} \big\|^2_{\Y(\I)},
\label{eq:Jreg}
\end{equation}
where $\lambda \in \RR^+$ is an adjustable regularization parameter,
$\bar{\theta}(T)$ represents a constitutive relation which our
reconstruction $\theta(T)$ should not differ too much from, whereas
$\|\cdot\|_{\Y(\I)}$ is the Hilbert space norm in which we measure the
deviation $(\theta - \bar{\theta})$. Thus, the regularization term in
\eqref{eq:Jreg}, i.e., the second one on the RHS, involves some
additional information which needs to be specified a priori, namely,
the choice of the reference relation $\bar{\theta}(T)$ and the space
$\Y(\I)$. As regards the reference function $\bar{\theta}(T)$, one
natural possibility is to consider a constant value corresponding to a
constant material property, and this is the solution we will adopt
below. We recall here that $\theta(T)$ is in fact a ``slack'' variable
and is related to the actual constitutive relation via
\eqref{eq:slack}. As regards the choice of the space $\Y(\I)$, we will
follow the discussion in \cite{bvp10} and consider a regularization
term involving derivatives, namely $\Y(\I) = \dot{H}^1(\I)$, where
$\dot{H}^1(\I)$ denotes the Sobolev space equipped with the semi--norm
$\| z \|_{\dot{H}^1(\I)} \triangleq \int_{T_{\alpha}}^{T_{\beta}}
\left(\Dpartial{z}{s}\right)^2 \, ds$, $\forall_{z \in H^1(\I)}$; the
regularization term in \eqref{eq:Jreg} then becomes
\begin{equation}
\dfrac{\lambda}{2} \big\|\theta - \bar{\theta} \big\|^2_{\dot{H}^1(\I)} =
\dfrac{\lambda}{2} \int_{T_{\alpha}}^{T_{\beta}} \left( \dfrac{d\theta}{ds} -
\dfrac{d \bar \theta}{ds}  \right)^2 ds
\label{eq:regH1}
\end{equation}
yielding the following $L_2$ gradient of the regularized cost functional
\begin{equation}
\begin{aligned}
\bnabla_{\theta}^{L_2} \J_{\lambda}(s) = &- 2\int_0^{t_f} \int_{\Omega} \delta (T(\x) - s)
\, \theta(s) \, [\bnabla \u + (\bnabla \u)^T] : \bnabla \u^* \ d\x \, d\tau \\ &+
\lambda \left\{ \dfrac{d\theta}{ds} \left[\delta(s - T_{\beta}) - \delta(s - T_{\alpha})\right]
- \dfrac{d^2\theta}{ds^2} \right\}.
\end{aligned}
\label{eq:DJregH1}
\end{equation}
We remark that in obtaining \eqref{eq:DJregH1} integration by parts
was applied to the directional derivative of the regularization term.
Expression \eqref{eq:DJregH1} can now be used to compute the Sobolev
gradients as discussed in Section \ref{sec:grad_adj}. We add that
penalty term \eqref{eq:regH1} is defined on the identifiability
interval $\I$ which is contained in the interval $\L$ on which the
Sobolev gradients are computed. Computational tests illustrating the
performance of the Tikhonov regularization on a problem with noisy
data will be presented in Section \ref{sec:noise}.  In that Section we
will also briefly analyze the effect of the regularization parameter
$\lambda$. We add that the stability and convergence of Tikhonov
regularization using the Sobolev norm $H^1$ in the regularization term
and applied to an inverse problem with similar mathematical structure,
but formulated for a simpler PDE than \eqref{eq:coupled_PDEs}, was
established rigorously in \cite{k03}.

\section{Numerical Approaches to Gradient Evaluation}
\label{sec:grad_eval}

Without loss of generality, hereafter we will focus our discussion
on the 2D case. For some technical reasons we will also assume that
\begin{equation}
  \forall_{t \in [0,t_f]} \quad \textrm{meas} \left\{ \x \in \Omega, \
  | \bnabla T(t,\x) | = 0 \right\} = 0,
\label{eq:dt0}
\end{equation}
i.e., that the temperature gradient may not vanish on subregions with
finite area. (This assumption is naturally satisfied when the
temperature evolution is governed by an equation of the parabolic type
such as \eqref{eq:coupled_heat}.)

A key element of reconstruction algorithm \eqref{eq:desc} is
evaluation of the cost functional gradients given, in the $L_2$
case, by expression \eqref{eq:grad_T_coupled}. The difficulty consists
in the fact that at every instant of time $t \in [0,t_f]$ and for every value of $s$ (i.e., the dependent
variable), one has to compute a line integral defined on the level set
\begin{equation}
\Gamma_s(t) \triangleq \{ \x \in \Omega, \ T(t,\x) = s \}
\label{eq:level_set_def}
\end{equation}
of the temperature field $T(t,\x)$. To focus attention on the main
issue, the time dependence will be omitted in the discussion below.
The integrand expression in such integrals is given in terms of
solutions of the direct and adjoint problems
\eqref{eq:coupled_PDEs}--\eqref{eq:coupled_BC} and
\eqref{eq:adjoint_coupled}--\eqref{eq:BC_adjoint} which are
approximated on a grid.  As will be shown below, this problem is
closely related to approximation of one--dimensional Dirac measures in
$\RR^d$, an issue which has received some attention in the literature
\cite{EngTornTs04,ZahTorn10,Mayo84,Smereka2006,Beale08,MinGib07,MinGib08,Tow09,Tow07}.
We will compare different computational approaches to this problem,
and in order to better assess their accuracy, we will first test them
on the generic expression
\begin{equation}
\label{eq:grad_T_simple_LHS}
f(s) = \int_{\Omega} \delta (\phi(s,\x)) g(\x) \, d\x
\end{equation}
for which the actual formula for the cost functional gradient
\eqref{eq:grad_T_coupled} is a special case (except for the
time integration). In \eqref{eq:grad_T_simple_LHS} the function
$\phi(s,\x) : \RR \times \Omega \rightarrow \RR$
represents the field whose $s$-level sets define the contours of
integration $\Gamma_s$, whereas the function $g(\x) : \Omega
\rightarrow \RR$ represents the actual integrand expression. We note
that by setting $\phi(s,\x) = T(\x) - s$, $g(\x) = 2\theta(T(\x))
[\bnabla \u(\x) + (\bnabla \u(\x))^T] : \bnabla \u^*(\x)$ and adding
time integration in \eqref{eq:grad_T_simple_LHS}, we recover the
original expression \eqref{eq:grad_T_coupled} for the cost functional
gradient. We emphasize, however, that the advantage of using
\eqref{eq:grad_T_simple_LHS} with some simple, closed--form
expressions for $\phi(s,\x)$ and $g(\x)$ as a testbed is that
this will make our assessment of the accuracy of the proposed methods
independent of the accuracy involved in the numerical solution of the
governing and adjoint PDEs (needed to approximate $\u$, $T$ and
$\u^*$).

In anticipation of one of the proposed numerical approaches,
it is useful to rewrite \eqref{eq:grad_T_simple_LHS} explicitly as a
line integral
\begin{equation}
\label{eq:grad_T_simple_RHS}
f(s) = \int_{\Gamma_s} \dfrac{g(\x)}{| \bnabla \phi|}\, d\sigma
\end{equation}
which is valid provided $| \bnabla \phi| \neq 0$ for every $\x \in
\Gamma_s$, cf.~assumption \eqref{eq:dt0} (proof of the equivalence of
expressions \eqref{eq:grad_T_simple_LHS} and
\eqref{eq:grad_T_simple_RHS} may be found, for example, in
\cite{b11}).  Formula \eqref{eq:grad_T_simple_RHS} makes it clear that
for a fixed value of $s$ expression \eqref{eq:grad_T_coupled} for the
cost functional gradient can be interpreted as a sum of line integrals
defined on the instantaneous $s$--level sets of the temperature field
$T(t,\x)$.

The problem of accurate numerical evaluation of the expressions given
by either \eqref{eq:grad_T_simple_LHS} or \eqref{eq:grad_T_simple_RHS}
has received much attention, especially since the invention of the
level-set approach by Osher and Sethian \cite{OshSeth88}.
Traditionally, the problem of integration over codimension--1
manifolds defined by a level--set function $\phi(\x)$ is studied in
terms of the numerical evaluation of either the left--hand side (LHS)
or right--hand side (RHS) expression in the following relation,
analogous to \eqref{eq:grad_T_simple_LHS}--\eqref{eq:grad_T_simple_RHS},
\begin{equation}
\label{eq:line_area_int}
\int_{\Gamma: \ \phi(\x) = 0} h(\x) d\sigma =
\int_{\Omega} \delta (\phi(\x)) |\bnabla \phi(\x) | h(\x) \, d\x,
\end{equation}
where, by absorbing the factor $|\bnabla \phi(\x) |^{-1}$ into the
definition of the function $h \; : \; \Omega \rightarrow \RR$, one
bypasses the problem of the points $\x \in \Gamma_s$ where
$|\bnabla \phi(\x) | = 0$, cf.~\eqref{eq:dt0}. These approaches fall
into two main groups: \smallskip \\
\hspace*{0.5cm}
\Bmp{\textwidth}
\begin{itemize}
  \item[(A)] reduction to a {\it line (contour) integral},
  cf.~\eqref{eq:grad_T_simple_RHS}, or the LHS of
  \eqref{eq:line_area_int}, and
\smallskip

  \item[(B)] evaluation of an {\it area integral},
  cf.~\eqref{eq:grad_T_simple_LHS}, or the RHS of
  \eqref{eq:line_area_int}.
\end{itemize}
\Emp
\smallskip \\
In the context of this classification, the methods of geometric
integration developed by Min and Gibou \cite{MinGib07, MinGib08} fall
into the first category. This approach is based on decomposing the
domain $\Omega$ into simplices, which in the simplest 2D case can be
achieved via a standard triangulation, and then approximating the
level sets given by $\phi(\x) = 0$ with piecewise splines inside each
simplex. Expression \eqref{eq:grad_T_simple_RHS} then breaks up into a
number of definite integrals which can be evaluated using standard
quadratures.

In practice, however, area integration techniques (B) seem to have
become more popular. One family of such techniques relies on
regularization $\delta_{\epsilon}$ of the Dirac delta function with a
suitable choice of the regularization parameter $\epsilon$ which
characterizes the size of the support. While in the simplest case in
which the parameter $\epsilon$ is determined based on the mesh size
the error is $\O(1)$ \cite{EngTornTs04}, recently developed approaches
\cite{EngTornTs04,ZahTorn10} achieve better accuracy by adjusting
$\epsilon$ based on the local gradient $| \bnabla \phi |$ of the
level--set function. Another family of area integration approaches is
represented by the work of Mayo \cite{Mayo84} further developed by
Smereka \cite{Smereka2006} where a discrete approximation $\tilde
\delta$ of the Dirac delta function was obtained. This approach can
also be regarded as yet another way to regularize delta function
$\delta (\phi (\x))$ using a fixed compact support in the
one--dimensional (1D) space of values $\phi (\x)$.  In the second
group of approaches we also mention consistent approximations to delta
function obtained by Towers in \cite{Tow09, Tow07} using the
level--set function and its gradient computed via finite differences.

In our present reconstruction problem, we have to evaluate the
gradient expression \eqref{eq:grad_T_coupled} for the whole range of
$T \in \L$, hence the discretization of the interval $\L$ will also
affect the overall accuracy of the reconstruction, in addition to the
accuracy characterizing evaluation of the gradient for a particular
value of $T$. This is an aspect of the present problem which is
outside the scope of earlier investigations concerning evaluation of
the contour integrals of grid--based data
\cite{EngTornTs04,ZahTorn10,Mayo84,Smereka2006,Beale08,MinGib07,MinGib08,Tow09,Tow07}.
Thus, we need to understand how the interplay of the discretizations
of the physical space $\Omega$ with the step size $h$ and the state
space $\L$ with the step size $h_T$ affects the accuracy of the
reconstruction. In principle, one could also consider the effect of
discretizing the time interval $[0,t_f]$, however, the corresponding
step size is linked to $h$ via the CFL condition, hence this effect
will not be separately analyzed here. There are also questions
concerning the computational complexity of the different approaches.
We will consider below the following three methods to evaluate
expression \eqref{eq:grad_T_simple_LHS}, or equivalently
\eqref{eq:grad_T_simple_RHS}, which are representative of the
different approaches mentioned above
\begin{enumerate}
\item[\#1]
line integration over approximate level sets which is a method from
group A based on a simplified version of the geometric
integration developed by Min and Gibou in \cite{MinGib07, MinGib08},

\item[\#2]
approximation of Dirac delta measures developed by Smereka in
\cite{Smereka2006} which is an example of a  regularization
technique and utilizes the area integration strategy from group B, and

\item[\#3]
approximation of contour integrals with area integrals, a method which
also belongs to group B and combines some properties of
regularization and discretization of Dirac delta measures discussed
above \cite{ZahTorn10, EngTornTs04,Smereka2006}; more details about
this approach, including an analysis for its accuracy, are provided in
Section \ref{sec:domain_int}.
\end{enumerate}

To fix attention, we now introduce two different finite--element (FEM)
discretizations of the domain $\Omega$ based on
\begin{itemize}
  \item triangular elements $\Omega^{\vartriangle}_i$, $i=1,\dots,N_{\vartriangle}$, such that
  \begin{equation}
    \Omega = \bigcup^{N_{\vartriangle}}_{i=1} \Omega^{\vartriangle}_i, \ \ \textrm{and}
    \label{eq:discr_trg}
  \end{equation}
  \item quadrilateral elements $\Omega^{\square}_i$, $i=1,\dots,N_{\square}$, such that
  \begin{equation}
    \Omega = \bigcup^{N_{\square}}_{i=1} \Omega^{\square}_i,
    \label{eq:discr_qdr}
  \end{equation}
\end{itemize}
where $N_{\vartriangle}$ and $N_{\square}$ are the total numbers of
the elements for each type of discretization (in case of uniform
triangulation one has $N_{\vartriangle} = 2N_{\square}$).  In our
computational tests we will assume that the functions $\phi(s,\x)$ and
$g(\x)$ are given either analytically, or in terms of the following
FEM representations
  \begin{equation}
    \phi(s,\x)|_{\Omega^{\vartriangle}_i}  = \Sigma_{k=1}^3 \phi_k^i \psi_k^i (\x),
    \qquad g(\x)|_{\Omega^{\vartriangle}_i} = \Sigma_{k=1}^3 g_k^i \psi_k^i (\x),
    \qquad i=1,\dots,N_{\vartriangle},
    \label{eq:ffem_trg}
  \end{equation}
  \begin{equation}
    \phi(s,\x)|_{\Omega^{\square}_i}  = \Sigma_{k=1}^4 \phi_k^i \psi_k^i (\x),
    \qquad g(\x)|_{\Omega^{\square}_i} = \Sigma_{k=1}^4 g_k^i \psi_k^i (\x),
    \qquad i=1,\dots,N_{\square},
    \label{eq:ffem_qdr}
  \end{equation}
where $\phi_k^i$ and $g_k^i$ are the given nodal values of the functions $\phi(s,\x)$
and $g(\x)$, whereas $\psi_k^i(\x)$ are the basis functions (linear in
\eqref{eq:ffem_trg} and bilinear in \eqref{eq:ffem_qdr} \cite{g06}).
We also discretize the reconstruction interval (solution space) $\L = [T_a, T_b]$
with the step size $h_T$ as follows
\begin{equation}
T_i = T_a + i \, h_T, \ \ i = 0, \dots, N_T, \ \ h_T = \dfrac{T_b - T_a}{N_T}.
\label{eq:T_discr}
\end{equation}

\subsection{Line Integration Over Approximate Level Sets}
\label{sec:line_int}

This approach is a variation of the geometric integration technique
developed by Min and Gibou \cite{MinGib07,MinGib08}. The main idea
behind both methods is decomposition of the domain $\Omega$ into
simplices, which in our simplest 2D case is represented by
triangulation \eqref{eq:discr_trg}, and then approximating the level
sets given by $\phi(s,\x) = 0$ with piecewise linear splines inside
each simplex (triangle).  While in the geometric integration approach
of Min and Gibou one uses linear interpolation to refine locally the
finite elements which contain the level sets $\phi(s,\x) = 0$ and then
the second--order midpoint rule for approximating line integrals over
the selected simplices, in the present method we employ analogous
approximations of the level sets, but without local refinement, to
reduce line integral \eqref{eq:grad_T_simple_RHS} to a 1D definite
integral which is then evaluated using standard quadratures.

The starting point for this approach is formula
\eqref{eq:grad_T_simple_RHS}.  For a fixed value of $s$ the
corresponding level set can be described as
\begin{equation}
\Gamma_s = \bigcup^{M(s)}_{j=1} \Gamma^j_s,
\label{eq:level_set_discr}
\end{equation}
where $\Gamma^j_s \subset \Omega_j^{\vartriangle}$ and $M(s)$ is the
total number of the finite elements containing segments of the level
set $\Gamma_s$. We thus need to approximate $\int_{\Gamma^j_s}
\dfrac{g(\x)}{| \bnabla \phi(s, \x)|}\, d\sigma$, i.e., the line
integral over the part of the level--set curve contained in the $j$-th
finite element $\Omega^{\vartriangle}_j$. In view of
\eqref{eq:ffem_trg}, the integrand expression $\varrho(s, \x)
\triangleq \dfrac{g(\x)}{| \bnabla \phi(s, \x)|}$ can be approximated
as
\begin{equation*}
\tilde \varrho (s, \x)|_{\Omega^{\vartriangle}_i} \cong \Sigma_{k=1}^3 \varrho_k^i \psi_k^i (\x),
\end{equation*}
where $\varrho_k^i$ are the known nodal values of the function
$\varrho (s, \x)$.  An approximation $\tilde \Gamma_s^j$ of the part
of the level set $\Gamma_s^j$ belonging to the $j$-th finite element
can be obtained in an explicit form $y = y(x)$, $x \in [x',x'']$, or a
parametric form $x = x(t)$, $y = y(t)$ with $t \in [t',t'']$, based on
representation \eqref{eq:ffem_trg} of the level--set function $\phi(s,
\x)$. This leads to the following two possible reductions of the line
integral to a definite integral
\begin{subequations}
\label{eq:int}
\begin{align}
& \int_{\tilde \Gamma_s^j} \tilde \varrho(s, \x) d \sigma =
\int_{x'}^{x''} \tilde \varrho(x,y(x)) \sqrt{\left( \frac{dy}{dx} \right) ^2 + 1} \, dx
\label{eq:int_1} \\
& \int_{\tilde \Gamma_s^j} \tilde \varrho(s, \x) d \sigma =
\int_{t'}^{t''} \tilde \varrho(x(t),y(t))
\sqrt{\left( \frac{dx}{dt} \right)^2 + \left( \frac{dy}{dt} \right)^2} \, dt
\label{eq:int_2}
\end{align}
\end{subequations}
which can be evaluated using standard quadratures for 1D definite
integrals. We then have
\begin{equation}
f(s)  \approx \Sigma_{j=1}^{M(s)} \int_{\tilde \Gamma_s^j} \tilde \varrho(s, \x)\, d \sigma.
\label{eq:FEM_line_int}
\end{equation}
We note that the accuracy of this approach is mainly determined by the
order of interpolation used to represent the level set $\tilde
\Gamma^j_s$ and the integrand expression $\tilde \varrho(s, \x)$ which
depend on the type of the finite elements used \cite{g06}. (The error
of the quadrature employed to evaluate \eqref{eq:int} does not have a
dominating effect.) As was mentioned in \cite{MinGib07}, the use of
triangulation \eqref{eq:discr_trg} together with linear interpolation
of $\phi(s,\x)$ and $\varrho(s, \x)$ results in the overall
second--order accuracy of this method.

\subsection{Approximation of Dirac Delta Measures}
\label{sec:Smereka_int}

This approach has formula \eqref{eq:grad_T_simple_LHS} for its
starting point and relies on a discrete approximation of the Dirac
delta function obtained by Smereka in \cite{Smereka2006}. It is
derived via truncation of the discrete Laplacian of the corresponding
Green's function. Suppose the domain $\Omega$ is covered with a
uniform Cartesian grid corresponding to \eqref{eq:discr_qdr}, i.e.,
with nodes $x_i = x_0 + ih$, $y_j = y_0 + jh$, where $i, j$ are
integer indices, $x_0, y_0 \in \RR$ and $h$ is the step size. The
first--order accurate approximation of the discrete Dirac delta
function at the node $(x_i,y_j)$ is
\begin{equation}
\tilde \delta (\phi_{i,j}) = \tilde \delta_{i,j}^{(+x)} +
\tilde \delta_{i,j}^{(-x)} + \tilde \delta_{i,j}^{(+y)} +
\tilde \delta_{i,j}^{(-y)},
\label{eq:delta_Smereka}
\end{equation}
where
\begin{align*}
& \tilde \delta_{i,j}^{(+x)} \triangleq \left\{
\begin{aligned}
\dfrac{|\phi_{i+1,j}D^0_x \phi_{i,j}|}{h^2 |D^+_x \phi_{i,j}| |\bnabla^{\epsilon}_0 \phi_{i,j}|}
\qquad & \textrm{if} \ \phi_{i,j} \phi_{i+1,j} \leq 0, \\
0, \qquad & \textrm{otherwise},
\end{aligned}
 \right. \\
& \tilde \delta_{i,j}^{(-x)} \triangleq \left\{
\begin{aligned}
\dfrac{|\phi_{i-1,j}D^0_x \phi_{i,j}|}{h^2 |D^-_x \phi_{i,j}| |\bnabla^{\epsilon}_0 \phi_{i,j}|}
\qquad & \textrm{if} \ \phi_{i,j} \phi_{i-1,j} < 0, \\
0, \qquad & \textrm{otherwise},
\end{aligned}
 \right. \\
& \tilde \delta_{i,j}^{(+y)} \triangleq \left\{
\begin{aligned}
\dfrac{|\phi_{i,j+1}D^0_y \phi_{i,j}|}{h^2 |D^+_y \phi_{i,j}| |\bnabla^{\epsilon}_0 \phi_{i,j}|}
\qquad & \textrm{if} \ \phi_{i,j} \phi_{i,j+1} \leq 0, \\
0, \qquad & \textrm{otherwise},
\end{aligned}
 \right. \\
& \tilde \delta_{i,j}^{(-y)} \triangleq \left\{
\begin{aligned}
\dfrac{|\phi_{i,j-1}D^0_y \phi_{i,j}|}{h^2 |D^-_y \phi_{i,j}| |\bnabla^{\epsilon}_0 \phi_{i,j}|}
\qquad & \textrm{if} \ \phi_{i,j} \phi_{i,j-1} < 0, \\
0, \qquad & \textrm{otherwise},
\end{aligned}
 \right.
\end{align*}
where for the discretized level--set function $\phi_{i,j} \triangleq
\phi(s, x_i,y_j)$ we have the following definitions
\begin{equation*}
D^+_x \phi_{i,j} \triangleq \dfrac{\phi_{i+1,j}-\phi_{i,j}}{h}, \quad
D^-_x \phi_{i,j} \triangleq \dfrac{\phi_{i,j}-\phi_{i-1,j}}{h}, \quad
D^0_x \phi_{i,j} \triangleq \dfrac{\phi_{i+1,j}-\phi_{i-1,j}}{2h}
\end{equation*}
and
\begin{equation*}
| \bnabla ^{\epsilon}_0 \phi_{i,j} | \triangleq \sqrt{ (D^0_x \phi_{i,j})^2 +
(D^0_y \phi_{i,j})^2 + \varepsilon},
\end{equation*}
in which $\varepsilon \ll 1$ is used for regularization
\cite{Smereka2006}. The expressions $D^+_y \phi_{i,j}$, $D^-_y
\phi_{i,j}$, $D^0_y \phi_{i,j}$ are defined analogously. Using the
definition of the discrete delta function from
\eqref{eq:delta_Smereka}, the value $f(s)$ in
\eqref{eq:grad_T_simple_LHS} can be thus approximated in the following
way
\begin{equation}
f(s) \approx h^2 \sum_{i,j} \tilde \delta_{i,j} \, g_{i,j},
\label{eq:FEM_Delta}
\end{equation}
where $g_{i,j}$ are the nodal values of the function $g(\x)$.
We note that this method was validated in \cite{Smereka2006}
exhibiting the theoretically predicted first order
of accuracy only in cases in which the level sets $\Gamma_s$ do not
intersect the domain boundary $\partial \Omega$, a situation which may
occur in the present reconstruction problem.

\subsection{Approximation of Contour Integrals with Area Integrals}
\label{sec:domain_int}

Our third method, in which the level--set integral
\eqref{eq:grad_T_simple_LHS} is approximated with an area integral
defined over a region containing the level set $\Gamma_s$,
cf.~\eqref{eq:level_set_def}, appears to be a new approach and will be
presented in some detail here. It consists of the following three
steps
\begin{enumerate}
\item for a fixed value of the state variable $s = T_i$ we define the interval
  $[T_{i-\frac{1}{2}}, T_{i+\frac{1}{2}}] = [s - \frac{1}{2}h_T, s +
  \frac{1}{2}h_T] \subset \L$; then, we have
\begin{equation}
f(s) \approx \frac{1}{h_T} \int_{s - \frac{1}{2}h_T}^{s + \frac{1}{2}h_T}
f(\zeta)\,d\zeta,
\label{eq:fs}
\end{equation}

\item now we define a subdomain $\Omega_{s,h_T} \subset \Omega$
  which contains all the points of $\Omega$ that lie between the two
  level--set curves $\Gamma_{s-\frac{1}{2}h_T}$ and
  $\Gamma_{s+\frac{1}{2}h_T}$
  \begin{equation}
    \Omega_{s,h_T} \triangleq \left\{ \x \in \Omega, \, T(\x) \in
    \left[s - \frac{1}{2}h_T, s + \frac{1}{2}h_T\right] \right\},
    \label{eq:Omega_cont}
  \end{equation}
  see Figure \ref{fig:Dom_Int}a; we then approximate $\Omega_{s,h_T}$
  with the region
  \begin{equation}
    \tilde \Omega_{s,h_T} \triangleq \bigcup^{N_{s,h_T}}_{j=1}
    \Omega^{\square}_{s,h_T;j}, \quad \textrm{where} \ \ \Omega^{\square}_{s,h_T;j} = \left\{ \Omega^{\square}_{j} \; : \;  \x^0_j \in
    \Omega^{\square}_{j} \ \textrm{and} \ T(\x^0_j) \in \left[s - \frac{1}{2}h_T, s + \frac{1}{2}h_T \right] \right\},
    \label{eq:Omega_FEM}
  \end{equation}
  see Figure \ref{fig:Dom_Int}b, which consists of the quadrilateral
  finite elements $\Omega^{\square}_{s,h_T;j}$, $j = 1, \dots, N_{s,h_T}$, with
  the center points $\x^0_j$ satisfying the condition $T(\x^0_j) \in
  [s - (1/2) h_T, s + (1/2) h_T]$,

\item in view of \eqref{eq:fs}, expression \eqref{eq:grad_T_simple_LHS} is
  approximated with an area integral over the region contained between
  the level--set curves $\Gamma_{s-\frac{1}{2}h_T}$ and
  $\Gamma_{s+\frac{1}{2}h_T}$, which is in turn approximated by the
  FEM region $\tilde \Omega_{s,h_T}$ given by \eqref{eq:Omega_FEM};
  finally, the integral over this region is approximated using the
  standard 2D compound midpoint rule as
    \begin{equation}
      \label{eq:FEM_DomainInt}
      f(s)  \approx \frac{1}{h_T} \int_{s - \frac{1}{2}h_T}^{s + \frac{1}{2}h_T}
      \int_{\Omega} \delta (s(\x) - \zeta) g(\x) \, d\x \, d\zeta \approx
      \dfrac{h^2}{h_T} \sum_{j=1}^{N_{s,h_T}} g(\x^0_j).
    \end{equation}
\end{enumerate}

As regards the accuracy of this approach, we have the following
\begin{theorem}
  Formula \eqref{eq:FEM_DomainInt} is second order accurate with
  respect to the discretization of the space domain and first order
  accurate with respect to the discretization of the state domain,
  i.e.,
  \begin{equation}
    \label{eq:direct_frm_acc}
    f(s) = \int_{\Omega} \delta (T(\x) - s) g(\x) \ d\x =
    \dfrac{h^2}{h_T} \sum_{j=1}^{N_{s,h_T}} g(\x^0_j) + \O(h^2) + \O(h_T).
  \end{equation}
\label{thm:m3}
\end{theorem}

\begin{proof}
  We start by integrating both sides of \eqref{eq:grad_T_simple_LHS} over the
  interval $[T_{i-\frac{1}{2}}; T_{i+\frac{1}{2}}] = [s - \frac{1}{2}h_T;
  s + \frac{1}{2}h_T]$ obtaining
  \begin{equation}
    \int_{s - \frac{1}{2}h_T}^{s + \frac{1}{2}h_T} f(\zeta) \ d\zeta =
    \int_{\Omega} \left[ \int_{s - \frac{1}{2}h_T}^{s + \frac{1}{2}h_T}
    \delta (T(\x) - \zeta) \ d\zeta \right] g(\x) \ d\x =
    \int_{\Omega} \chi_{[s - \frac{1}{2}h_T, s + \frac{1}{2}h_T]}
    (T(\x)) \ g(\x) \ d\x,
    \label{eq:dom_int_int}
  \end{equation}
  where the characteristic function
  \begin{equation*}
    \chi_{[T', T'']} (T(\x)) = \left\{
    \begin{alignedat}{2}
    &1\quad & & \textrm{for} \ T(\x) \in [T', T''], \\
    &0& & \textrm{for} \ T(\x) \notin [T', T'']
    \end{alignedat}
    \right.
  \end{equation*}
  describes the subdomain $\Omega_{s, h_T}$ introduced earlier in
  \eqref{eq:Omega_cont}, cf.~Figure \ref{fig:Dom_Int}a.
  Now, using a second--order accurate midpoint rule for 1D integration, we
  can express the LHS of \eqref{eq:dom_int_int} as
  \begin{equation}
    \int_{s - \frac{1}{2}h_T}^{s + \frac{1}{2}h_T} f(\zeta) \ d\zeta =
    f \left( \dfrac{T_{i-\frac{1}{2}}+T_{i+\frac{1}{2}}}{2} \right) \cdot h_T
    + \O(h_T^2) = f(s) h_T + \O(h_T^2).
    \label{eq:dom_int_T}
  \end{equation}
  Approximation of the RHS in \eqref{eq:dom_int_int} takes place in two steps.
  In the first step we approximate the actual integration domain $\Omega_{s,h_T}$
  with the union of the finite elements $\tilde \Omega_{s,h_T}$,
  cf.~\eqref{eq:Omega_FEM}. In order to estimate the error
  \begin{equation*}
    E_1 \triangleq \left| \int_{\Omega_{s,h_T}} g(\x) \ d\x -
    \int_{\tilde \Omega_{s,h_T}} g(\x) \ d\x \right|
  \end{equation*}
  of this step, we divide the set of cells $\tilde \Omega_{s, h_T}$ into two
  subsets, see Figure \ref{fig:Dom_Int}b,
  \begin{equation*}
    \tilde \Omega_{s, h_T} = \tilde \Omega^*_{s, h_T} \cup \tilde \Omega'_{s, h_T},
  \end{equation*}
  where $\tilde \Omega^*_{s, h_T}$ consists of the cells with all 4 vertices
  $\{\x_k\}_{k=1}^4$ satisfying the condition $T(\x_k) \in
  [s - \frac{1}{2}h_T, s + \frac{1}{2}h_T]$. The subregion $\tilde \Omega'_{s, h_T}$,
  defined as the compliment of $\tilde \Omega^*_{s, h_T}$ in $\tilde \Omega_{s, h_T}$,
  represents the union of ``truncated'' cells, i.e., cells which have at least one node
  outside $\Omega_{s, h_T}$. This subregion is in turn further subdivided into two subsets,
  i.e.,
  \begin{equation*}
    \tilde \Omega'_{s, h_T} = \tilde \Omega'_{s, h_T; in} \cup \tilde \Omega'_{s, h_T; out},
  \end{equation*}
  where
  \begin{equation*}
    \tilde \Omega'_{s,h_T;in} \triangleq \left\{ \x \in \tilde \Omega'_{s,h_T}, \, T(\x) \in
    \left[s - \frac{1}{2}h_T, s + \frac{1}{2}h_T\right] \right\},
  \end{equation*}
  \begin{equation*}
    \tilde \Omega'_{s,h_T;out} \triangleq \left\{ \x \in \tilde \Omega'_{s,h_T}, \, T(\x) \notin
    \left[s - \frac{1}{2}h_T, s + \frac{1}{2}h_T\right] \right\}.
  \end{equation*}
  We have to define one more set $\tilde \Omega''_{s, h_T}$ which consists of the cells with
  at least one vertex $\{\x_k\}_{k=1}^4$ satisfying the condition $T(\x_k) \in
  [s - \frac{1}{2}h_T, s + \frac{1}{2}h_T]$, but whose center points $\x^0_j$ lie outside
  $\Omega_{s,h_T}$. We also further subdivide this set into two subsets
  \begin{equation*}
    \tilde \Omega''_{s, h_T} = \tilde \Omega''_{s, h_T; in} \cup \tilde \Omega''_{s, h_T; out},
  \end{equation*}
  where
  \begin{equation*}
    \tilde \Omega''_{s,h_T;in} \triangleq \left\{ \x \in \tilde \Omega''_{s,h_T}, \, T(\x) \in
    \left[s - \frac{1}{2}h_T, s + \frac{1}{2}h_T\right] \right\},
  \end{equation*}
  \begin{equation*}
    \tilde \Omega''_{s,h_T;out} \triangleq \left\{ \x \in \tilde \Omega''_{s,h_T}, \, T(\x) \notin
    \left[s - \frac{1}{2}h_T, s + \frac{1}{2}h_T\right] \right\}.
  \end{equation*}
  We thus have
  \begin{equation}
    \int_{\Omega_{s,h_T}} g(\x) \ d\x = \int_{\tilde \Omega^*_{s,h_T}} g(\x) \ d\x +
    \int_{\tilde \Omega'_{s,h_T;in}} g(\x) \ d\x + \int_{\tilde \Omega''_{s,h_T;in}} g(\x) \ d\x,
    \label{eq:int_exact}
  \end{equation}
  \begin{equation}
    \begin{aligned}
    \int_{\tilde \Omega_{s,h_T}} g(\x) \ d\x &= \int_{\tilde \Omega^*_{s,h_T}} g(\x) \ d\x +
    \int_{\tilde \Omega'_{s,h_T}} g(\x) \ d\x \\
    &= \int_{\tilde \Omega^*_{s,h_T}} g(\x) \ d\x +
    \int_{\tilde \Omega'_{s,h_T;in}} g(\x) \ d\x + \int_{\tilde \Omega'_{s,h_T;out}} g(\x) \ d\x,
    \end{aligned}
    \label{eq:int_approx}
  \end{equation}
  so that the domain approximation error can be estimated as follows
  \begin{equation}
    \begin{aligned}
    E_1 = \left| \int_{\tilde \Omega'_{s,h_T;out}} g(\x) \ d\x -
    \int_{\tilde \Omega''_{s,h_T;in}} g(\x) \ d\x \right| &\leq
    \left| \int_{\tilde \Omega'_{s,h_T;out}} g(\x) \ d\x \right| +
    \left| \int_{\tilde \Omega''_{s,h_T;in}} g(\x) \ d\x \right| \\
    &\leq \max_{\x \in \Omega} |g(\x)| \
    \left( \left| \tilde \Omega'_{s,h_T} \right| + \left| \tilde \Omega''_{s,h_T} \right|
    \right) = \O(h^2),
    \end{aligned}
  \end{equation}
  where $| \Omega | \triangleq \textrm{meas} \ \Omega$. The second
  error in the approximation of the RHS of \eqref{eq:dom_int_int} is
  related to the accuracy of the quadrature applied to $\int_{\tilde
    \Omega_{s,h_T}} g(\x) \ d\x$ and for the 2D compound midpoint rule
  is $E_2 = \O(h^2)$, so that we obtain
  \begin{equation}
    \left| \int_{\tilde{\Omega}_{s,h_T}} g(\x) \, d\x - \sum_{j=1}^{N_{s,h_T}} g(\x^0_j) \, h^2
    \right| = \O(h^2).
    \label{eq:dom_int_x}
  \end{equation}
  Comparing \eqref{eq:FEM_DomainInt} with \eqref{eq:dom_int_T} and
  dividing both sides into $h_T$ we finally obtain
  \eqref{eq:direct_frm_acc} which completes the proof.
\end{proof}

\begin{figure}
\begin{center}
\mbox{
\subfigure[]{
\includegraphics[width=0.45\textwidth]{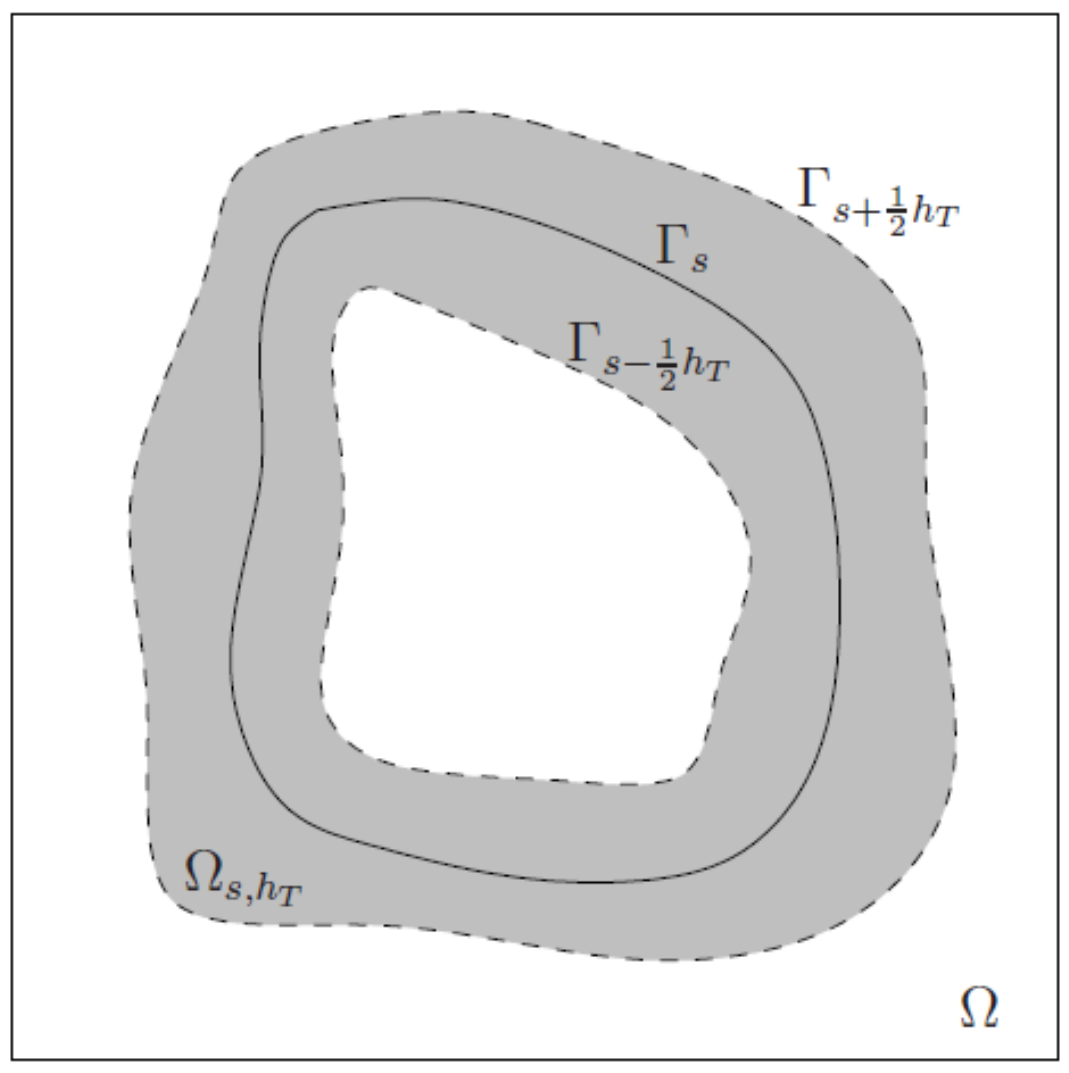}}
\subfigure[]{
\includegraphics[width=0.45\textwidth]{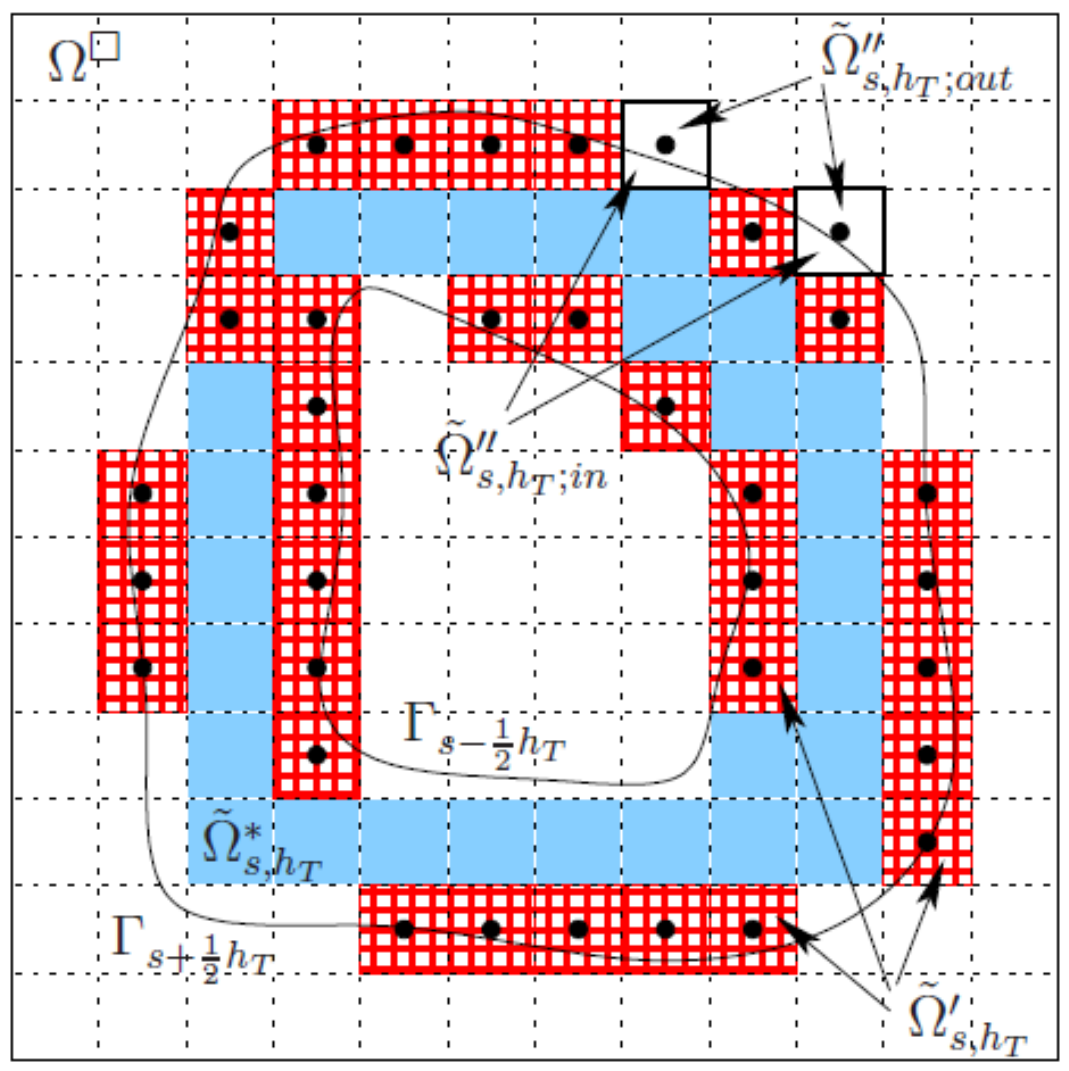}}}
\end{center}
\caption{Illustration of approach \#3 where line integral
  \eqref{eq:grad_T_simple_LHS} is approximated with an area integral
  (see Section \ref{sec:domain_int}): (a) region $\Omega_{s, h_T}$
  which lies between the two level--set curves
  $\Gamma_{s-\frac{1}{2}h_T}$ and $\Gamma_{s+\frac{1}{2}h_T}$ and (b)
  its approximation with the region $\tilde \Omega_{s, h_T} = \tilde
  \Omega^*_{s, h_T} \cup \tilde \Omega'_{s, h_T}$, where checked cells
  represent $\tilde \Omega'_{s, h_T}$ and shaded cells represent
  $\tilde \Omega^*_{s, h_T}$.  Figure (b) also shows a part of the
  region $\tilde \Omega''_{s, h_T} = \tilde \Omega''_{s, h_T; in} \cup
  \tilde \Omega''_{s, h_T; out}$ represented by 2 elements in the top
  right corner.}
\label{fig:Dom_Int}
\end{figure}
So far, we have considered the discretizations of the physical and state spaces,
$\Omega$ and $\L$, as independent. We remark that using the relationship
\begin{equation}
  \min_{\x \in \Omega} | \bnabla T(\x) | \cdot h \leq h_T \leq
  \max_{\x \in \Omega} | \bnabla T(\x) | \cdot h
\end{equation}
one could link the corresponding discretization parameters $h$ and $h_T$ to
each other.

\section{Computational Results}
\label{sec:results}

\subsection{Comparison of Different Approaches to Gradient Evaluation}
\label{sec:grads_approaches}

In this Section we discuss the accuracy and efficiency of the three
methods for evaluation of expression \eqref{eq:grad_T_simple_LHS}
presented in Sections \ref{sec:line_int}, \ref{sec:Smereka_int} and
\ref{sec:domain_int}. In order to assess their utility for the
parameter reconstruction problem studied in this work, we will
consider the following three test cases
\begin{enumerate}
\renewcommand{\labelenumi}{(\roman{enumi})}
\item single (fixed) value of $s$ with $\phi(s,\x)$ and $g(\x)$ given
  analytically,
\item parameter $s$ varying over a finite range with $\phi(s,\x)$ and $g(\x)$
  given analytically,
\item parameter $s$ varying over a finite range with $\phi(s,\x)$ and $g(\x)$
  given in terms of solutions of the direct and adjoint problem.
\end{enumerate}
Tests (ii) and (iii) with $s$ spanning the entire interval $\L$ are
particularly relevant for the present reconstruction problem, as they
help us assess the accuracy of the cost functional gradients over
their entire domains of definition, including the values of $s$ for
which the level sets $\Gamma_s$ intersect the domain boundary
$\partial \Omega$.  Results of tests (i)--(iii) are presented below.

\subsubsection{Tests for a Single Value of $s$ with $\phi(s,\x)$
and $g(\x)$ Given Analytically}
\label{sec:test_1}

Here we employ our three methods to compute numerically the value of a line
integral over the circle $x^2 + y^2 = 1$
\begin{enumerate}
\item[(a)] in domain $\Omega_1 = [-2, 2]^2$ which contains the entire curve
  \begin{equation}
    \label{eq:test1a_int}
    I_{ex,1} = \int_{x^2 + y^2 = 1} (3x^2 - y^2) \ d \sigma = 2 \pi
  \end{equation}
(this test problem is actually borrowed from \cite{Smereka2006}),

\item[(b)] and in domain $\Omega_2 = [0, 2]^2$ which contains only a part of the curve
  in the first quadrant
  \begin{equation}
    \label{eq:test1b_int}
    I_{ex,2} = \int_{x^2 + y^2 = 1, \, x,y > 0} (3x^2 - y^2) \ d \sigma = \frac{\pi}{2}.
  \end{equation}
\end{enumerate}
The main difference between test cases (a) and (b) is that while in
(a) the contour is entirely contained in the domain $\Omega_1$, it
intersects the domain boundary $\partial \Omega_2$ in case (b). As
shown in Figure \ref{fig:test1}, methods \#1 (line integration) and
\#3 (area integration) in both cases show the expected accuracy of $\O(h^2)$, where
$h = \Delta x = \Delta y = 2^{-(3+i)}$, $i = 1 \dots 6$, while method
\#2 (delta function approximation) is $\O(h^{3/2})$ accurate in case
(a) and only $\O(h^1)$ accurate in case (b). We also add that the line
integration method exhibits the smallest constant prefactor
characterizing the error.

\begin{figure}
\begin{center}
\mbox{
\subfigure[]{\includegraphics[width=0.5\textwidth]{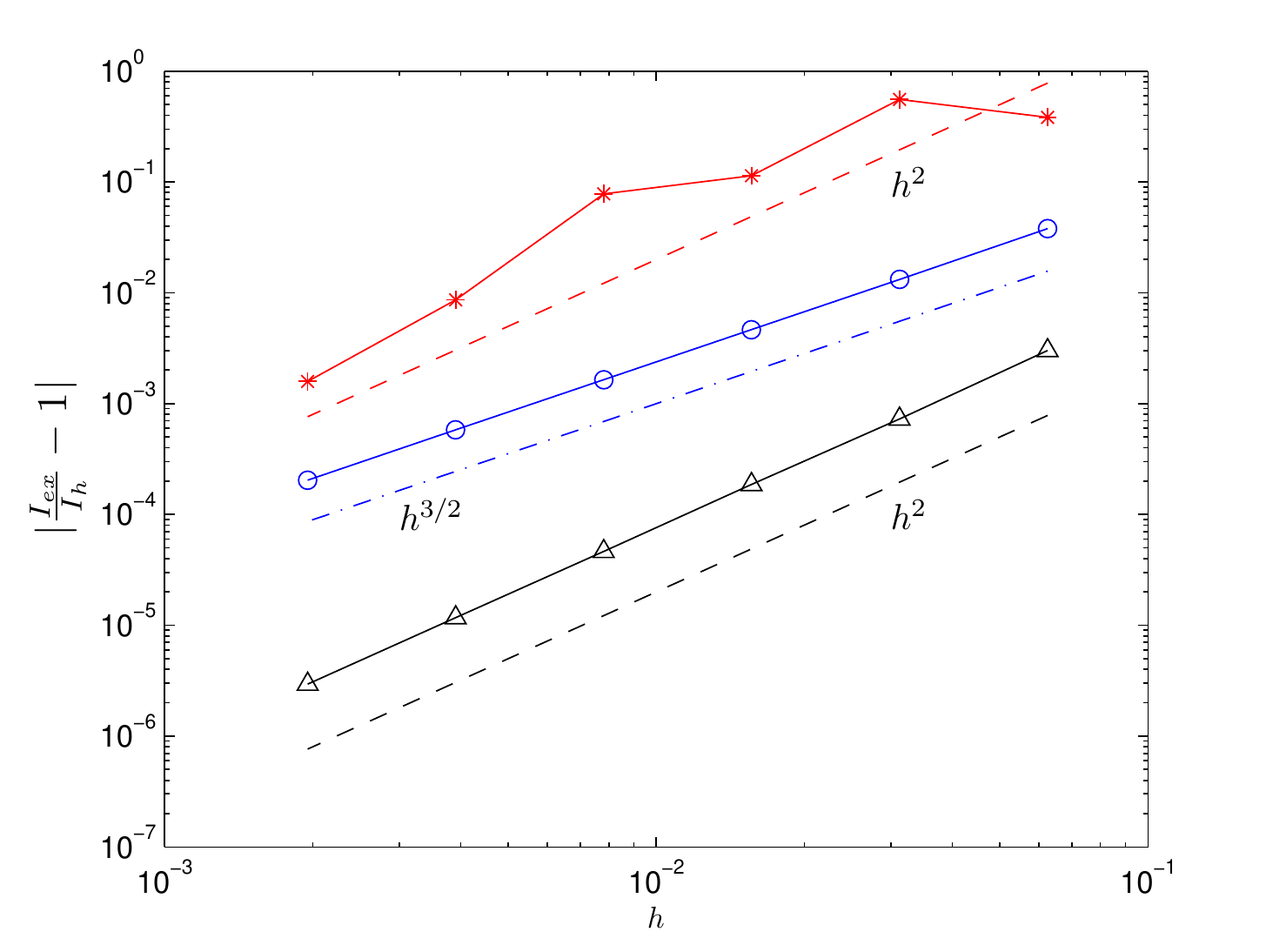}} \quad
\subfigure[]{\includegraphics[width=0.5\textwidth]{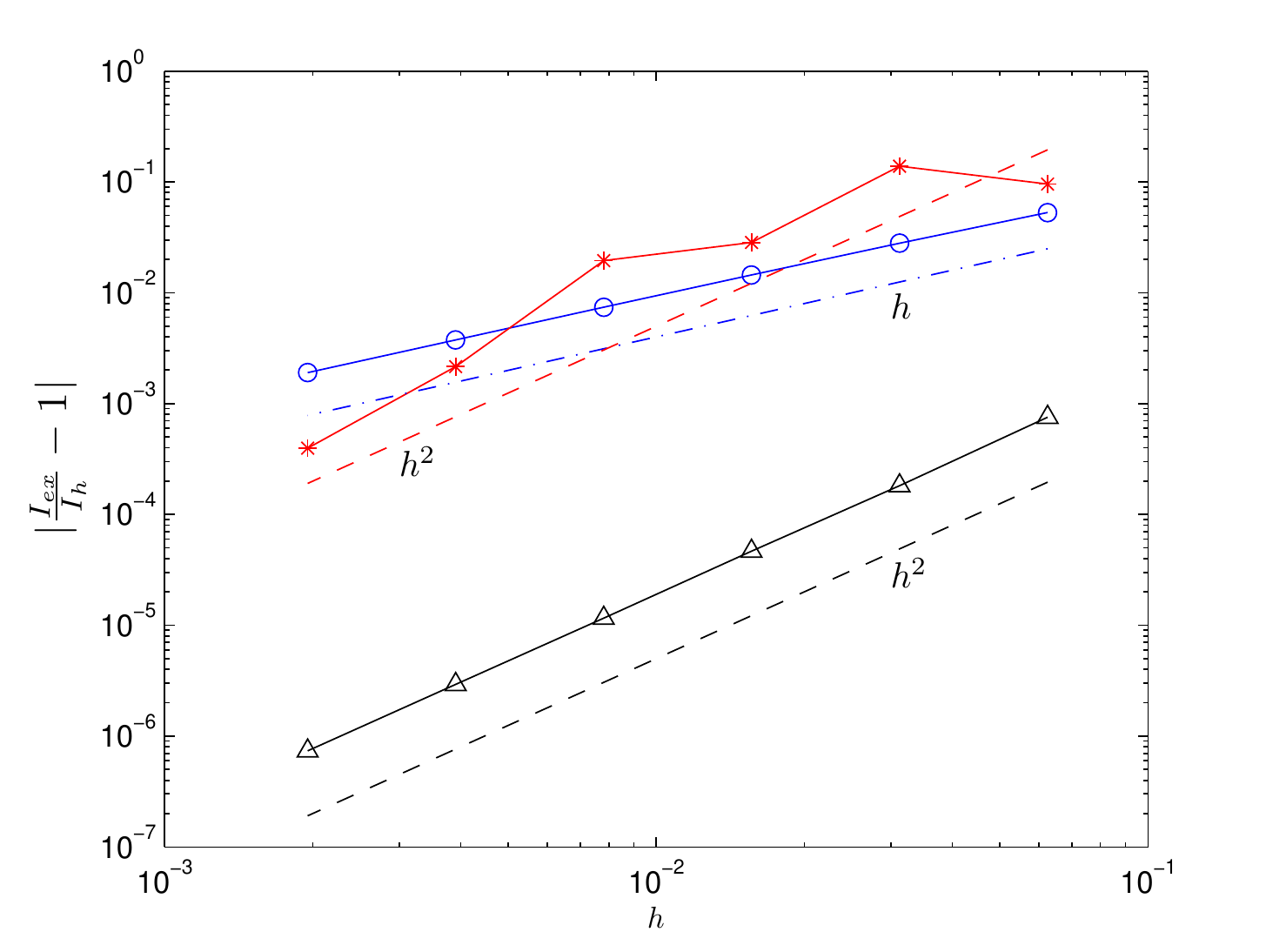}}}
\caption{Relative error $| \frac{I_{ex,i}}{I_h} - 1 |$, $i=1,2$, versus
  discretization step $h = \Delta x = \Delta y$ in the numerical
  approximation $I_h$ of (a) line integral \eqref{eq:test1a_int} and (b)
  line integral \eqref{eq:test1b_int}, see Section
  \ref{sec:test_1}. Triangles represent the line integration approach
  (method \#1), circles represent the results obtained using the delta
  function approximation (method\ \#2), whereas asterisks show the
  data for the area integration approach (method \#3).}
\label{fig:test1}
\end{center}
\end{figure}

\subsubsection{Tests for $s$ Varying Over a Finite Range with $\phi(s,\x)$ and $g(\x)$
  Given Analytically}
\label{sec:test_2}

In order to analyze this case we will introduce a new diagnostic quantity.
We begin with the integral transform formula \eqref{eq:change_vars} applied
to some perturbation $\mu'(T(\x))$
\begin{equation}
\mu'(T(\x)) = \int_{-\infty}^{+\infty} \delta(\phi(s,\x)) \mu'(s) \, ds,
\label{eq:perturb_transf}
\end{equation}
where $\phi(s,\x) = T(\x) - s$. Multiplying both sides of
\eqref{eq:perturb_transf} by $g(\x)$, integrating over the domain
$\Omega$ and changing the order of integration we obtain the following
useful identity
\begin{equation}
\label{eq:test2_int}
\int_{\Omega} \mu'(T(\x)) g(\x) \, d\x =
\int_{-\infty}^{+\infty} f(s) \mu'(s) \, ds
\end{equation}
with $f(s)$ defined in \eqref{eq:grad_T_simple_LHS}, where the RHS has
the structure of the Riesz identity for the G\^ateaux differential of
the cost functional, cf.~\eqref{eq:dJ4}, whereas the LHS is a simple
area integral which can be easily evaluated using high--accuracy
quadratures. Given the formal similarity of the RHS of
\eqref{eq:test2_int} and the Riesz formula \eqref{eq:dJ4}, this test
is quite relevant for the optimization problem we are interested in
here. We will thus use our three methods to evaluate the RHS of \eqref{eq:test2_int} and
compared it to the LHS, which is evaluated with high precision on a
refined grid in $\Omega$, so that it can be considered ``exact''. Our
tests are based on the following data
\begin{itemize}
\item spatial domain $\Omega = [0,1]^2$ discretized with $h = \Delta x =
  \Delta y = 2^{-(4+i)}$, $i = 1 \dots 7$,

\item state domain $\L = [T_a, T_b]$, where $T_a = 100$, $T_b = 700$,
discretized using $N_T = 200, 1000, 10000$  points for methods \#2 and \#3 and
$N_T = 200, 2000, 20000$ points for method \#1; $h_T = \dfrac{T_b - T_a}{N_T}$,

\item $T(\x) = 100(x^2+y^2)+300$, \, $g(\x) = \cos(x)+3\sin(2y-1)$, $\x \in \Omega$,

\item perturbations used
    $\mu'_1(s) = \exp(-\frac{s}{1000})$,
    $\mu'_2(s) = \frac{10}{s^2}$ and
    $\mu'_3(s) = -\frac{s^2}{90000}+\frac{2s}{225}+\frac{2}{9}$,
    $s \in [T_a, T_b]$.
\end{itemize}

\begin{figure}
\begin{center}
\mbox{
\subfigure[]{\includegraphics[width=0.33\textwidth]{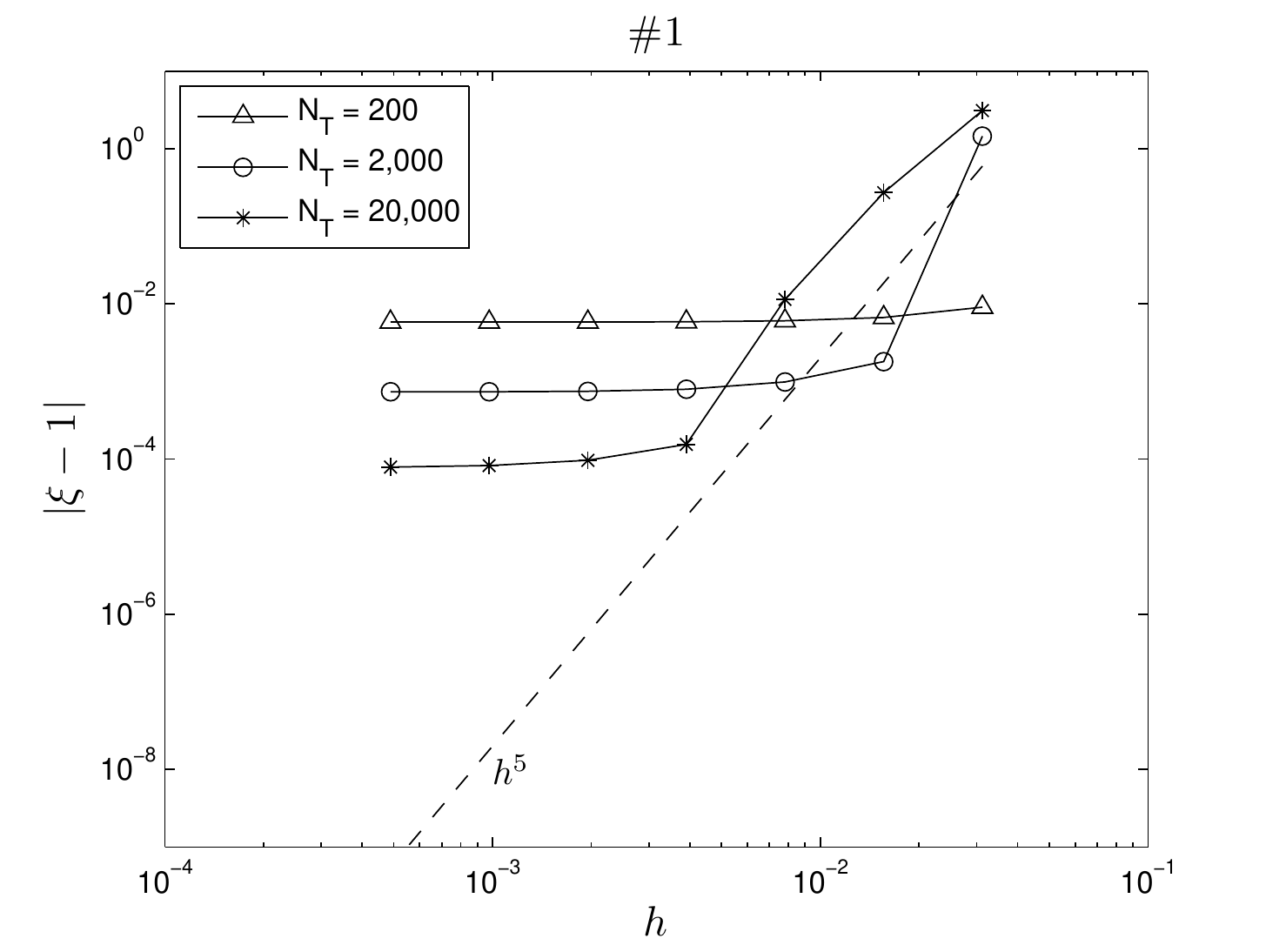}}
\subfigure[]{\includegraphics[width=0.33\textwidth]{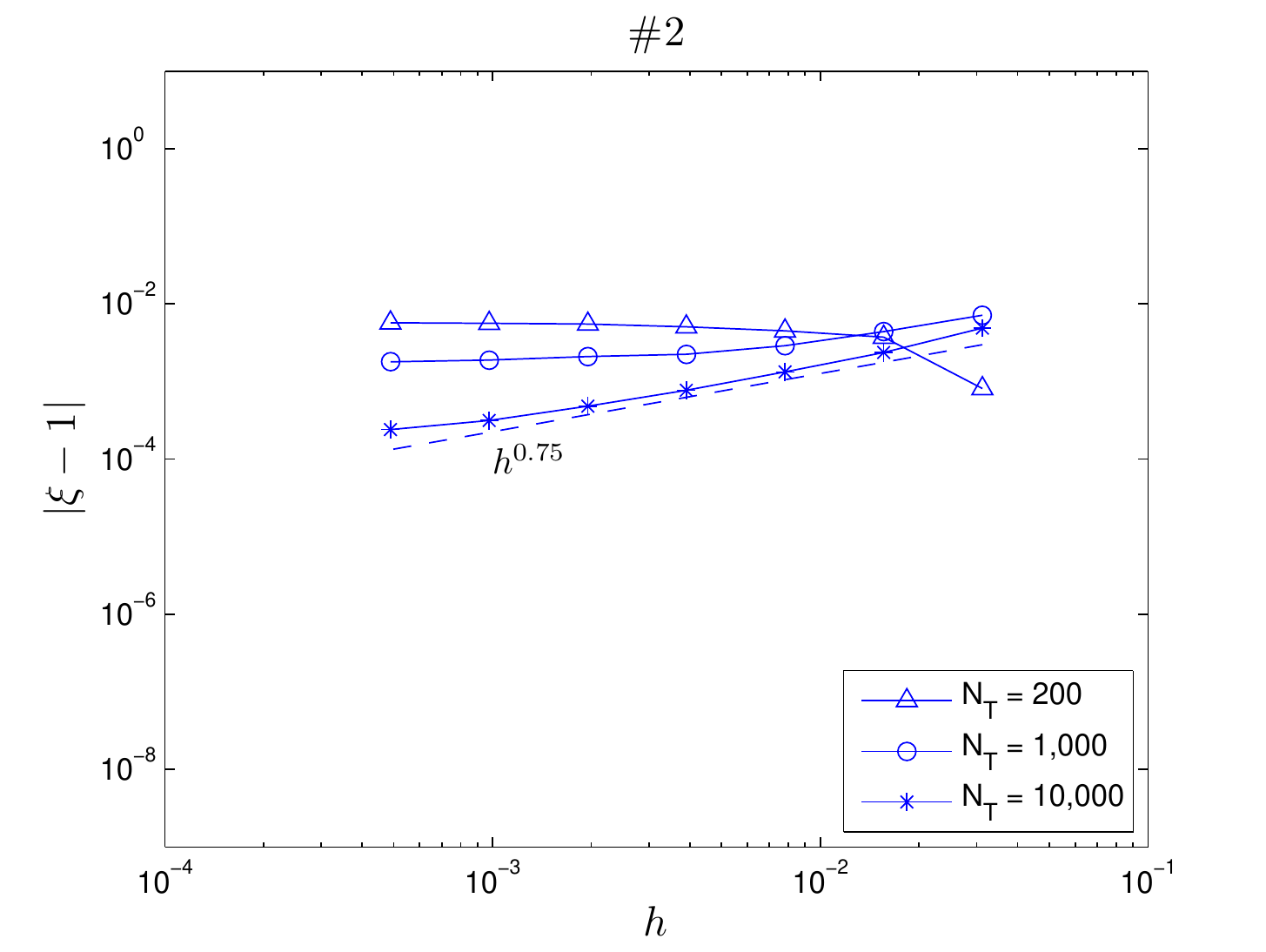}}
\subfigure[]{\includegraphics[width=0.33\textwidth]{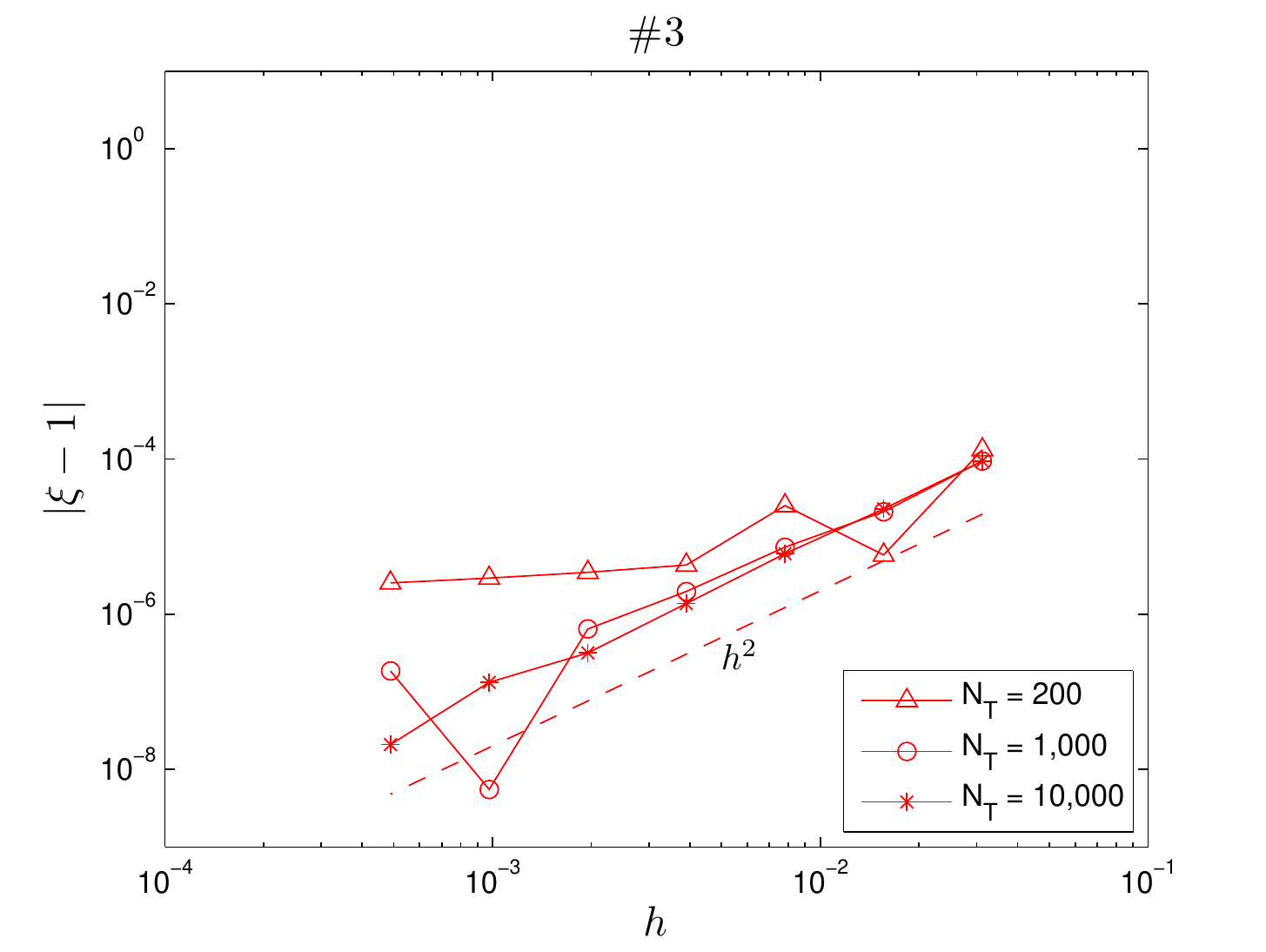}}}
\mbox{
\subfigure[]{\includegraphics[width=0.33\textwidth]{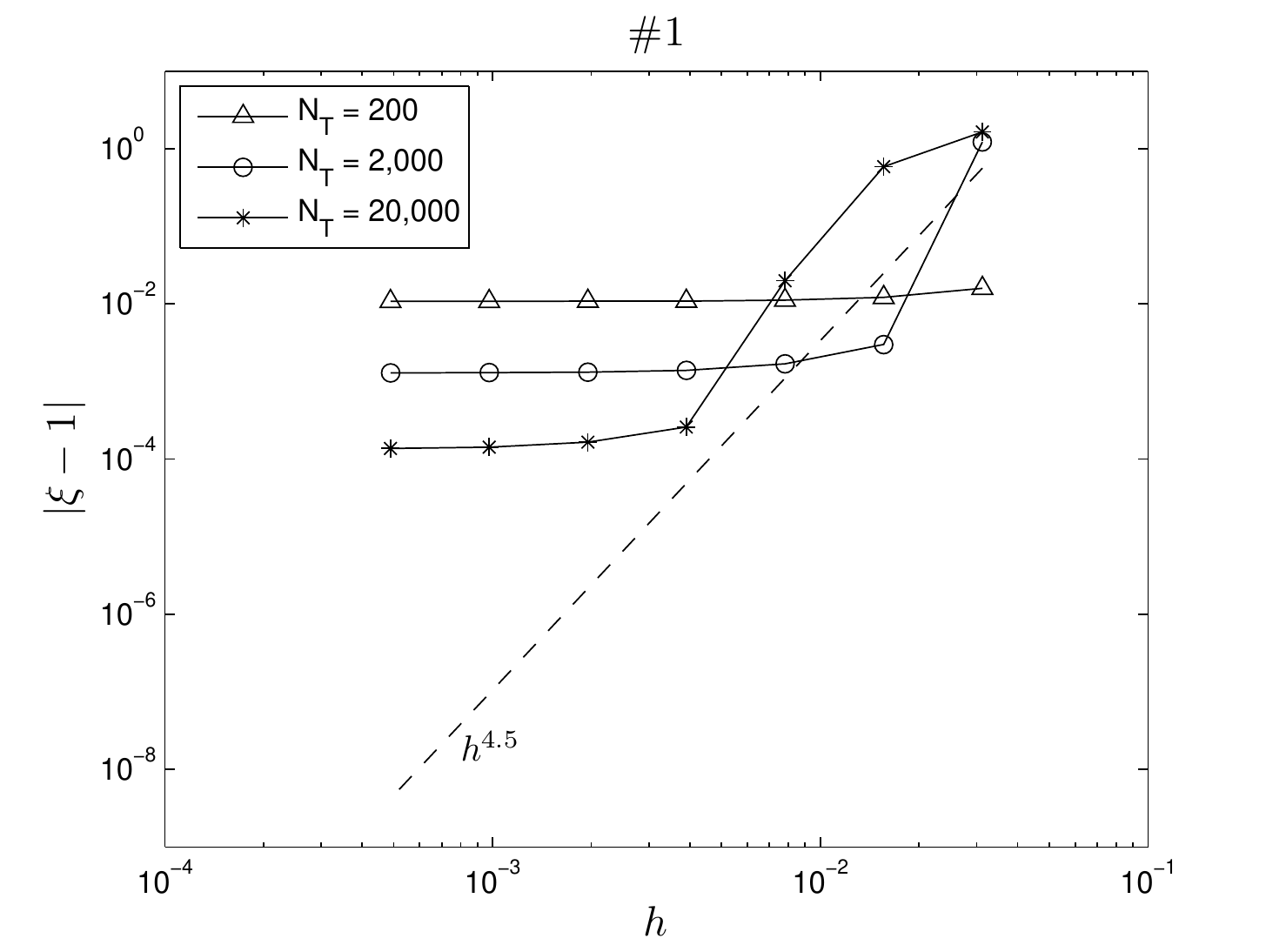}}
\subfigure[]{\includegraphics[width=0.33\textwidth]{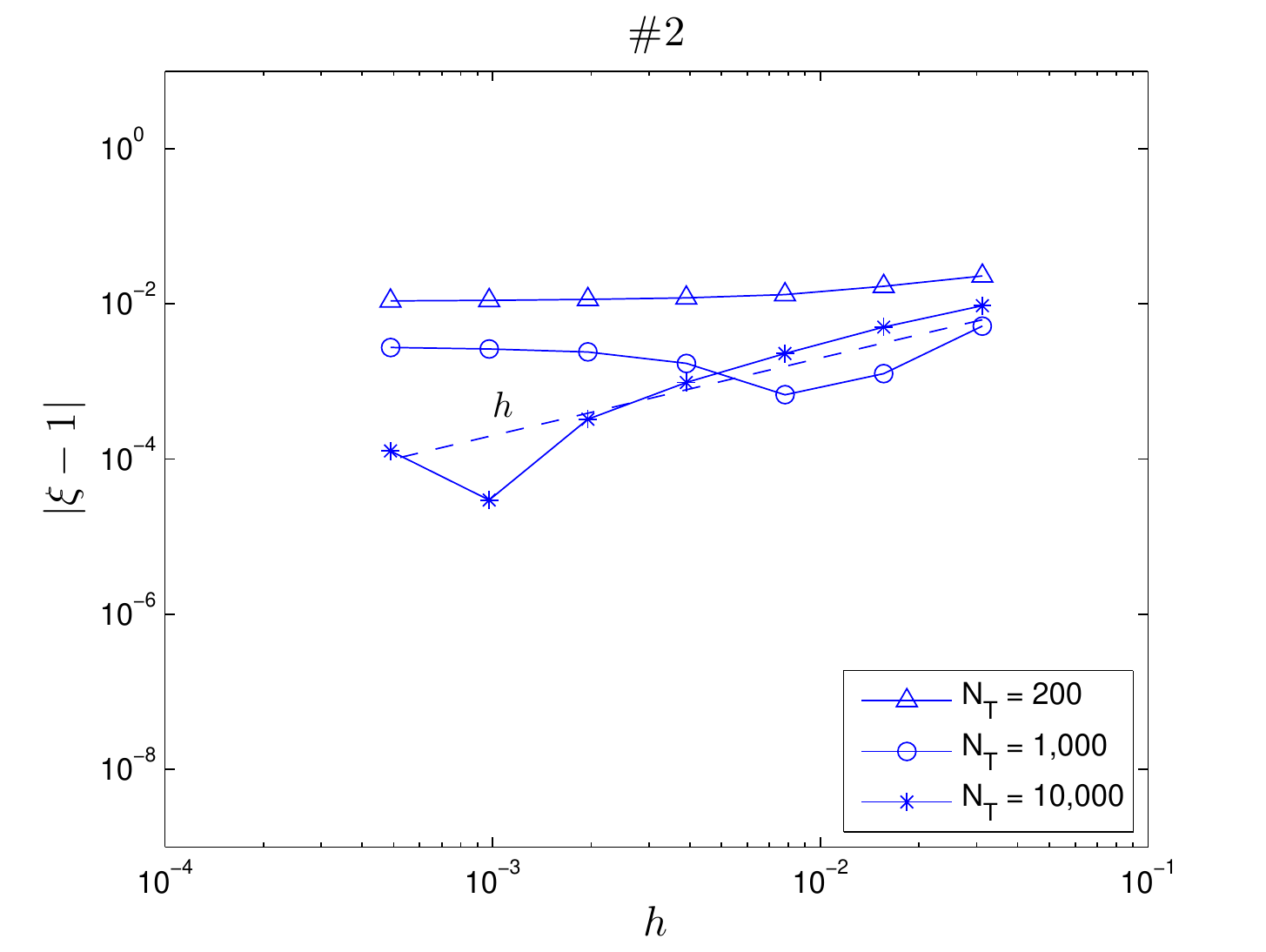}}
\subfigure[]{\includegraphics[width=0.33\textwidth]{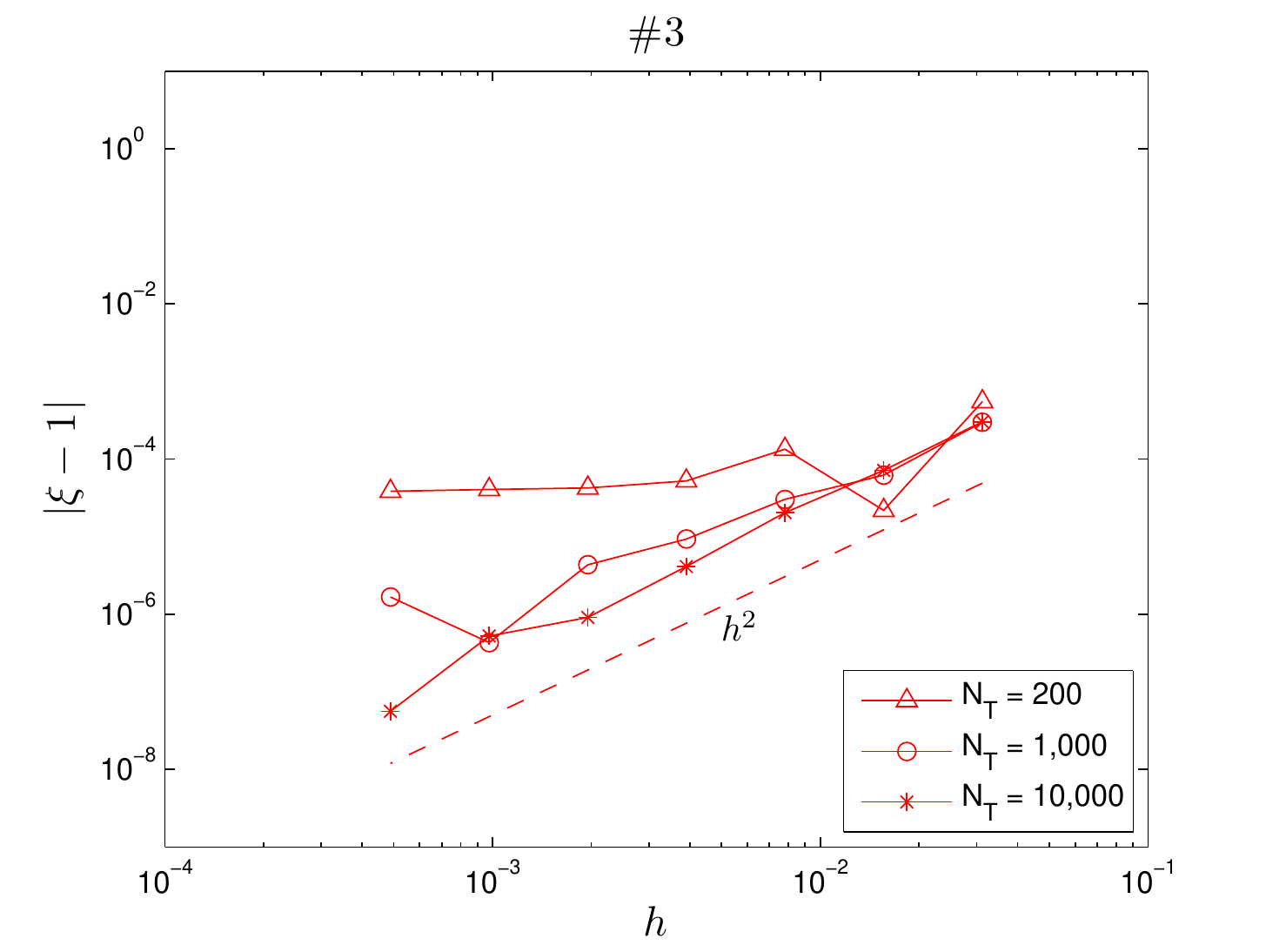}}}
\mbox{
\subfigure[]{\includegraphics[width=0.33\textwidth]{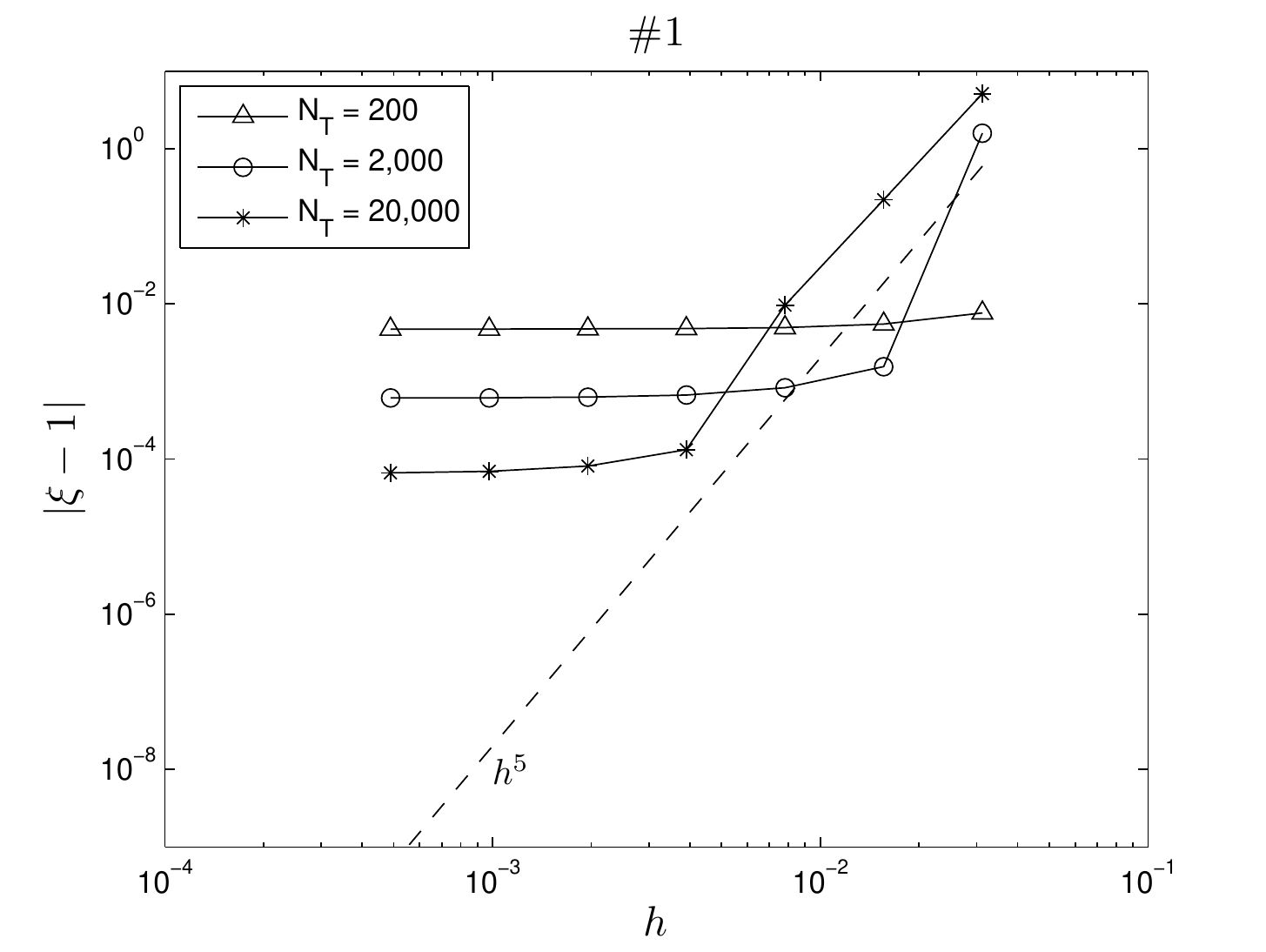}}
\subfigure[]{\includegraphics[width=0.33\textwidth]{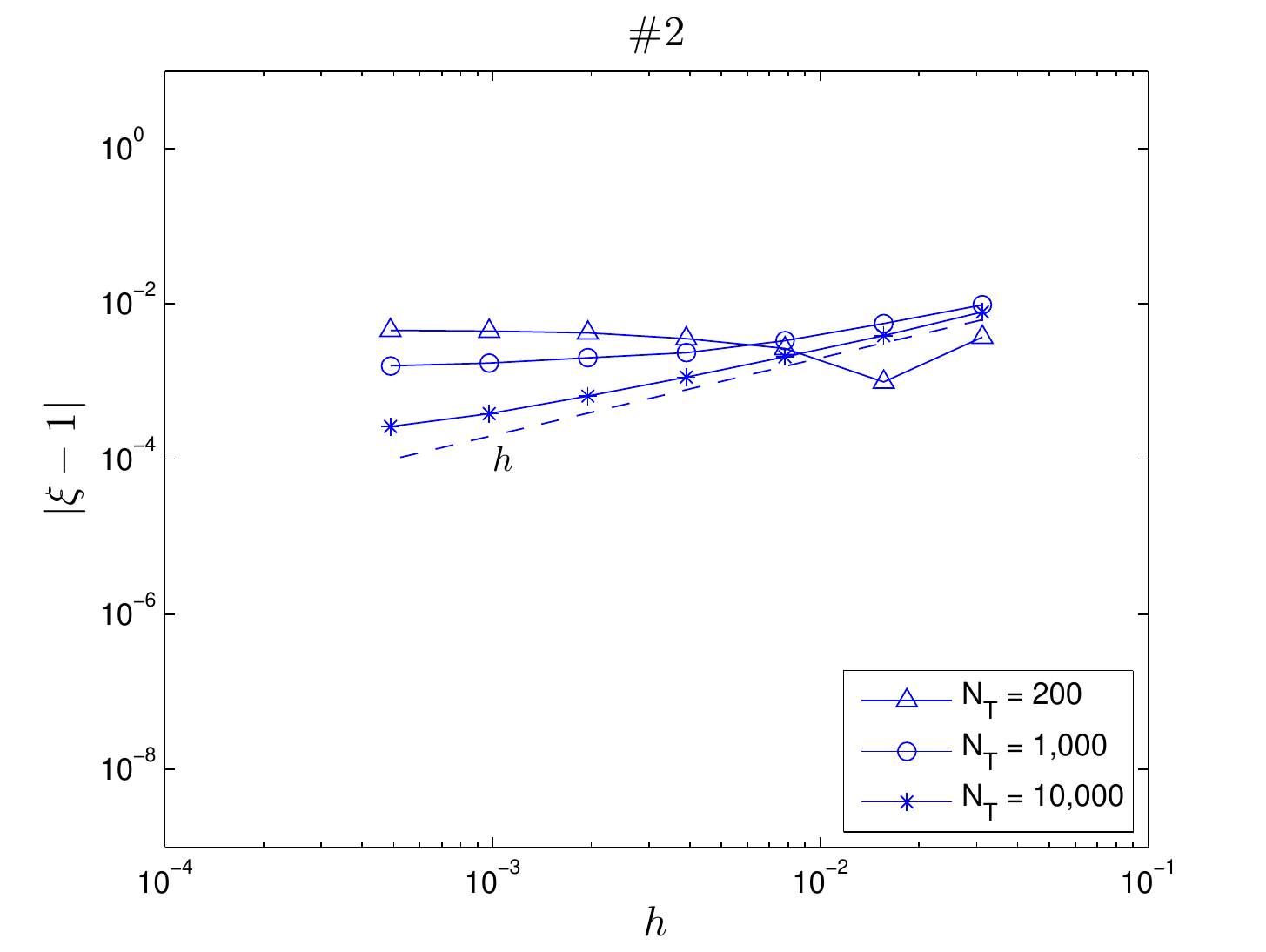}}
\subfigure[]{\includegraphics[width=0.33\textwidth]{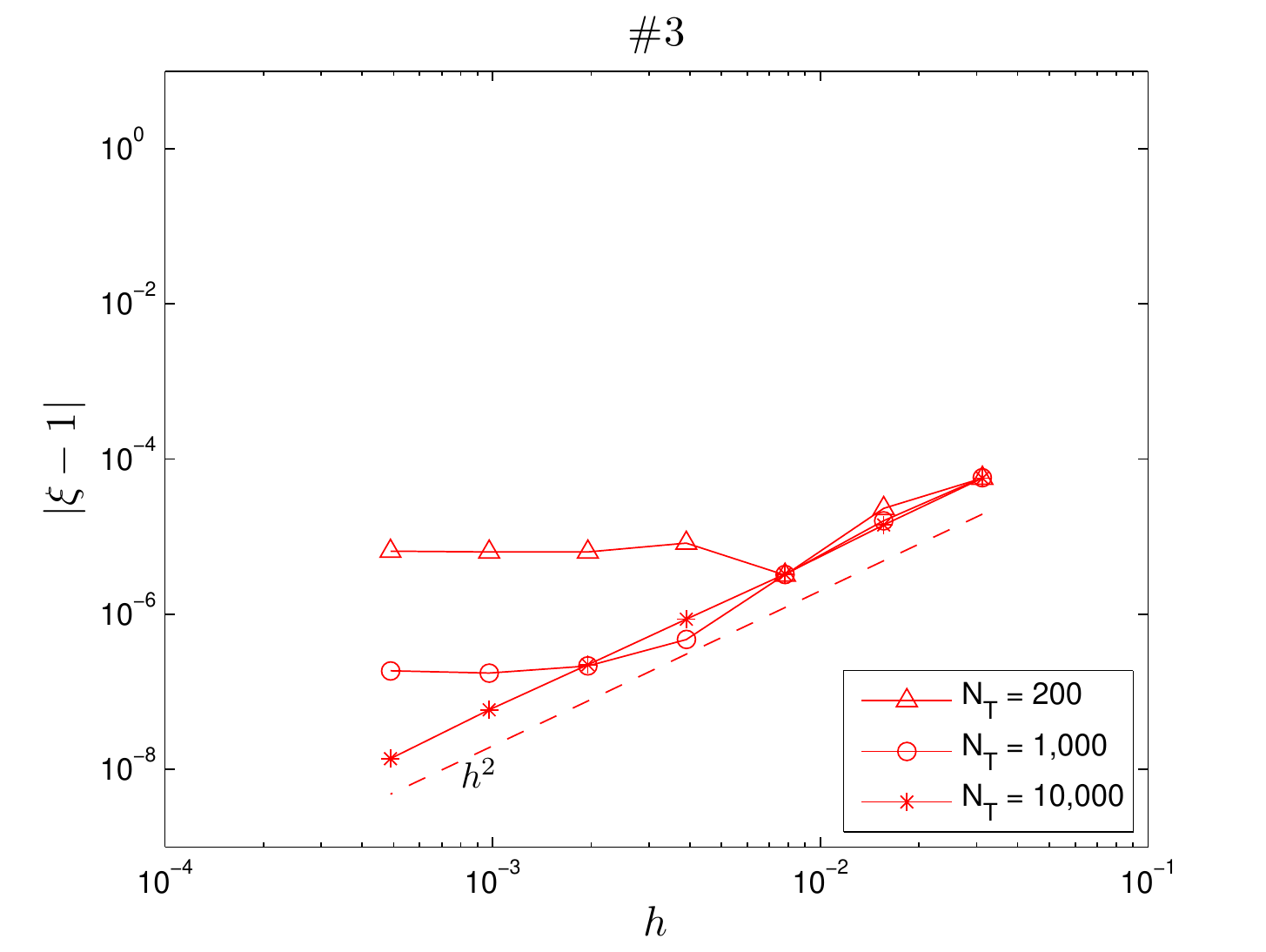}}}
\caption{Relative error $| \xi - 1 |$, where $\xi = \frac{\textrm{LHS of
  \eqref{eq:test2_int}}}{\textrm{RHS of \eqref{eq:test2_int}}}$, versus discretization
  step $h = \Delta x = \Delta y$ in approximating the RHS in \eqref{eq:test2_int}.
  The first, second and third row of figures show the results for $\mu'_1$, $\mu'_2$ and
  $\mu'_3$, respectively, while the figures in the first, second and third column
  represent line integration (\#1), delta function approximation (\#2) and
  area integration (\#3) methods, respectively.}
\label{fig:test2a}
\end{center}
\end{figure}

\begin{figure}
\begin{center}
\mbox{
\subfigure[]{\includegraphics[width=0.33\textwidth]{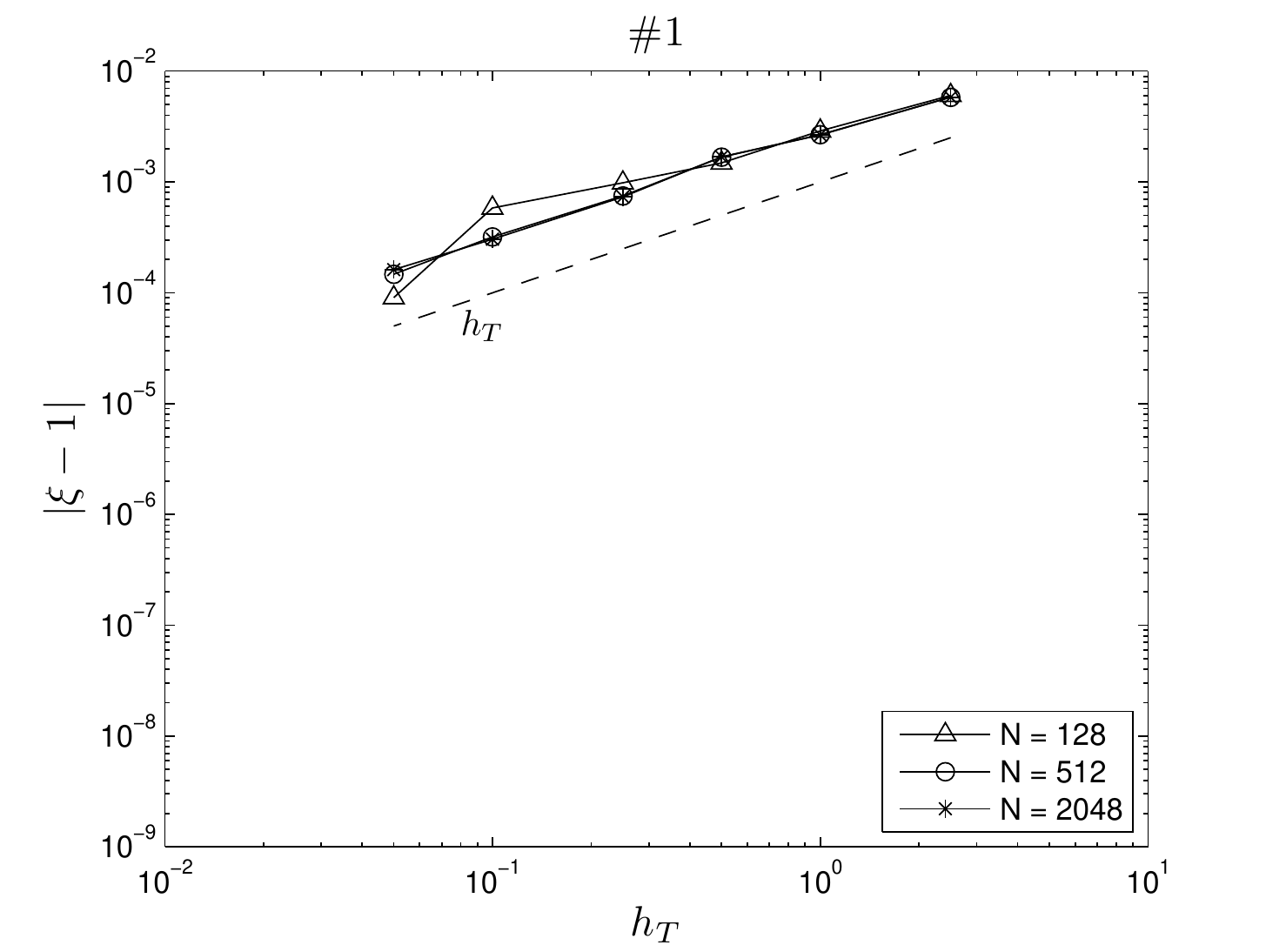}}
\subfigure[]{\includegraphics[width=0.33\textwidth]{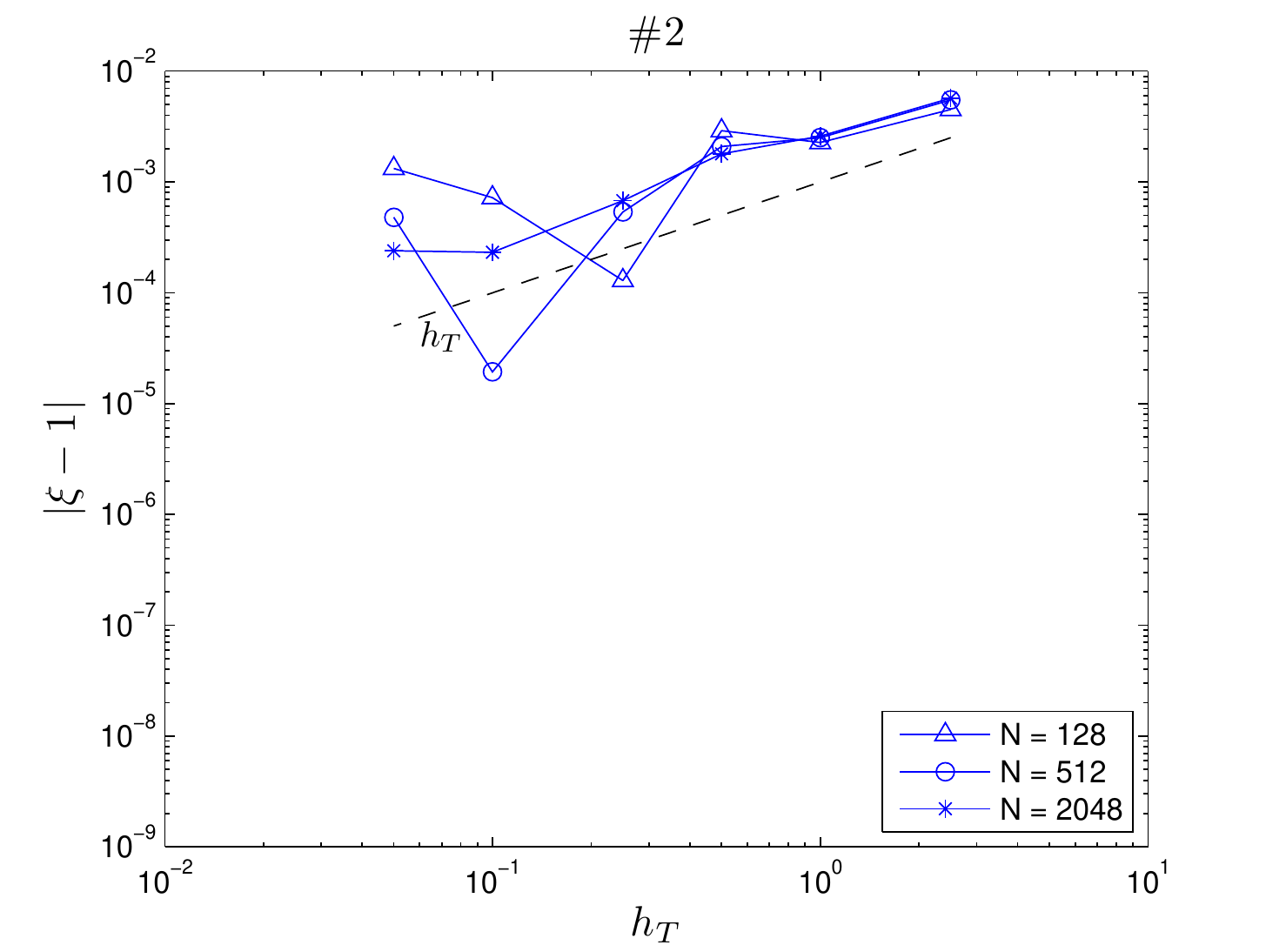}}
\subfigure[]{\includegraphics[width=0.33\textwidth]{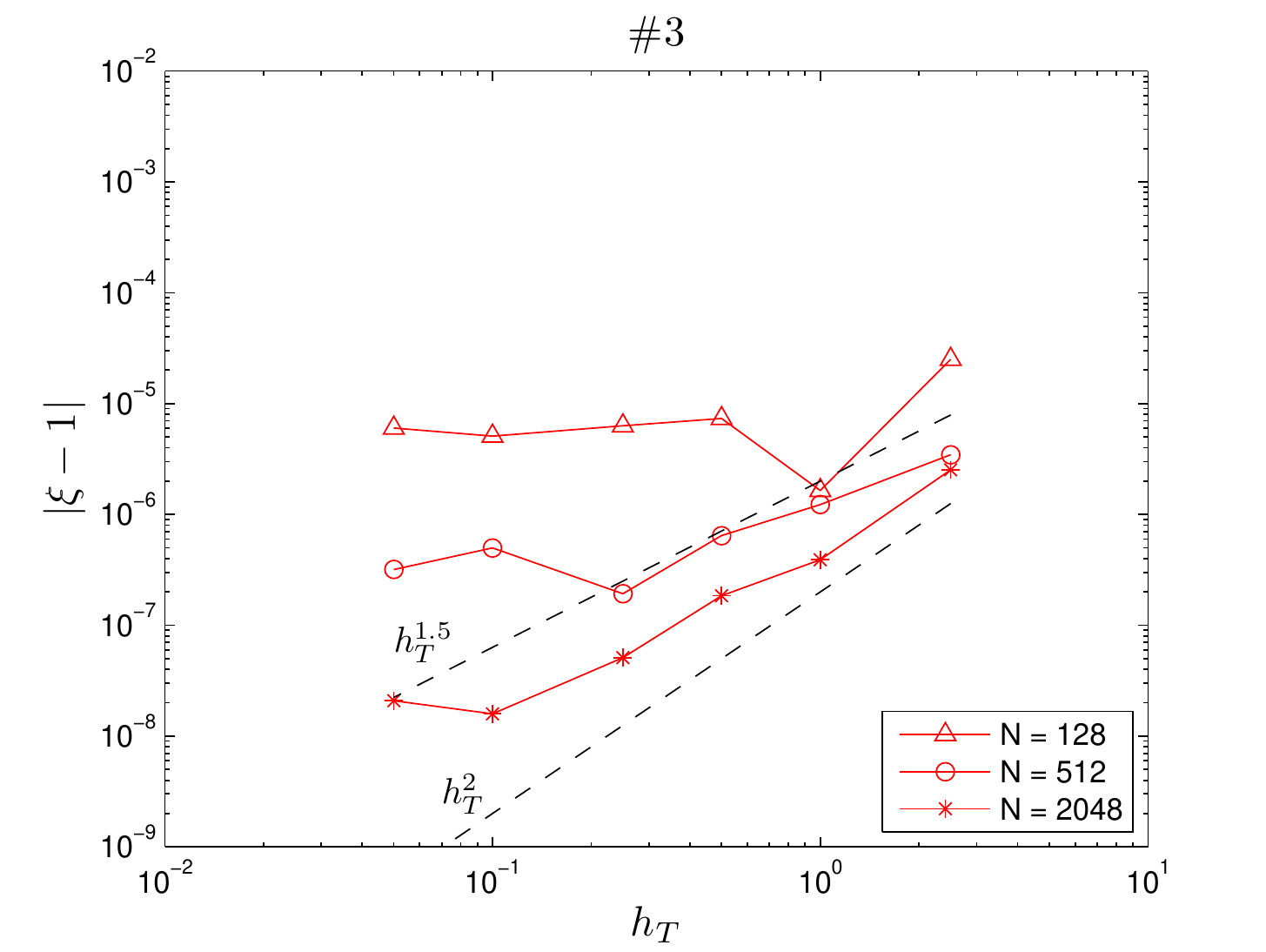}}}
\mbox{
\subfigure[]{\includegraphics[width=0.33\textwidth]{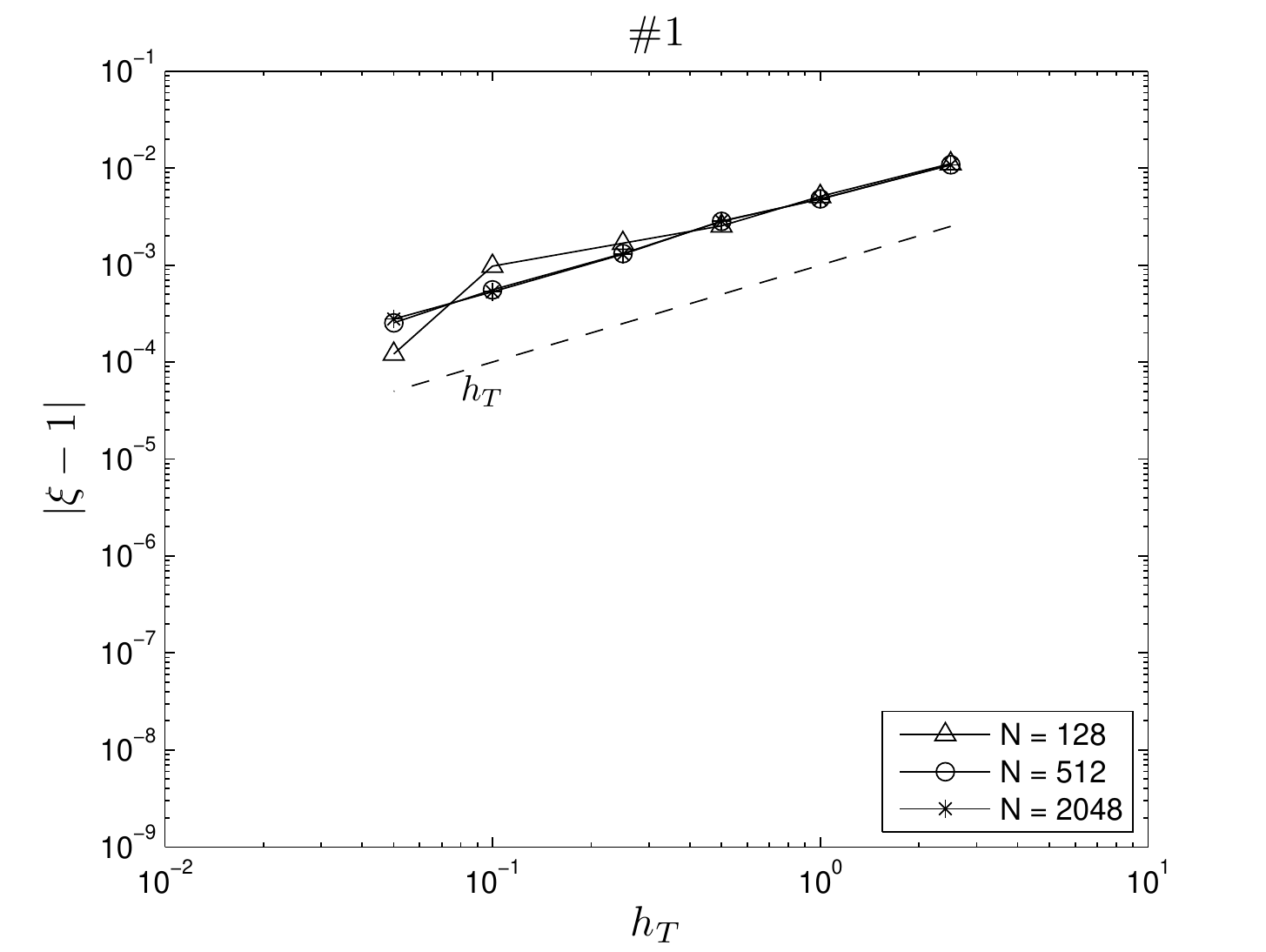}}
\subfigure[]{\includegraphics[width=0.33\textwidth]{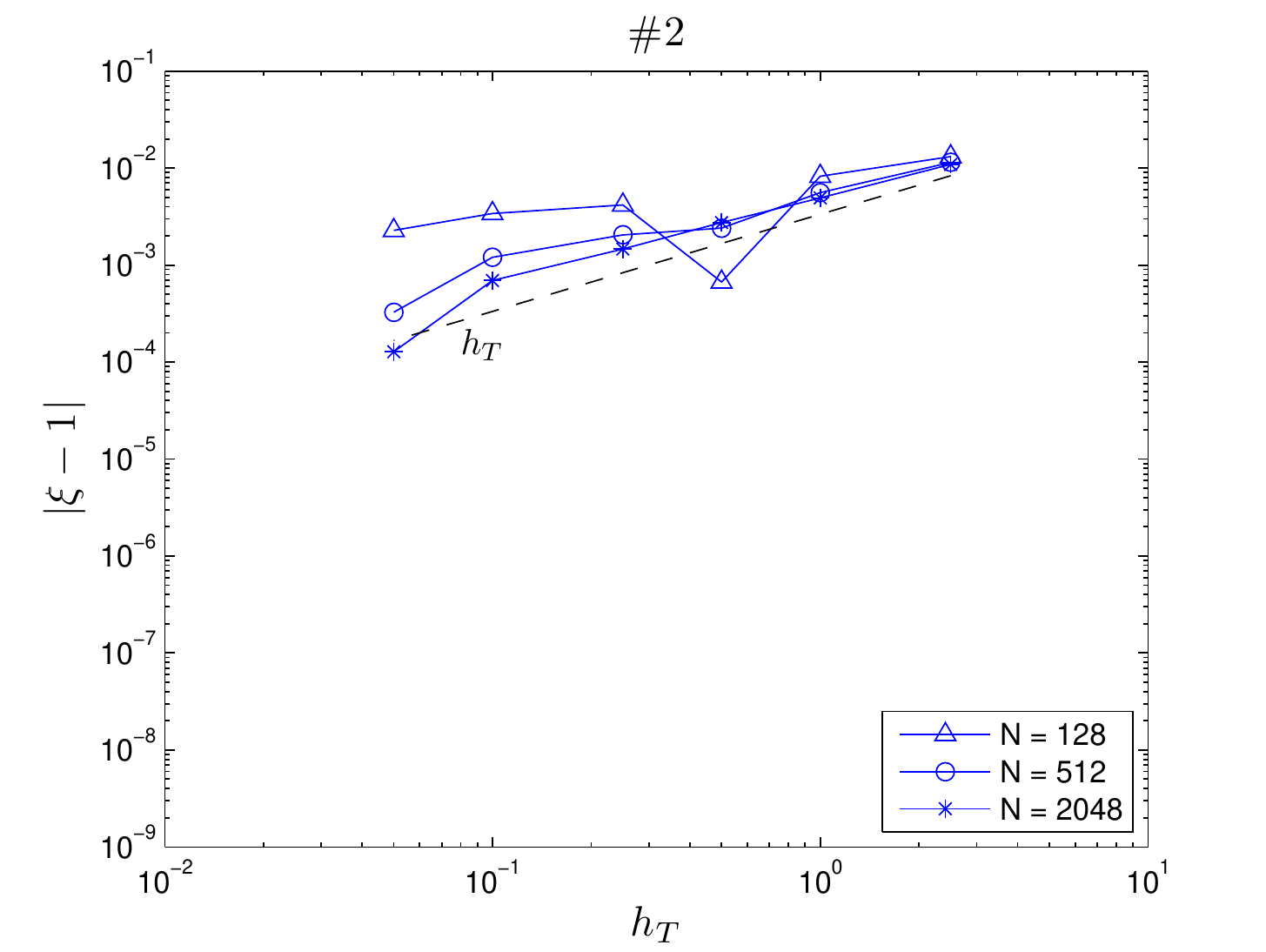}}
\subfigure[]{\includegraphics[width=0.33\textwidth]{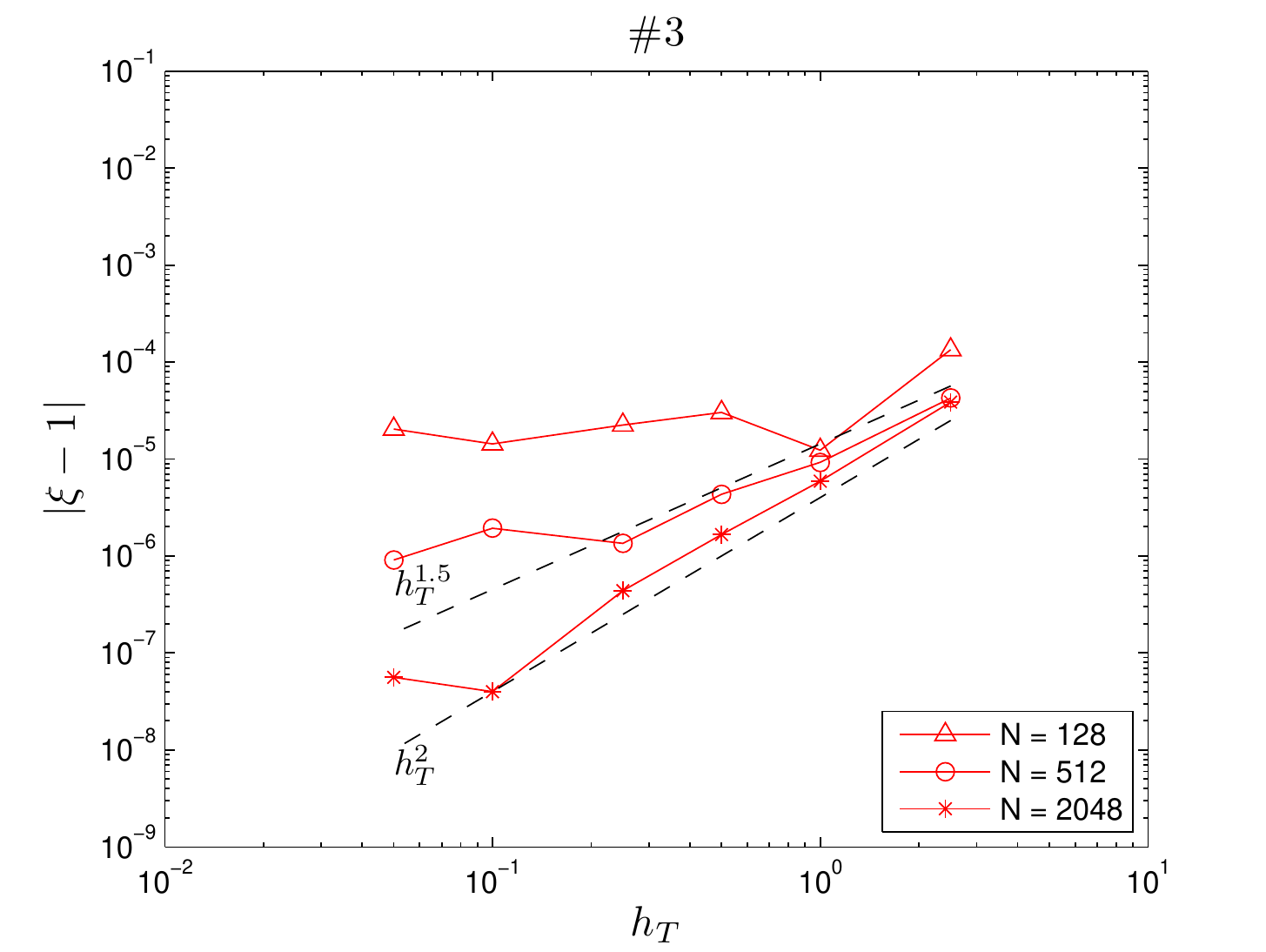}}}
\mbox{
\subfigure[]{\includegraphics[width=0.33\textwidth]{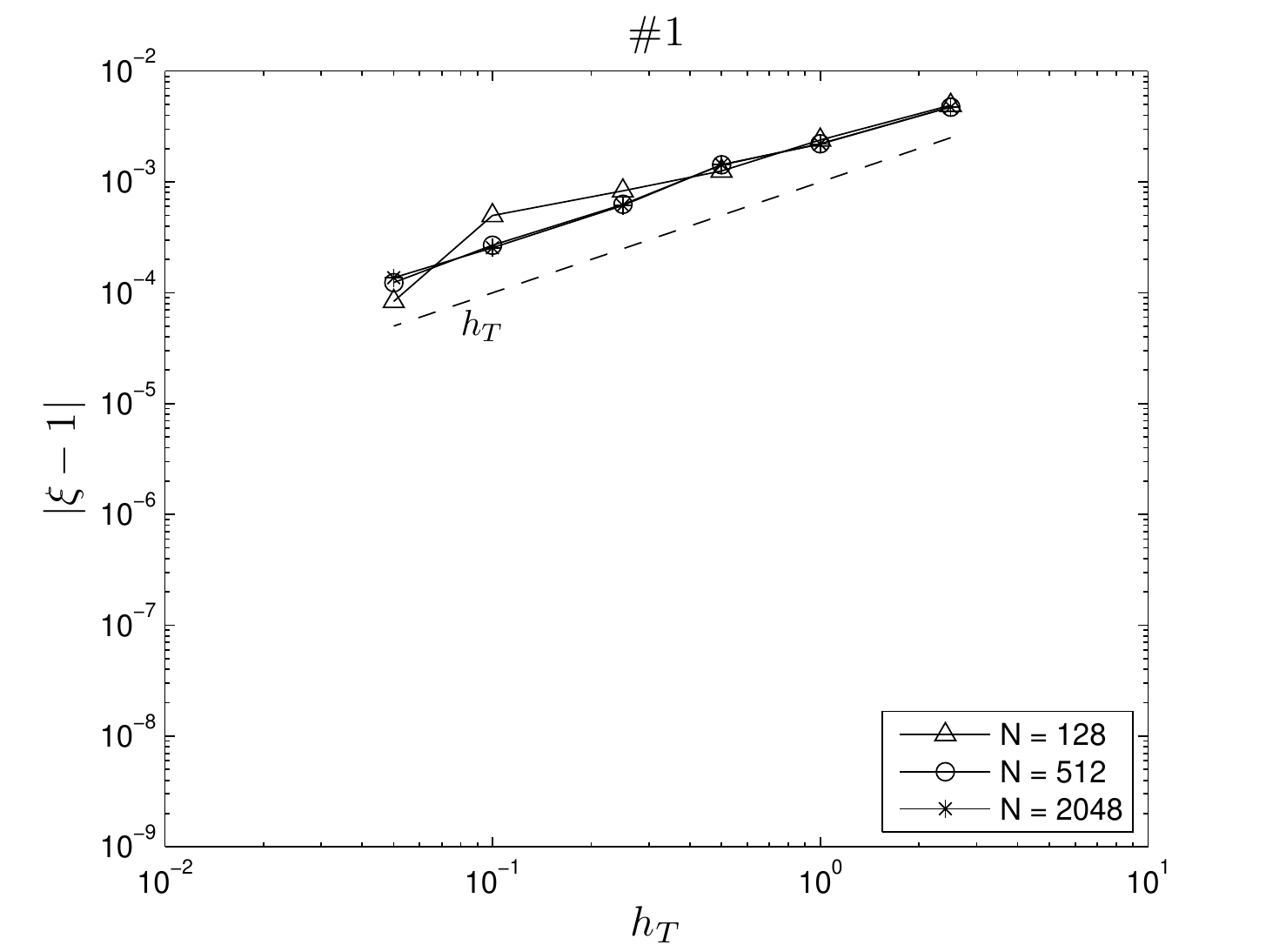}}
\subfigure[]{\includegraphics[width=0.33\textwidth]{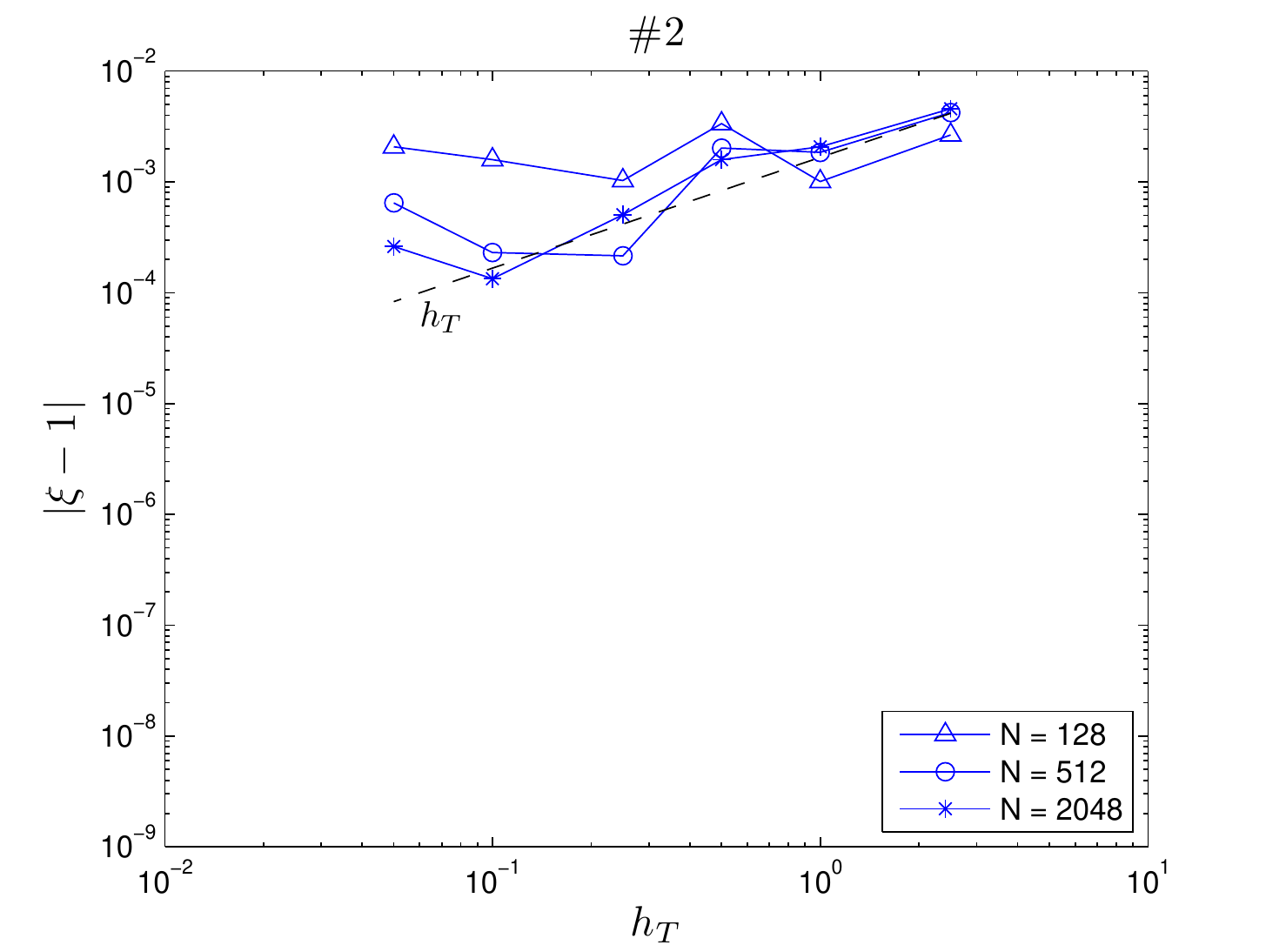}}
\subfigure[]{\includegraphics[width=0.33\textwidth]{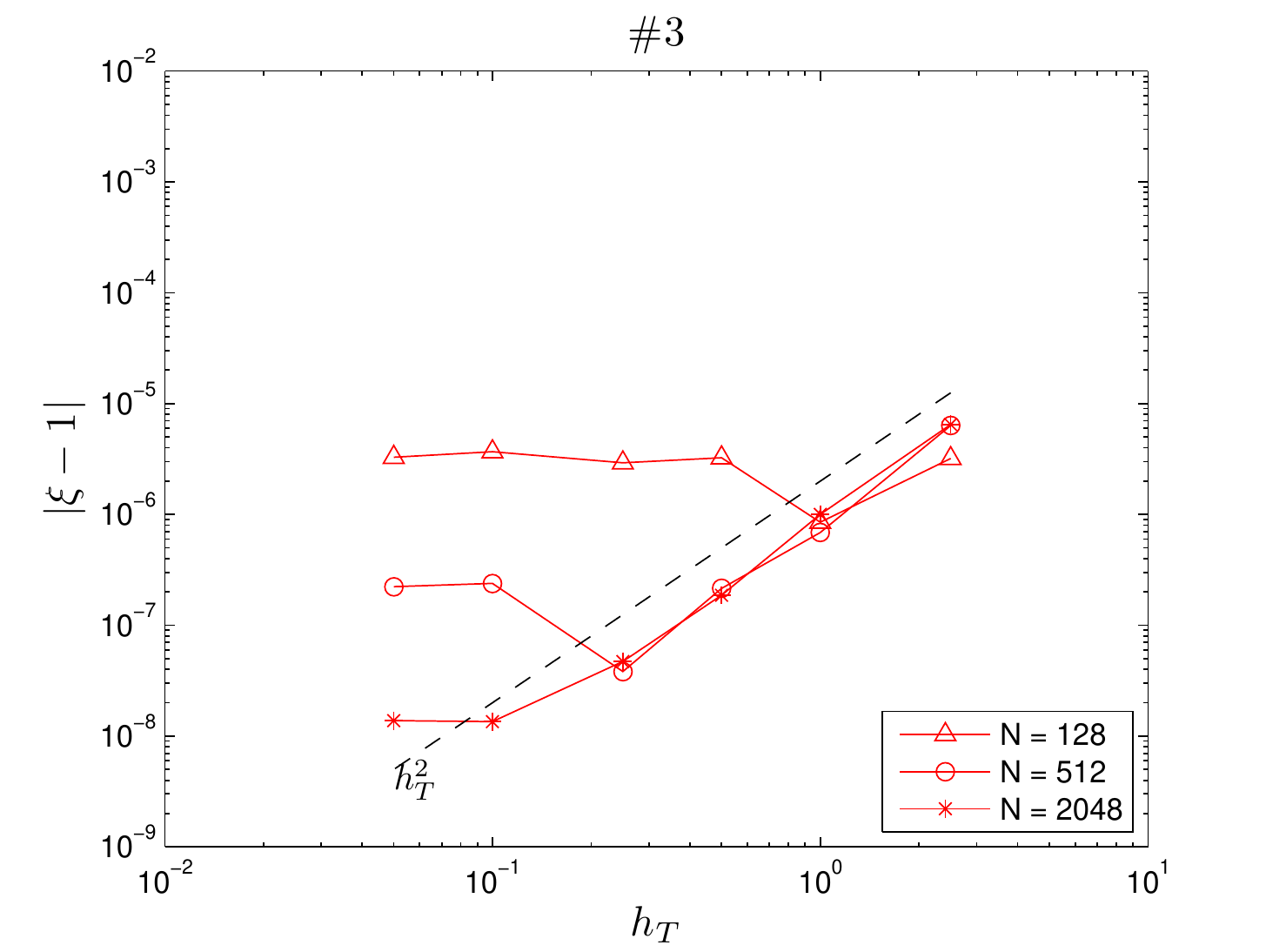}}}
\caption{Relative error $| \xi - 1 |$, where $\xi = \frac{\textrm{LHS of
  \eqref{eq:test2_int}}}{\textrm{RHS of \eqref{eq:test2_int}}}$, versus discretization
  step $h_T$ in the state space $\L$ in approximating the RHS in \eqref{eq:test2_int}.
  The first, second and third row of figures show the results for $\mu'_1$, $\mu'_2$ and
  $\mu'_3$, respectively, while the figures in the first, second and third column
  represent line integration (\#1), delta function approximation (\#2) and
  area integration (\#3) methods, respectively.}
\label{fig:test2b}
\end{center}
\end{figure}

As is evident from Figure \ref{fig:test2a}, all three methods show
similar qualitative behavior, namely, the error decreases with
decreasing $h$ until it saturates which is due to the error terms
depending on $h_T$ becoming dominant. The saturation value of the
error depends on the state space resolution $h_T$ and is different for
the different methods.  Method \#3 (area integration) reveals accuracy
$\O(h^2)$, whereas method \#2 (delta function approximation) is again
only of accuracy about $\O(h)$ for the same discretization of the
interval $\L$.  Method \#1 (line integration) performs better and
shows accuracy up to $\O(h^5)$, but requires much finer resolution in
the state space, namely $N_T > 20000$ ($h_T < 0.03$) is needed for
this behavior to be visible. On the other hand, method \#3 (area
integration) leads to the smallest errors for all the cases tested.

Analogous data is plotted in Figure \ref{fig:test2b} now as a function of
the state space resolution $h_T$ with $h$ acting as a
parameter. Similar trends are observed as in Figure \ref{fig:test2a},
namely, the errors decrease with $h_T$ until they eventually saturate
when the error terms depending on $h$ become dominant. Methods \#1 and
\#2 reveal accuracy $\O(h_T)$, whereas method \#3 has accuracy
$\O(h_T^{1.5 \div 2})$ which is actually better than stipulated
by Theorem \ref{thm:m3}, cf.~\eqref{eq:direct_frm_acc}. Method \#3 is
also characterized by the smallest value of the constant prefactor
leading to the smallest overall errors.

\subsubsection{Tests for $s$ Varying Over a Finite Range with
$\phi(s,\x)$ and $g(\x)$ Given by Solutions of Direct and Adjoint Problem}
\label{sec:test_3}

We now repeat the test described in Section \ref{sec:test_2} using
$\phi(s,\x) = T(\x) - s$ and $g(\x) = [\bnabla \u(\x) + (\bnabla \u(\x))^T] :
\bnabla \u^*(\x)$, where the fields $\u$, $T$ and $\u^*$ come
from the solutions of the direct and adjoint problem
\eqref{eq:coupled_PDEs}--\eqref{eq:coupled_BC} and
\eqref{eq:adjoint_coupled}--\eqref{eq:BC_adjoint}
at some fixed time $t$, see Figure \ref{fig:test3_fields} (details of
these computations will be given in Section \ref{sec:estim}). As
before, we discretize the domain $\Omega = [0,1]^2$ with the step $h =
\Delta x = \Delta y = 2^{-(4+i)}$, $i = 1 \dots 7$, and the state space
$\L = [T_a, T_b]$, $T_a = 100$ and $T_b = 700$, with the step $h_T =
0.06$ ($N_T = 10000$).

\begin{figure}
\begin{center}
\mbox{
\subfigure[]{\includegraphics[width=0.5\textwidth]{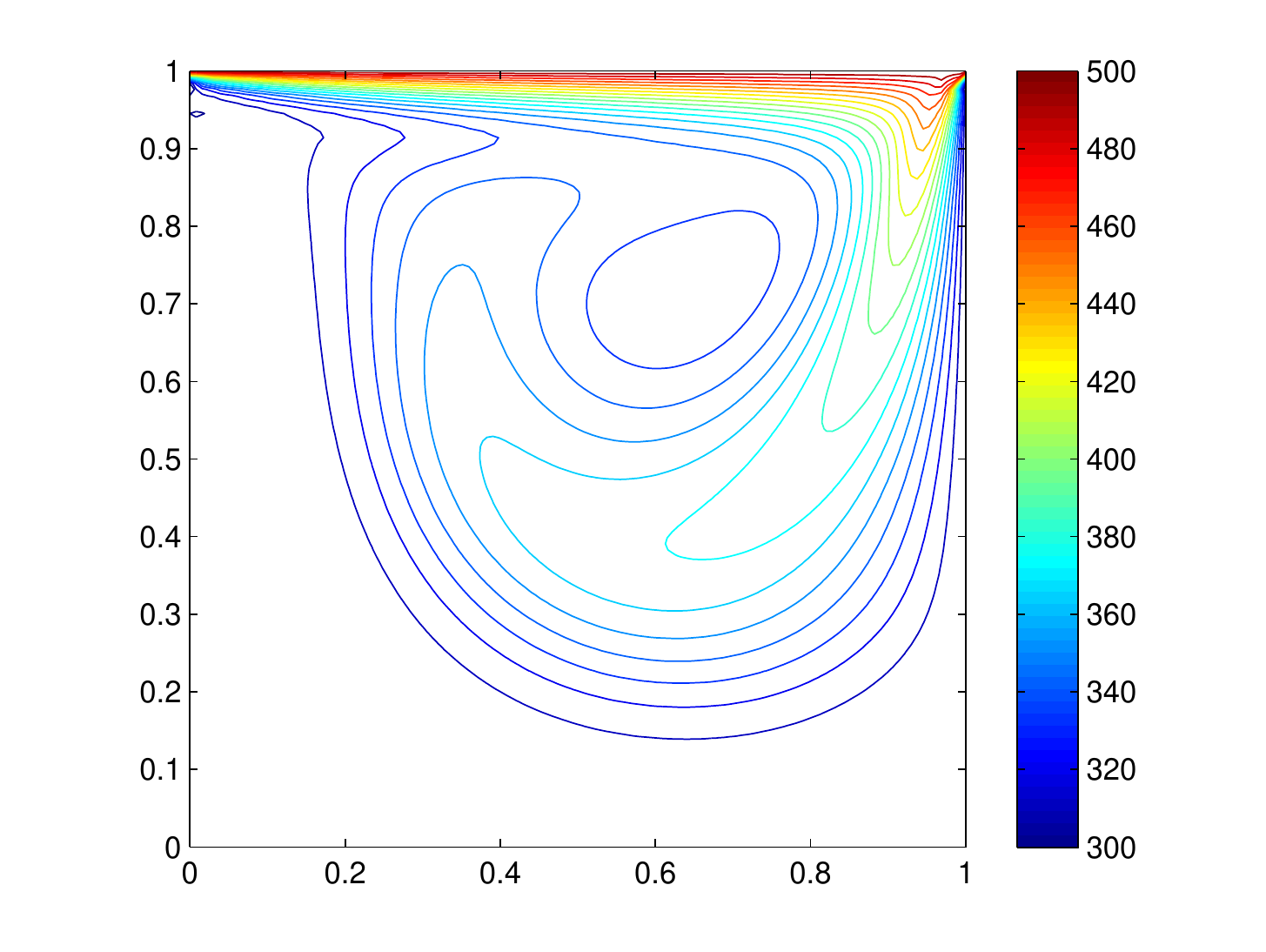}}
\subfigure[]{\includegraphics[width=0.5\textwidth]{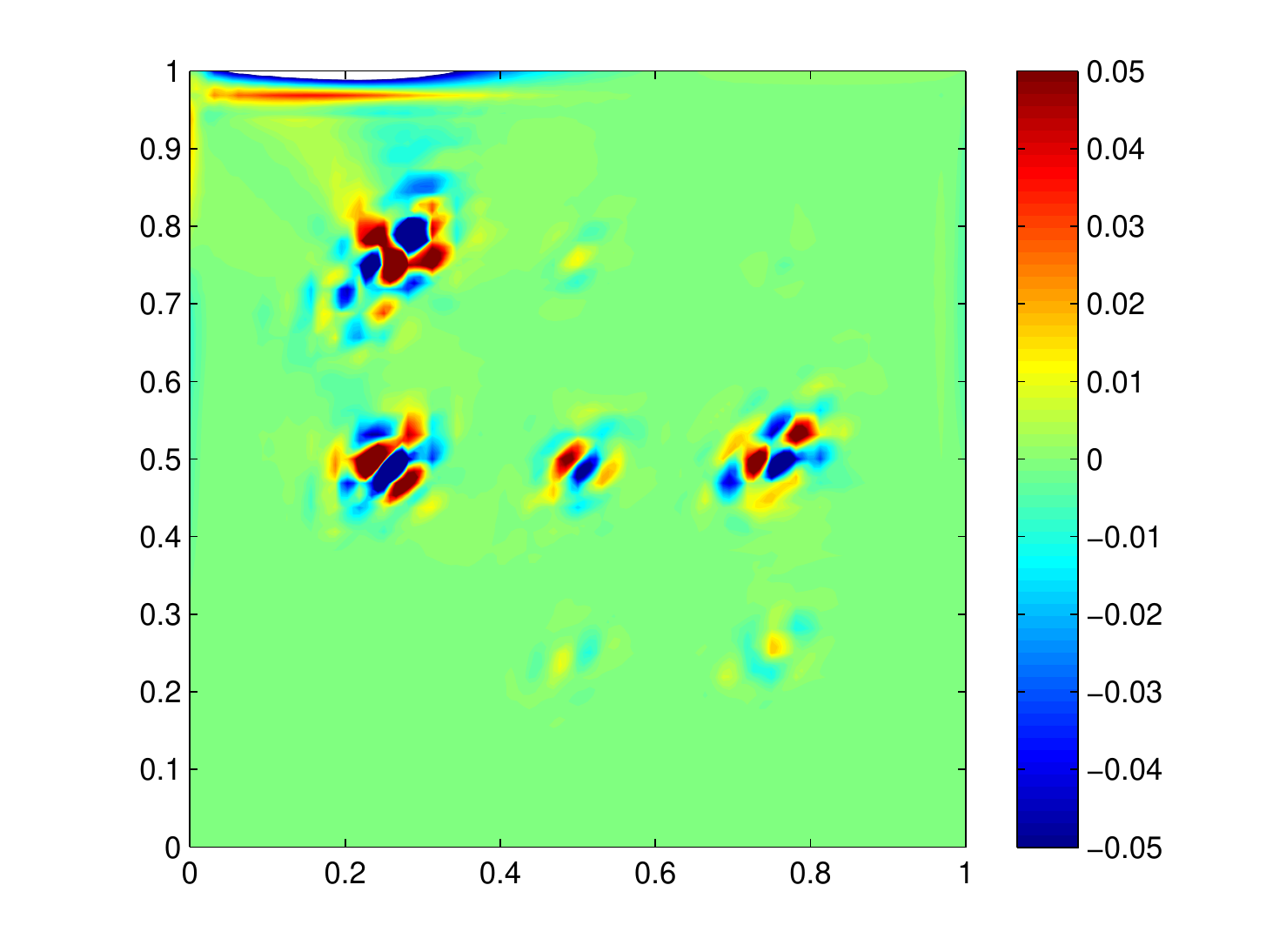}}}
\caption{The fields (a) $T(\x)$ (level set contours) and (b) $g(\x) =
  [\bnabla \u(\x) + (\bnabla \u(\x))^T] : \bnabla \u^*(\x)$ obtained
  at some time $t$ by solving
  \eqref{eq:coupled_PDEs}--\eqref{eq:coupled_BC} and
  \eqref{eq:adjoint_coupled}--\eqref{eq:BC_adjoint}.}
\label{fig:test3_fields}
\end{center}
\end{figure}

The data shown in Figure \ref{fig:test3} confirms our findings from
Sections \ref{sec:test_1} and \ref{sec:test_2}, namely, that in this
case as well the error of all three methods decreases with $h$ until
it eventually saturates when the errors depending on $h_T$ become
dominant. Method \#3 is again characterized by the smallest prefactor
and hence leads to much smaller overall errors than in methods \#1 and
\#2. The computational complexity of our three approaches is addressed
in Figure \ref{fig:CPU_time}, where $N_e$ is defined as a number of
computational elements, i.e., $N_e = N_{\triangle}$, or $N_e =
N_{\square}$, using \eqref{eq:discr_trg} or \eqref{eq:discr_qdr},
respectively. We see that, while the complexity of methods \#1 and \#2
scales as $\O(\sqrt{N_e})$, method \#3 exhibits the scaling of
$\O(N_e)$. On the other hand, however, method \#3 has the smallest
prefactor and, at least in the range of resolutions considered here,
results in the shortest execution time.

In conclusion, these observations make the area integration approach (method \#3)
the method of choice for the present parameter reconstruction problem,
and this is the approach we will use in all subsequent computations.

\begin{figure}
\begin{center}
\mbox{
\subfigure[]{\includegraphics[width=0.33\textwidth]{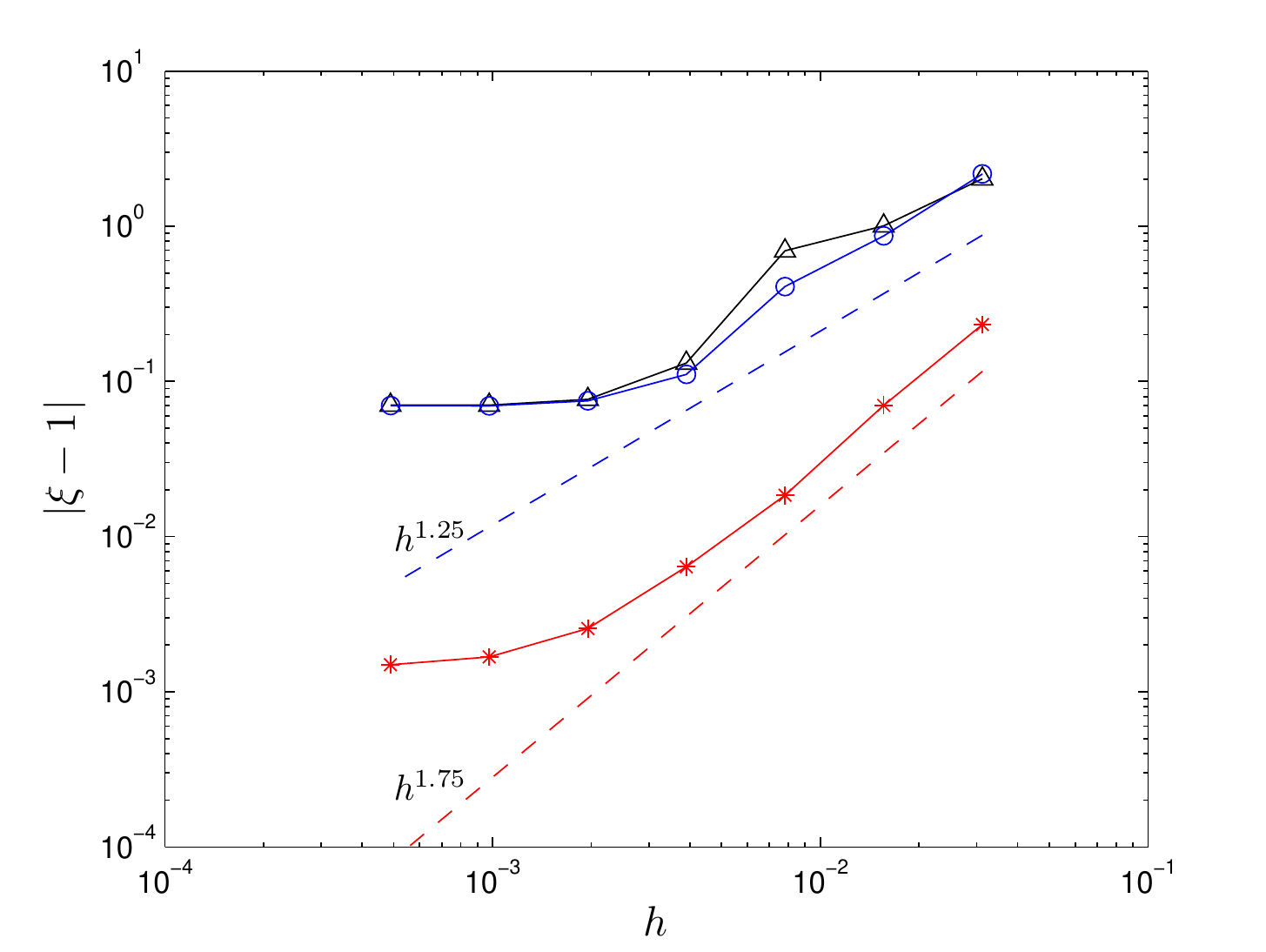}}
\subfigure[]{\includegraphics[width=0.33\textwidth]{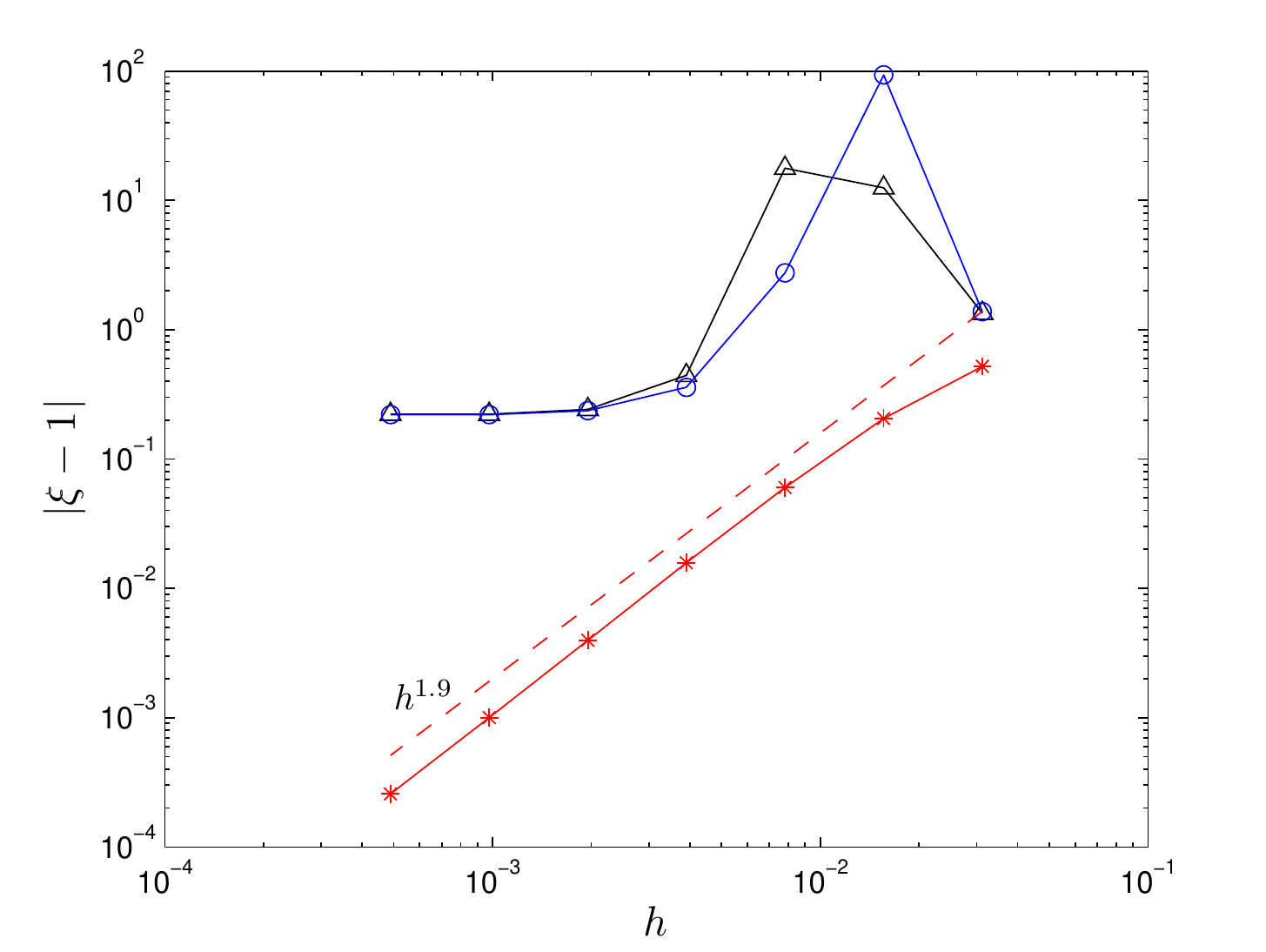}}
\subfigure[]{\includegraphics[width=0.33\textwidth]{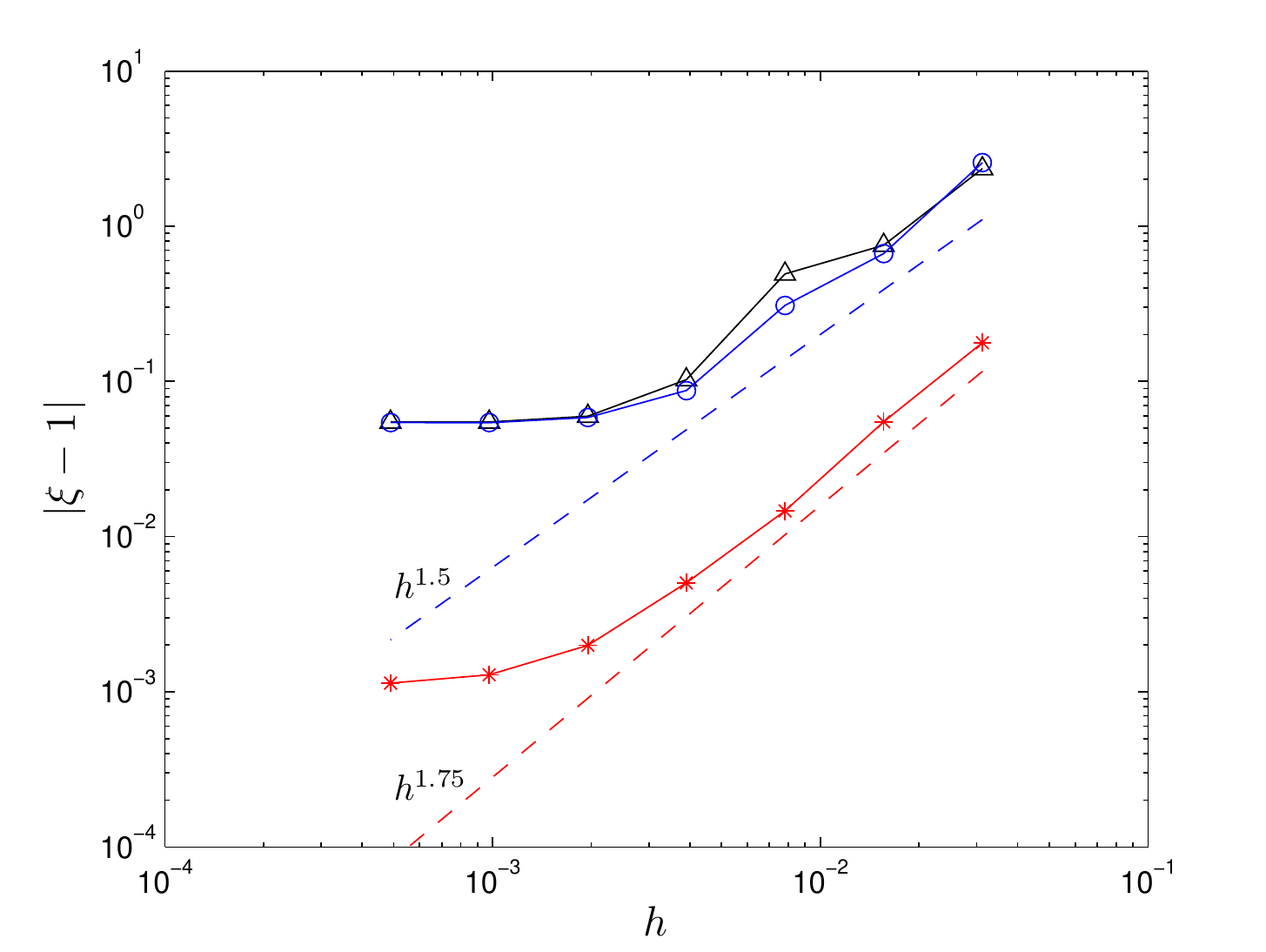}}}
\caption{Relative error $| \xi - 1 |$, where $\xi = \frac{\textrm{LHS
      of \eqref{eq:test2_int}}}{\textrm{RHS of
      \eqref{eq:test2_int}}}$, versus discretization step $h = \Delta
  x = \Delta y$ in estimating the RHS in \eqref{eq:test2_int}, where
  $\phi(s, \x) = T(\x) - s$, and $g(\x) = [\bnabla \u + (\bnabla
  \u)^T] : \bnabla \u^* $ is obtained by solving
  \eqref{eq:coupled_PDEs}--\eqref{eq:coupled_BC} and
  \eqref{eq:adjoint_coupled}--\eqref{eq:BC_adjoint}.  Figures (a), (b)
  and (c) show the results for $\mu'_1$, $\mu'_2$ and $\mu'_3$,
  respectively, using the same discretization of the state space $\L$
  with $N_T = 10000$. Triangles represent the line integration
  approach (\#1), circles show the results for the method of the delta
  function approximation (\#2), while asterisks show the data from the
  area integration approach (\#3).}
\label{fig:test3}
\end{center}
\end{figure}

\begin{figure}
\begin{center}
\includegraphics[width=0.5\textwidth]{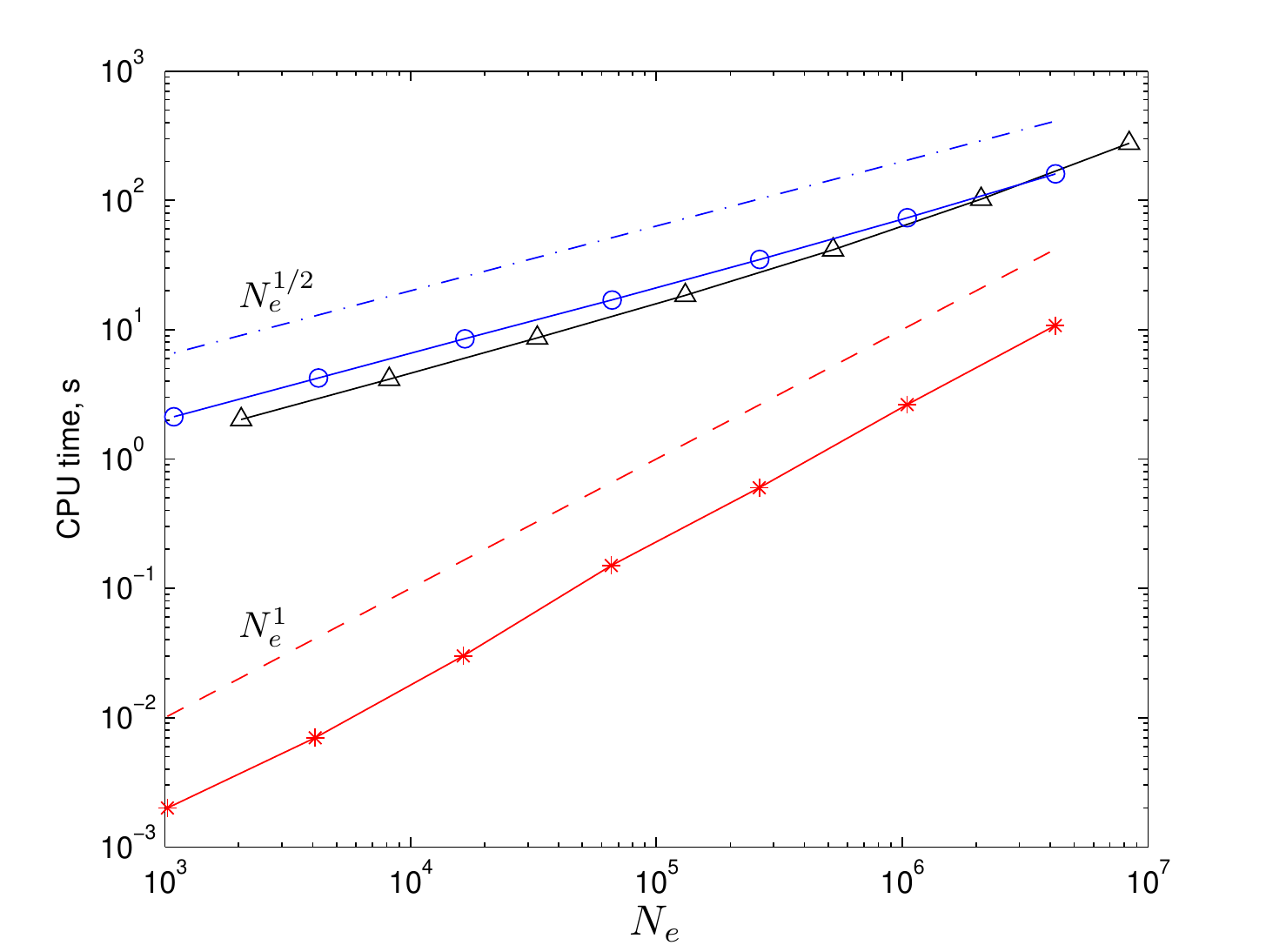}
\caption{CPU time (in seconds) versus the number $N_e$ of
  computational elements used in the different approaches, namely,
  (triangles) finite elements $\Omega^{\vartriangle}_i$ in the line
  integration approach (method \#1), (circles) grid nodes in the
  method of the delta function approximation (\#2) and (asterisks)
  finite elements $\Omega^{\square}_i$ in the area integration
  approach (\#3). The data shown corresponds to the estimation of the
  RHS in \eqref{eq:test2_int} for $\mu'_1$ using the same
  discretization of the state space with $N_T = 10000$, $\phi(s, \x) =
  T(\x) - s$ and $g(\x) = [\bnabla \u + (\bnabla \u)^T] : \bnabla \u^*
  $ obtained by solving \eqref{eq:coupled_PDEs}--\eqref{eq:coupled_BC}
  and \eqref{eq:adjoint_coupled}--\eqref{eq:BC_adjoint}.}
\label{fig:CPU_time}
\end{center}
\end{figure}

\subsection{Models for Constitutive Relations}
\label{sec:models}

For validation purposes one needs an algebraic expression to
represent the dependence of the viscosity coefficient on temperature
which could serve as the ``true'' material property we will seek to
reconstruct. The dynamic viscosity in liquids is usually approximated
by exponential relations \cite{Perez07} and one of the most common
expression for the coefficient of the dynamic viscosity is the law
of Andrade (also referred to as the Nahme law) which is given
in the dimensional form valid for $T$ expressed in Kelvins
in \eqref{eq:Andrade} below
\begin{equation}
\label{eq:Andrade}
\tilde \mu(T) = C_1 e^{C_2/T},
\end{equation}
where $C_1, \ C_2 > 0$ are constant parameters.
As regards the thermal conductivity $k$, since it typically reveals a
rather weak dependence on the temperature, for the sake of simplicity
we will treat it as a constant setting $k = 0.002$ in all computations
presented in this paper.

\subsection{Model Geometry and PDE Solvers}
\label{sec:solvers}

To validate the accuracy and performance of the proposed approach to
reconstruct $\mu(T)$, we use a simple 2D lid--driven (shear--driven)
cavity flow, cf.~Figure \ref{fig:cavity}, as our model problem.  Due
to the simplicity of its geometry and boundary conditions, the
lid--driven cavity flow problem has been used for a long time to
validate novel solution approaches and codes \cite{bs06, Ghia1982}.
Numerical results are available for different aspect ratios and the
problem was solved in both laminar and turbulent regimes using
different numerical techniques. Thus, this problem is a useful testbed
as there is a great deal of numerical data that can be used for
comparison. Our code for solving direct problem
\eqref{eq:coupled_PDEs}--\eqref{eq:coupled_BC} and adjoint problem
\eqref{eq:adjoint_coupled}--\eqref{eq:BC_adjoint} has been implemented
using {\tt FreeFem++} \cite{FreeFem}, an open--source, high--level
integrated development environment for the numerical solution of PDEs
based on the the Finite Element Method. The direct solver has been
thoroughly validated against available benchmark data from
\cite{bs06,Ghia1982}, and all details are reported in \cite{b11}.

\begin{figure}
\begin{center}
\includegraphics[width=0.5\textwidth]{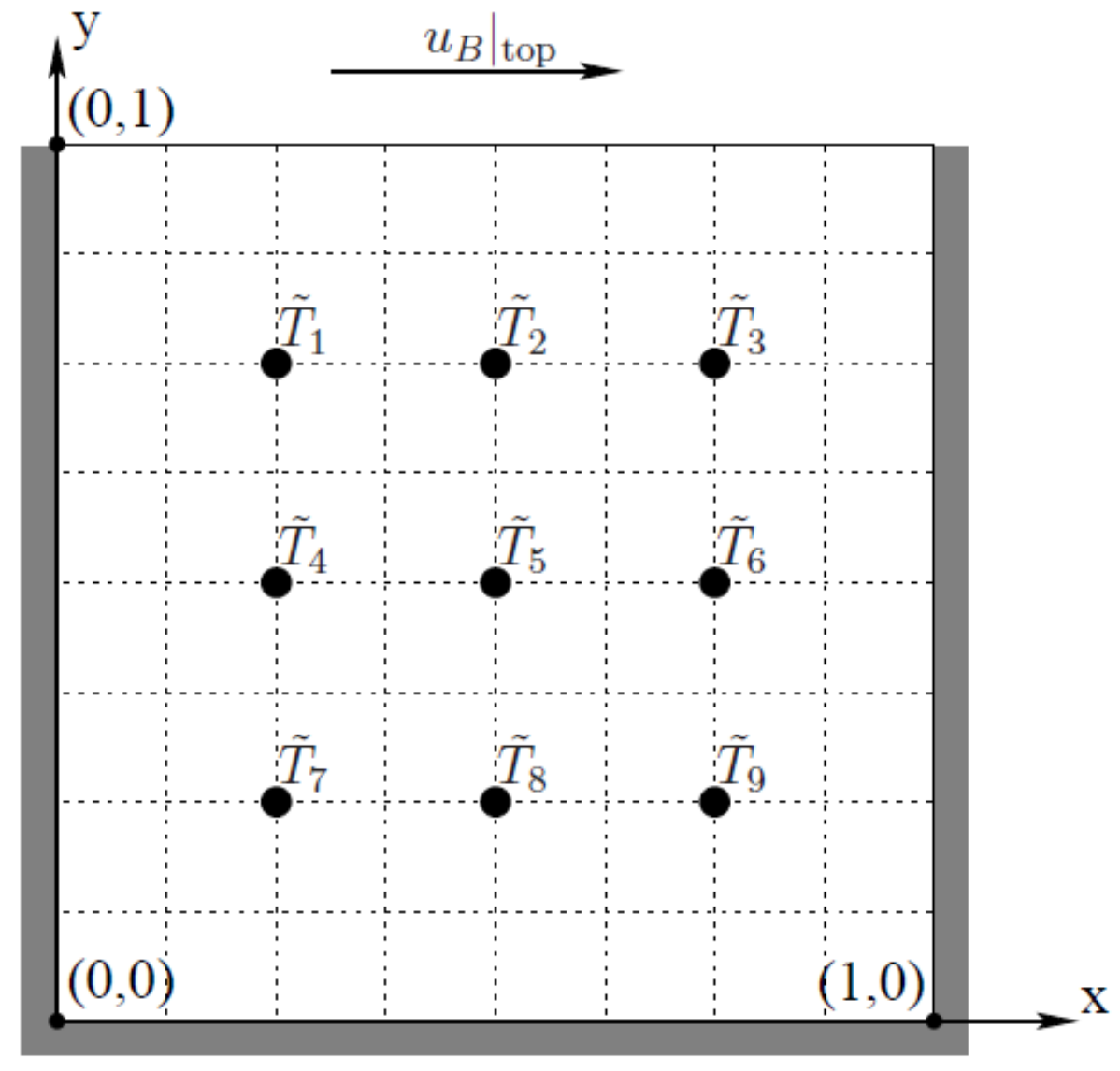}
\end{center}
\caption{Geometry of the 2D lid-driven (shear-driven) cavity.}
\label{fig:cavity}
\end{figure}

To solve numerically the direct problem we discretize system
\eqref{eq:coupled_PDEs}--\eqref{eq:coupled_BC} in time using a
second--order accurate semi-implicit approach. Spatial discretization
is carried out using triangular finite elements \eqref{eq:discr_trg}
and the P2 piecewise quadratic (continuous) representations for the
velocity $\u$ and the temperature $T$ fields, and the P1 piecewise
linear (continuous) representation for the pressure $p$ field.  The
system of algebraic equations obtained after such discretization is
solved at every time step with {\tt UMFPACK}, a solver for
nonsymmetric sparse linear systems \cite{umfpack}. We add that
incompressibility is ensured by an implicit treatment of equation
\eqref{eq:coupled_incompress}. Stability is enforced by choosing the
time step $\Delta t$ so that it satisfies the following CFL condition
\begin{equation} \left| \Delta t \left(\frac{\max_{\x \in \Omega} u(\x)}{h_x} +
      \frac{\max_{\x \in \Omega} v(\x)}{h_y} \right) \right| \leq 1.
\label{eq:CFL}
\end{equation}
The same technique is used for the numerical solution of adjoint
problem \eqref{eq:adjoint_coupled}--\eqref{eq:BC_adjoint}.

All our computations are performed using a 2D square domain $\Omega =
[0,1]^2$ shown in Figure \ref{fig:cavity}. Governing system
\eqref{eq:coupled_PDEs}--\eqref{eq:coupled_BC} and adjoint system
\eqref{eq:adjoint_coupled}--\eqref{eq:BC_adjoint} are discretized on a
uniform mesh with $N = N_x = N_y = 32$ grid points in every direction
using triangular finite elements combined with the cubic spline
interpolation of the function $\mu(T(\x))$. The rather modest spatial
resolution used is a consequence of the fact that in a single
reconstruction problem the governing and adjoint systems need to be
solved $\O(10^3 - 10^4)$ times, hence there are limitations of the
computational time. Unless stated otherwise, the interval $\L =
[100.0, 700.0]$ is discretized using an equispaced grid with $N_T =
600$ points. The actual constitutive relation $\tilde \mu(T)$ we seek
to reconstruct is given by Andrade law \eqref{eq:Andrade} with $C_1 =
0.001$ and $C_2 = 10^3$. In the computational tests reported below we
used $M = 9$ measurement points distributed uniformly inside the
cavity (Figure \ref{fig:cavity}). To mimic an actual experimental
procedure, first relation \eqref{eq:Andrade} is used in combination
with governing system \eqref{eq:coupled_PDEs}--\eqref{eq:coupled_BC}
to obtain pointwise temperature measurements
$\{\tilde{T}_i\}_{i=1}^M$. Relation \eqref{eq:Andrade} is then
``forgotten'' and is reconstructed using gradient--based algorithm
\eqref{eq:desc}. In terms of the initial guess in \eqref{eq:desc},
unless stated otherwise, we take a constant approximation $\mu_0$ to
\eqref{eq:Andrade}, given by $\mu_0 = \frac{1}{2} \left( \tilde \mu
  (T_{\alpha}) + \tilde \mu (T_{\beta}) \right) = \frac{C_1}{2} \left(
  e^{C_2/T_{\alpha}} + e^{C_2/T_{\beta}} \right)$ which translates
into the following expression for the new optimization variable
$\theta$, cf.~\eqref{eq:slack}, $\theta_0 = \sqrt{\mu_0 - m_{\mu}}$,
where $m_{\mu} = \frac{1}{2} \tilde \mu (T_{\beta}) = \frac{C_1}{2}
e^{C_2/T_{\beta}}$.  Since in the present problem the viscosity
$\mu(T)$ is a function of the temperature, the Reynolds number is
defined locally (both in space and in time) and varies in the range
$Re = 0.05 \div 240$.

Unless stated otherwise, the boundary conditions for the temperature
are $T_B|_{\textrm{top}} = T_{\beta} = 500$ and $T_B|_{\textrm{else}}
= T_{\alpha} = 300$ which results in the identifiability region $\I =
[300.0,500.0]$. The velocity boundary conditions $\u_B = [u_B,v_B]^T$
are given by $u_{B}|_{\textrm{top}} = U_0 \cos(2\pi t)$, $U_0 = 1$ and
$v_{B}|_{\textrm{top}} = 0$ on the top boundary segment and
$\u_B|_{\textrm{else}} = \0$ on the remaining boundary segments. Their
time--dependent character ensures that the obtained flow is unsteady
at the values of the Reynolds number for which self--sustained
oscillations do not spontaneously occur (the study of higher Reynolds
numbers was restricted by the numerical resolution used, see comments
above). The initial conditions $\{\u_0, T_0 \}$ used in the
reconstruction problem correspond to a developed flow obtained at $t =
10$ from the following initial and boundary conditions
$T_B|_{\textrm{top}} = 500$, $T_B|_{\textrm{else}} = 300$ and
$u_{B}|_{\textrm{top}} = 1$, $v_{B}|_{\textrm{top}} = 0$,
$\u_B|_{\textrm{else}} = \0$, $T_0 = 300$, $\u_0 = \0$. We emphasize
that adjoint system \eqref{eq:adjoint_coupled}--\eqref{eq:BC_adjoint}
is in fact a terminal--value problem which needs to be integrated
backwards in time, and its coefficients depend on the solution $\{\u,
T\}$ of the direct problem around which linearization is performed at
the given iteration. Our reconstructions are performed using the
following time windows $[0, t_f]$, $t_f = \left\{ \frac{1}{4},
  \frac{1}{2}, 1 \right\}$ which correspond to a fraction of, or a
full, forcing cycle in the boundary conditions described above. These
time windows are all discretized with the time step $\Delta t = 5
\cdot 10^{-3}$ in both the direct and adjoint problems.  This choice
of the time step $\Delta t$ ensures stability by satisfying the CFL
condition \eqref{eq:CFL}.

\subsection{Validation of Gradients}
\label{sec:kappa}

In this Section we present results demonstrating the consistency of
the cost functional gradients obtained with the approach described in
Section \ref{sec:grad_adj}. In Figure \ref{fig:grads} we present the
$L_2$ and several Sobolev $H^1$ gradients obtained at the first
iteration. In the first place, we observe that, as was
anticipated in Section \ref{sec:grad_adj}, the $L_2$ gradients indeed
exhibit quite irregular behaviour lacking necessary smoothness
which makes them unsuitable for the reconstruction of constitutive
relations with required properties, cf.~\eqref{eq:S_theta}. On the
other hand, the gradients extracted in the Sobolev space $H^1$ are
characterized by the required smoothness and therefore hereafter we
will solely use the Sobolev gradients. Next, in Figure
\ref{fig:kappa_test} we present the results of a diagnostic test
commonly employed to verify the correctness of the cost functional
gradient \cite{hnl02}. It consists in computing the directional
G\^ateaux differential $\J'(\theta;\theta')$ for some selected
perturbations $\theta'$ in two different ways, namely, using a
finite--difference approximation and using \eqref{eq:grad_T_coupled}
which is based on the adjoint field, and then examining the ratio of
the two quantities, i.e.,
\begin{equation}
\kappa (\epsilon) \triangleq \dfrac{\epsilon^{-1} \left[
\J(\theta + \epsilon \theta') - \J(\theta) \right]}
{\int_{-\infty}^{+\infty} \bnabla_{\theta} \J (s) \, \theta'(s) \, ds}
\label{eq:kappa}
\end{equation}
for a range of values of $\epsilon$. If the gradient $\bnabla_{\theta}
\J(\theta)$ is computed correctly, then for intermediate values of
$\epsilon$, $\kappa(\epsilon)$ will be close to the unity.
Remarkably, this behavior can be observed in Figure
\ref{fig:kappa_test} over a range of $\epsilon$ spanning about 6
orders of magnitude for three different perturbations $\theta'(T)$.
Furthermore, we also emphasize that refining the time step $\Delta t$
used in the time--discretization of
\eqref{eq:coupled_PDEs}--\eqref{eq:coupled_BC} and
\eqref{eq:adjoint_coupled}--\eqref{eq:BC_adjoint} yields values of
$\kappa (\epsilon)$ closer to the unity. The reason is that in the
``optimize--then--discretize'' paradigm adopted here such
refinement of discretization leads to a better approximation of the
continuous gradient \eqref{eq:grad_T_coupled}. We add that the
quantity $\log_{10} | \kappa (\epsilon) - 1 |$ plotted in Figure
\ref{fig:kappa_test}b shows how many significant digits of accuracy
are captured in a given gradient evaluation. As can be expected, the
quantity $\kappa(\epsilon)$ deviates from the unity for very small
values of $\epsilon$, which is due to the subtractive cancellation
(round--off) errors, and also for large values of $\epsilon$, which is
due to the truncation errors, both of which are well--known effects.

\begin{figure}
\begin{center}
\includegraphics[width=0.8\textwidth]{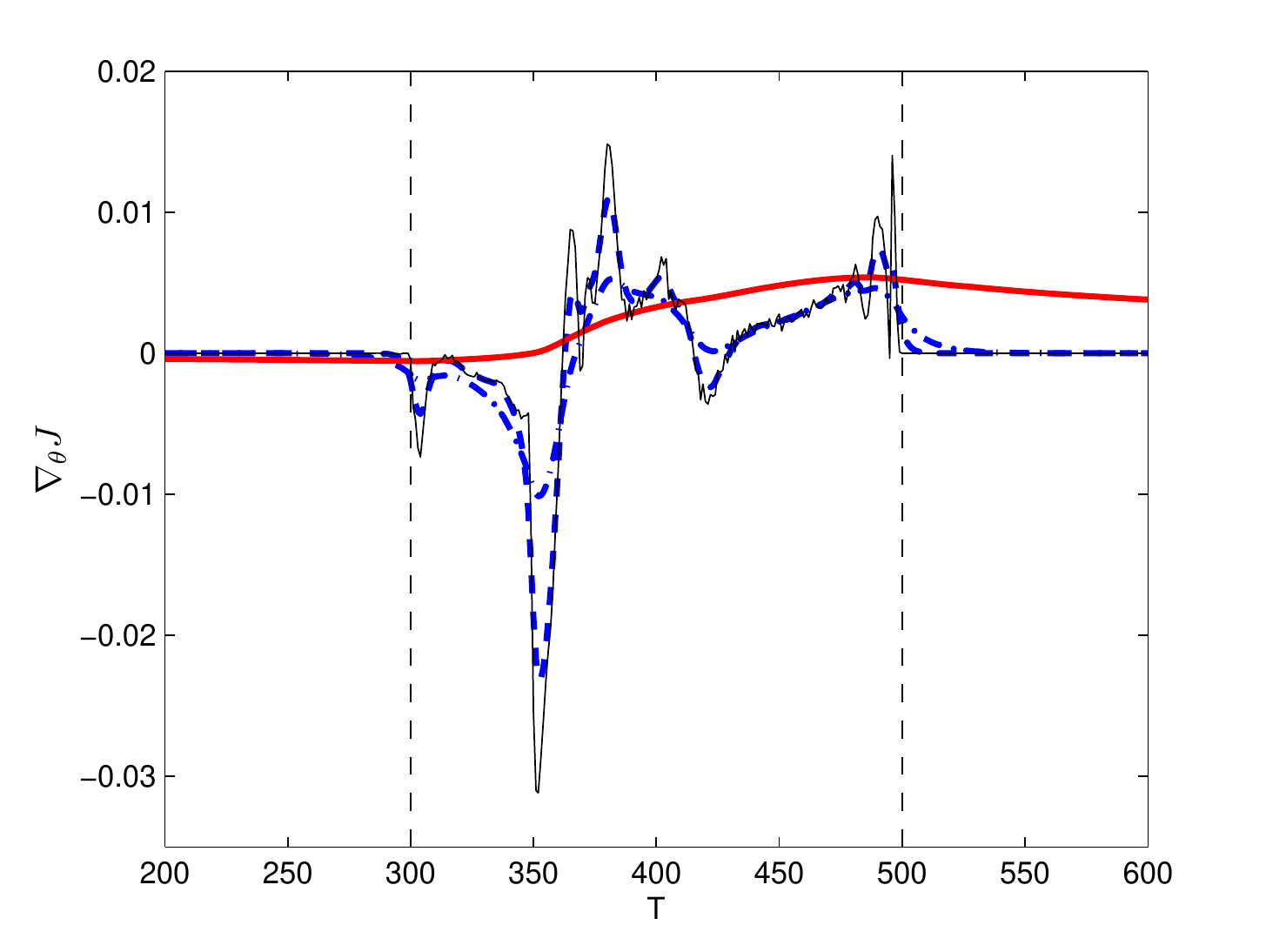}
\caption{Comparison of (thin solid line) the $L_2$ gradient
  $\bnabla^{L_2}_{\theta} \J$ and the Sobolev gradients
  $\bnabla^{H^1}_{\theta} \J$ defined in \eqref{eq:helm} for different
  values of the smoothing coefficient (thick dashed line) $\ell =
  2.5$, (thick dash--dotted line) $\ell = 10.0$ and (thick solid line)
  $\ell = 200.0$ at the first iteration with the initial guess
  $\mu_{0} = const = 0.0177$.  The vertical dashed lines represent the
  boundaries of the identifiability interval $\I$ and the vertical
  scale in the plot is arbitrary.}
\label{fig:grads}
\end{center}
\end{figure}

\begin{figure}
\centering
\mbox{
\subfigure[]{\includegraphics[width=0.5\textwidth]{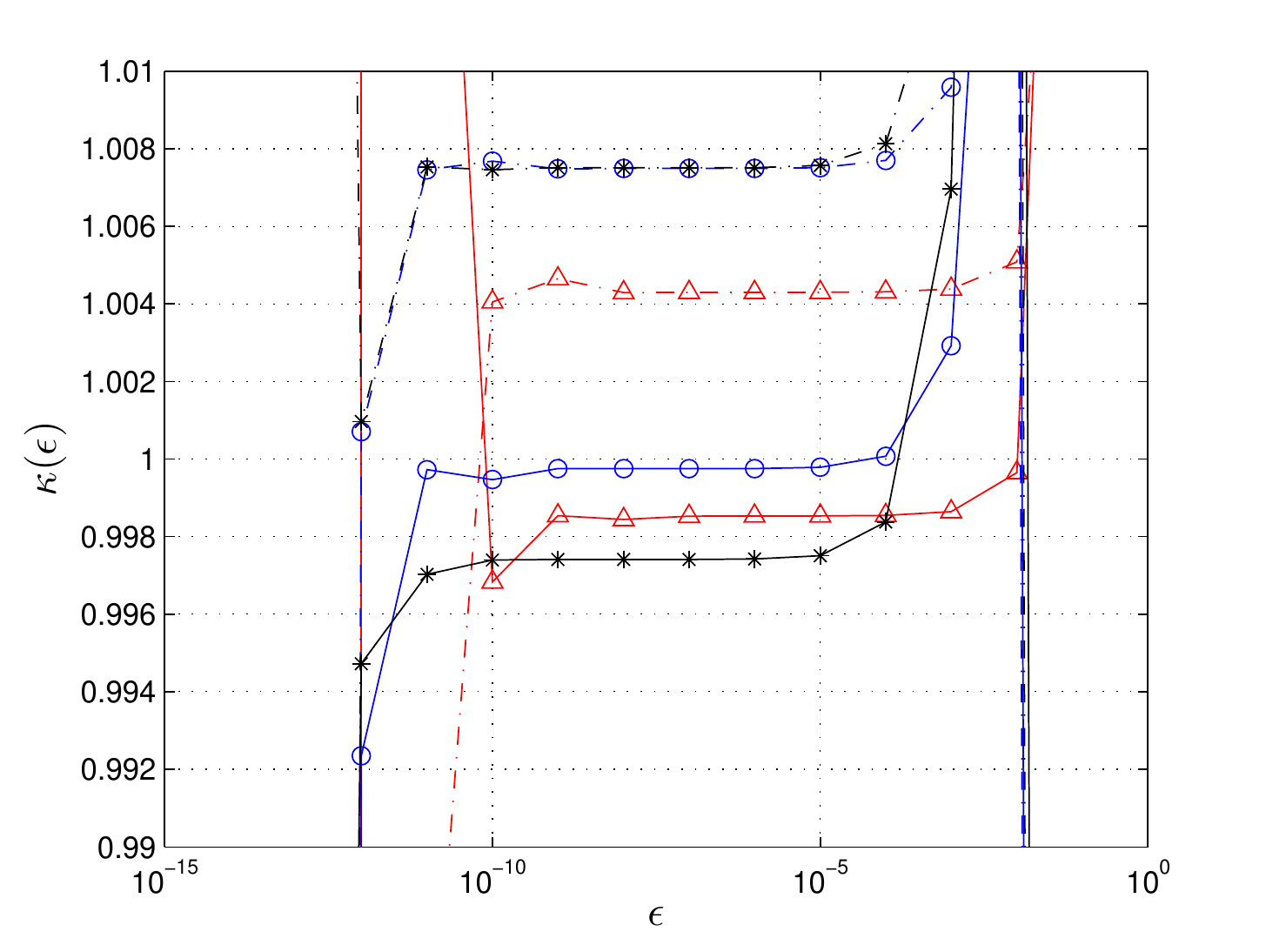}}
\subfigure[]{\includegraphics[width=0.5\textwidth]{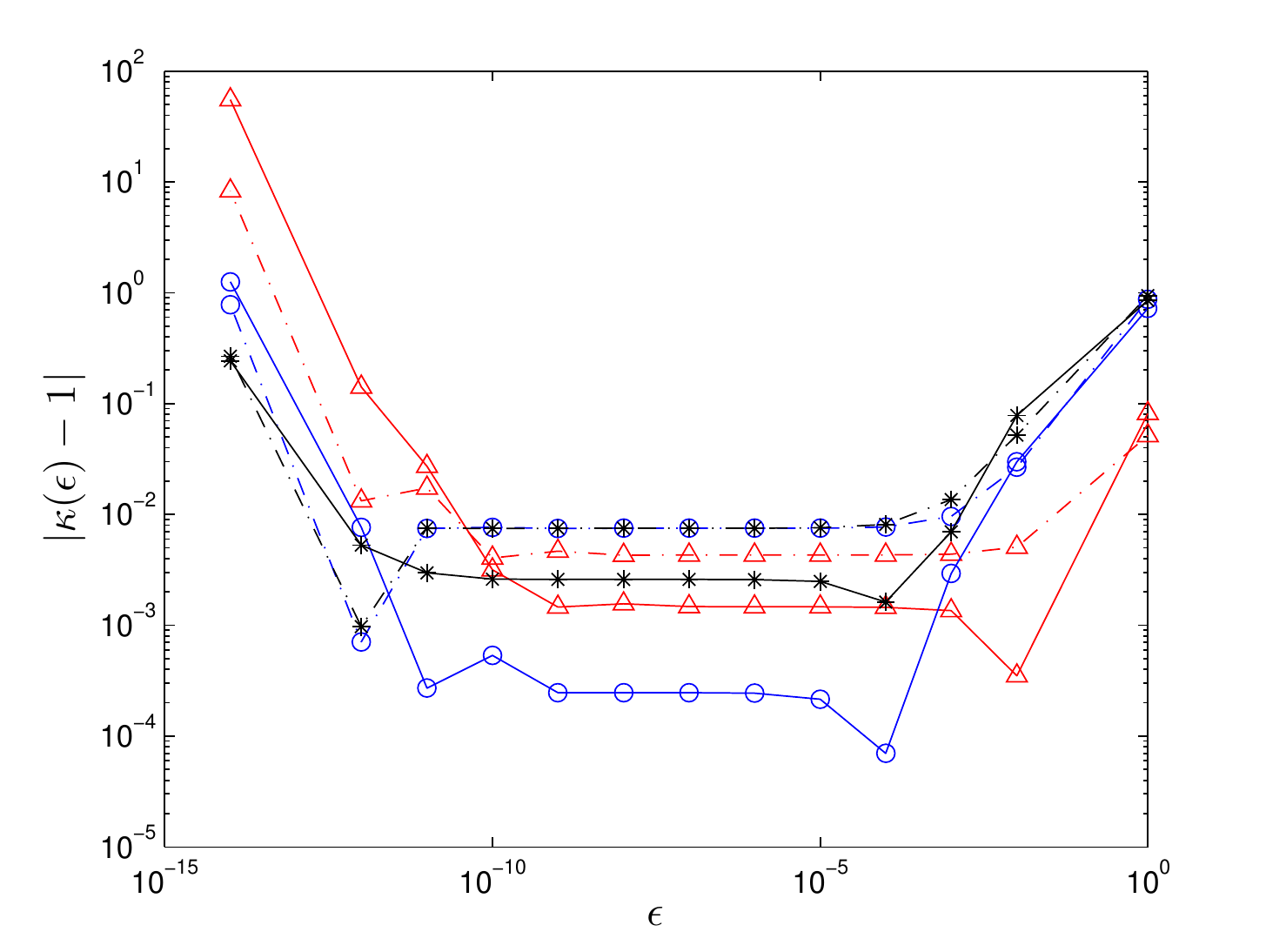}}}
\caption{ The behavior of (a) $\kappa (\epsilon)$ and (b) $\log_{10}
  |\kappa(\epsilon) -1 |$ as a function of $\epsilon$ for different
  perturbations (triangles) $\theta'(T) = \frac{10}{T}$,
  (circles) $\theta'(T) = e^{-\frac{T}{1000}}$ and
  (asterisks) $\theta'(T) = -\frac{T^2}{90000} + \frac{2T}{225} + \frac{2}{9}$.
  The time steps used in the time integration of
  \eqref{eq:coupled_PDEs}--\eqref{eq:coupled_BC} and
  \eqref{eq:adjoint_coupled}--\eqref{eq:BC_adjoint} are (dash--dotted
  line) $\Delta t = 5.0 \cdot 10^{-3}$ and (solid line) $\Delta t =
  5.0 \cdot 10^{-4}$.}
\label{fig:kappa_test}
\end{figure}

\subsection{Reconstruction Results}
\label{sec:estim}

We solve minimization problem \eqref{eq:minJslack} using the Steepest
Descent (SD), Conjugate Gradient (CG) and BFGS algorithms \cite{nw00}
and, unless indicated otherwise, using Sobolev gradients computed with
$\ell=200.0$ which was found by trial--and--error to result in the
fastest rate of convergence of iterations \eqref{eq:desc}. The
termination condition used was $\left| \frac{\J(\theta^{(n)})-
    \J(\theta^{(n-1)})}{\J(\theta^{(n-1)})} \right| < 10^{-6}$.  The
behavior of the cost functional $\J(\theta^{(n)})$ as a function of
the iteration count $n$ is shown in Figure \ref{fig:CF}a for all three
minimization algorithms (SD, CG and BFGS).  We note that in all cases
a decrease over several orders of magnitude is observed in just a few
iterations. Since the three descent methods tested reveal comparable
performance, our subsequent computations will be based on the steepest
descent method as the simplest one of the three. The effect of the
different initial guesses $\mu_0$ on the decrease of the cost
functional is illustrated in Figure \ref{fig:CF}b. Again, a
substantial decrease of the cost functional corresponding to about
5--6 orders of magnitude is observed for all the initial guesses
tested. Reconstructions $\hat{\mu}(T)$ of the constitutive relation
obtained using the initial guess $\mu_0 = \frac{C_1}{2} \left(
  e^{C_2/T_{\alpha}} + e^{C_2/T_{\beta}} \right) = 0.0177$ and the
optimization time windows with $t_f = \frac{1}{4}, \frac{1}{2}, 1$ are
shown in Figure \ref{fig:base_T-4}. Comparing the accuracy of the
reconstruction obtained for these different time windows, we can
conclude that better results are achieved on shorter time windows $t_f
= \frac{1}{4}, \frac{1}{2}$. Given considerations of the computational
time, hereafter we will therefore focus on the case with $t_f =
\frac{1}{4}$. In Figure \ref{fig:new_guess} we show the
reconstructions $\hat{\mu}(T)$ of the constitutive relation obtained
from different initial guesses already tested in Figure \ref{fig:CF}b
such as constant values of $\mu_0$, $\mu_0(T)$ varying linearly with
the temperature $T$ and $\mu_0$ given as a rescaling of the true
relationship $\tilde {\mu}$. As expected, the best results are
obtained in the cases where some prior information about the true
material property is already encoded in the initial guess $\mu_0$,
such as the slope, cf.~Figure \ref{fig:new_guess}c, or the exponent,
cf.~Figure \ref{fig:new_guess}d. We may also conclude that, since all
the reconstructions shown in Figures \ref{fig:base_T-4} and
\ref{fig:new_guess} are rather different, the iterations starting from
different initial guesses converge in fact to different local
minimizers. However, it should be emphasized that in all cases the
reconstructions do capture the main qualitative features of the actual
material property with differences becoming apparent only outside the
measurement span interval $\M$. In order to make our tests more
challenging, in the subsequent computations we will use the initial
guess $\mu_0 = \frac{1}{2} \left(\tilde \mu(T_{\alpha}) + \tilde
  \mu(T_{\beta})\right) = 0.0177$ (cf.~Figure \ref{fig:base_T-4})
which contains little prior information about the true material
property.

In Figure \ref{fig:shift_combo} we show the results of the
reconstruction performed with a larger identifiability region $\I =
[250.0, 700.0] = \L$ which is done by adopting the corresponding
values for temperature boundary conditions \eqref{eq:BC_T}, i.e.,
$T_B|_{\textrm{else}}=250$ and $T_B|_{\textrm{top}}=700$. We note that
in this case the target interval for the reconstruction $\L$ has been
chosen to coincide with identifiability interval $\I$. Fairly accurate
reconstruction can be observed in this problem as well, and we
emphasize that this is also the case for values of the temperature
outside the identifiability interval studied in the previous case.
This example demonstrates that accurate reconstruction on different
intervals $\L$ can in fact be achieved by adjusting the
identifiability region via a suitable choice of temperature boundary
conditions \eqref{eq:BC_T}. This process can be interpreted as
adjusting the conditions of the actual experiment used to obtain the
measurements $\{\tilde T_i \}_{i=1}^M$.

\begin{figure}
\begin{center}
\mbox{
\subfigure[]{\includegraphics[width=0.5\textwidth]{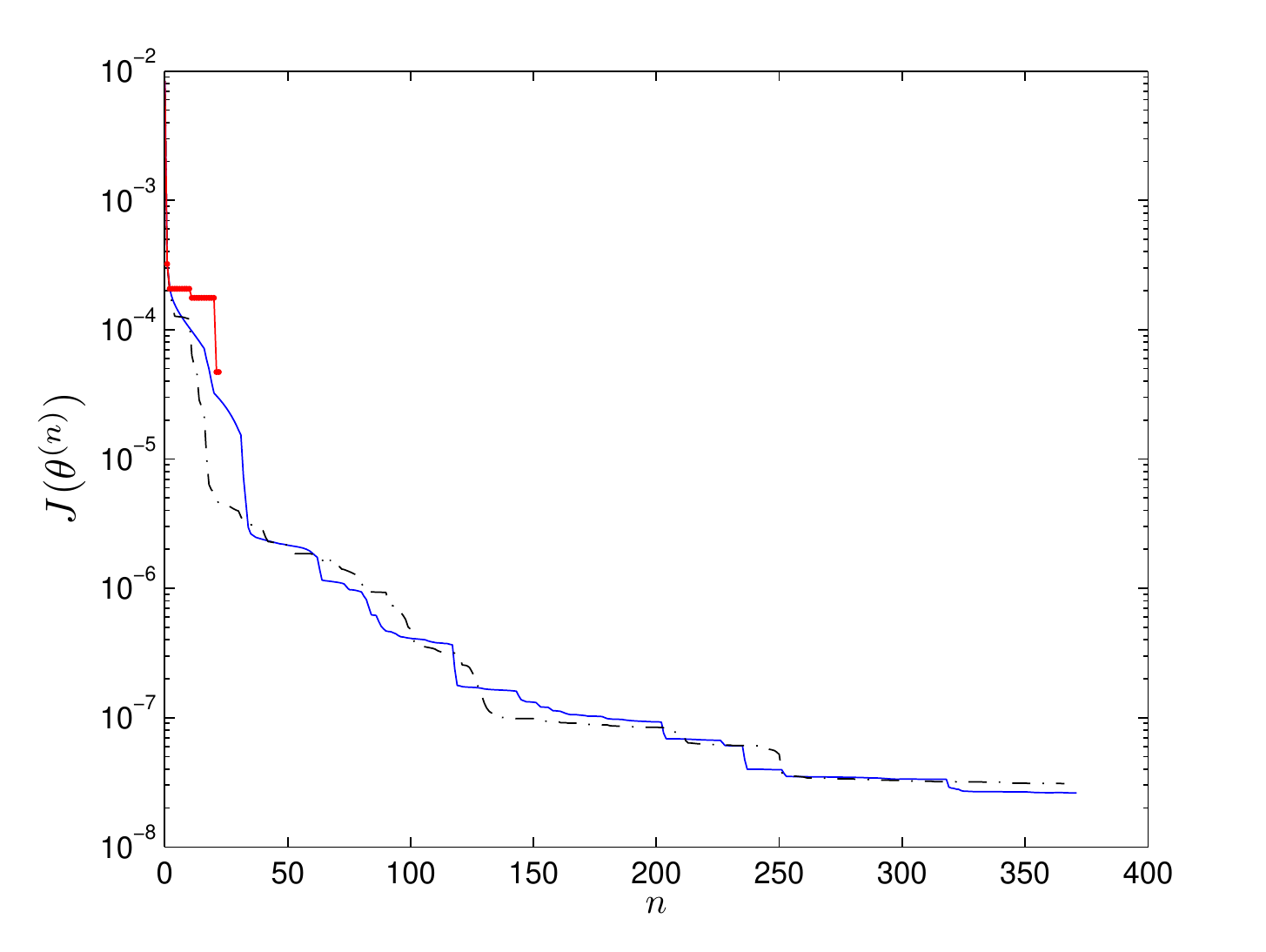}}
\subfigure[]{\includegraphics[width=0.5\textwidth]{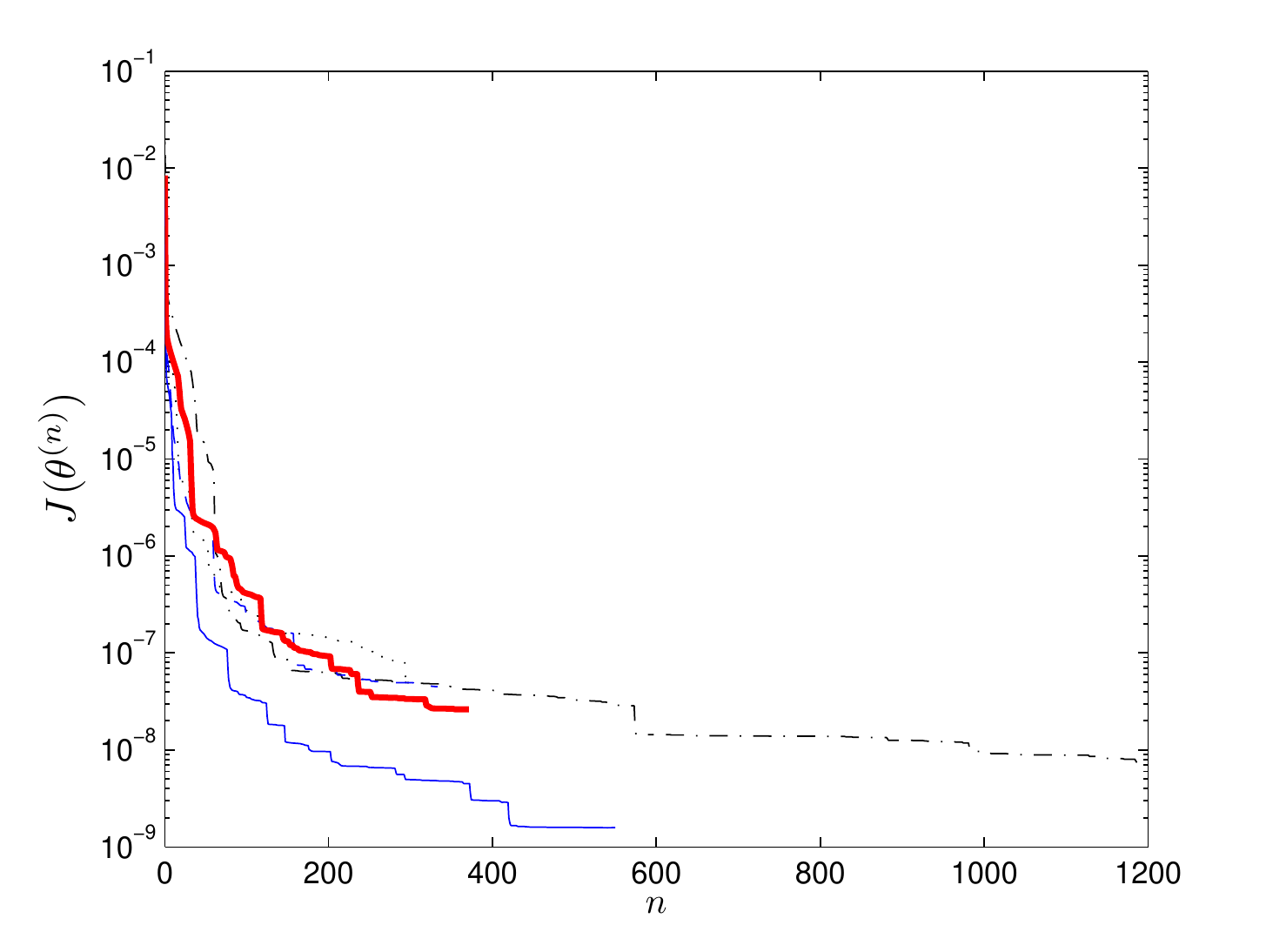}}}
\caption{(a) Decrease of the cost functional $\J(\theta^{(n)})$ with
  iterations $n$ using the Sobolev gradient $\bnabla^{H^1}_{\theta}
  \J$ defined in \eqref{eq:helm} and obtained with (solid line)
  steepest descent, (dash dotted line) conjugate gradient and (line
  with dots) BFGS methods with initial guess $\mu_{0} = 0.0177$.  (b)
  Decrease of the cost functional $\J(\theta^{(n)})$ with iterations
  $n$ for different initial guesses: (dots) $\mu_0 = \tilde
  {\mu}(T_{\alpha}) = 0.0280$, (dash--dotted line) $\mu_0 = \tilde
  {\mu}(T_b) = 0.0042$, (dashed line) $\mu_0(T)$ varying linearly
  between $\tilde {\mu}(T_{\alpha})$ and $\tilde{\mu}(T_{\beta})$,
  (thin solid line) $\mu_0 = \frac{1}{2} \tilde{\mu}(T)$ and (thick
  solid line) $\mu_{0} = 0.0177$.}
\label{fig:CF}
\end{center}
\end{figure}

\begin{figure}
\begin{center}
\mbox{
\subfigure[]{\includegraphics[width=0.5\textwidth]{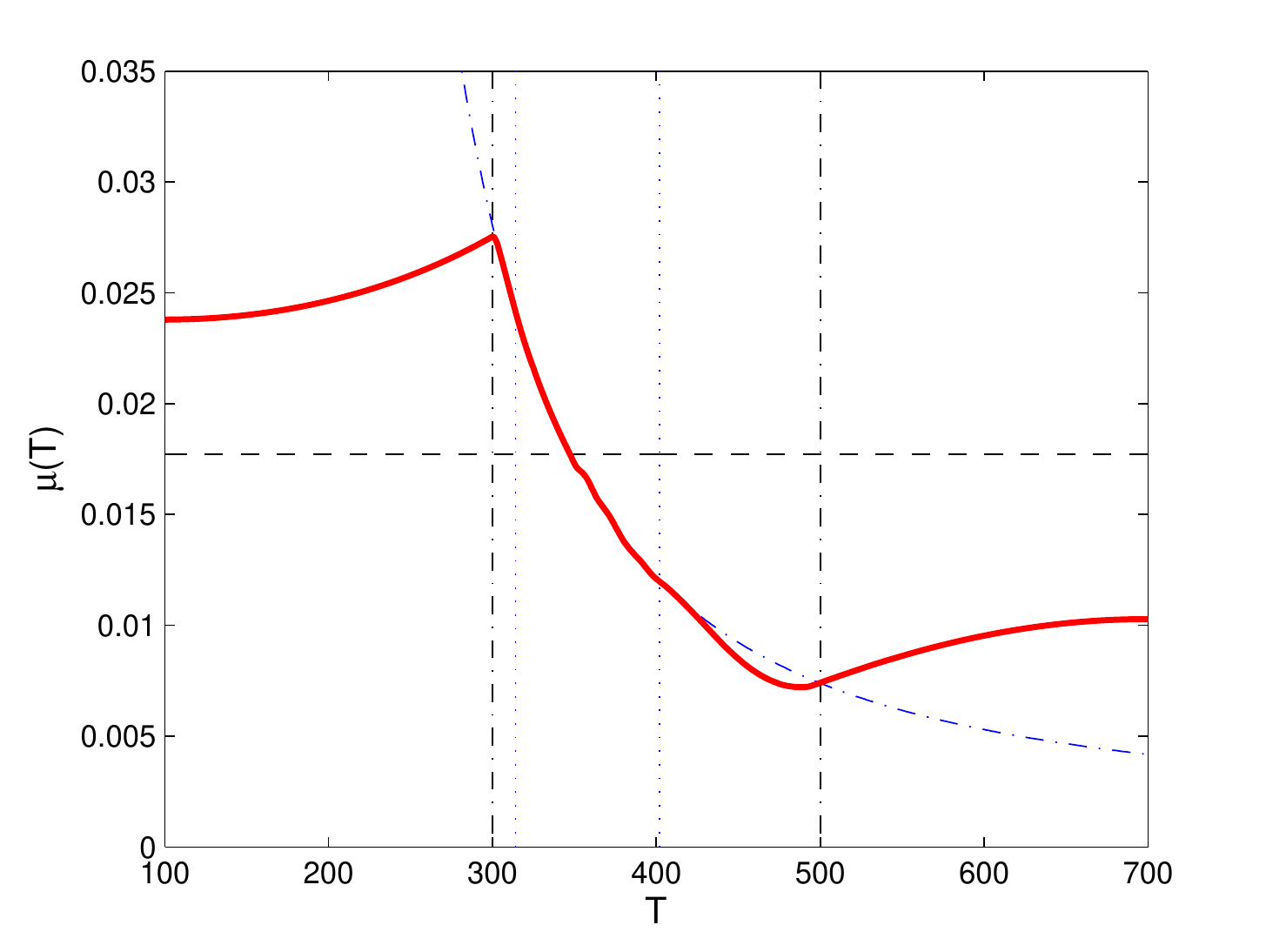}}
\subfigure[]{\includegraphics[width=0.5\textwidth]{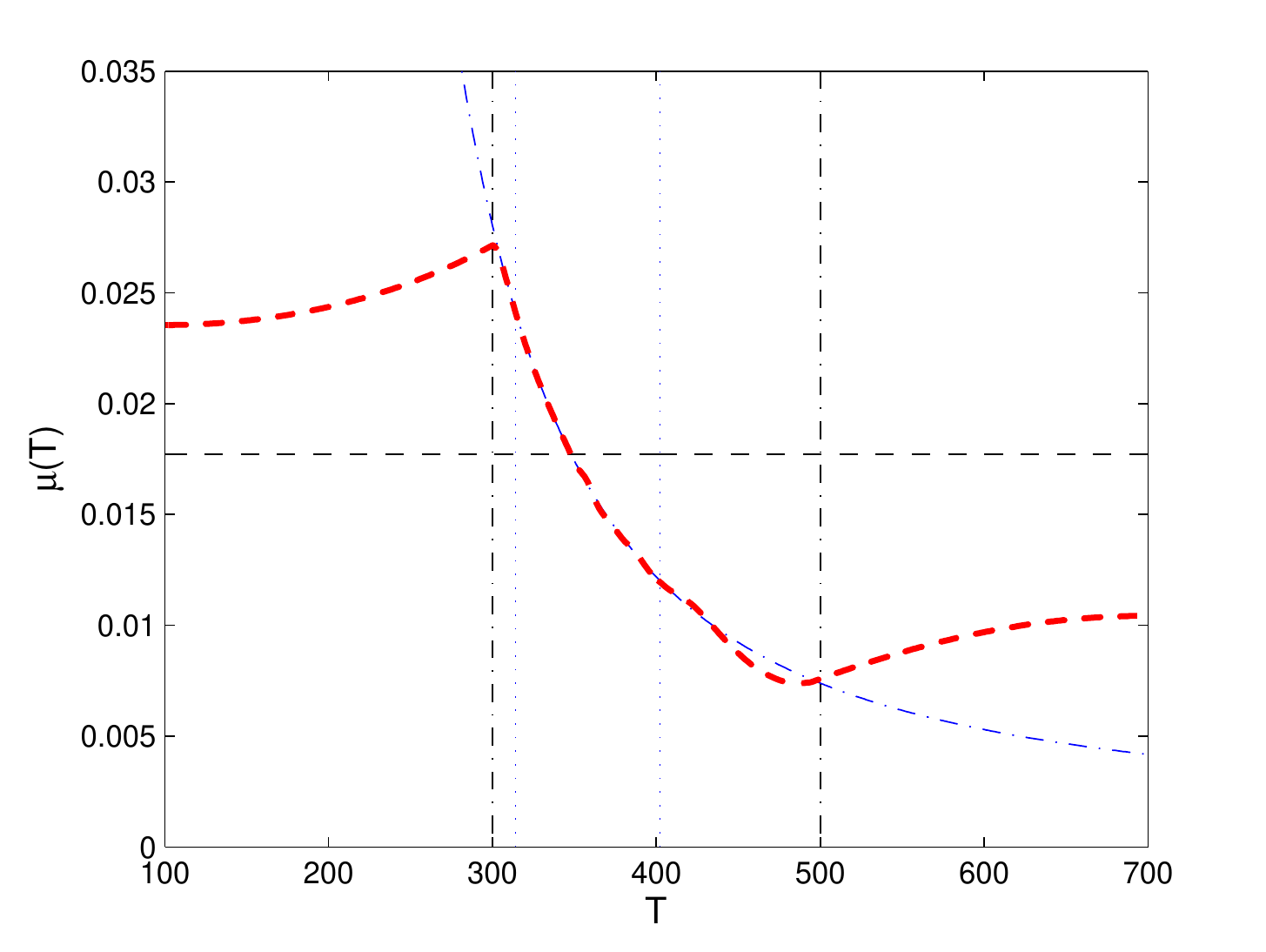}}}
\mbox{
\subfigure[]{\includegraphics[width=0.5\textwidth]{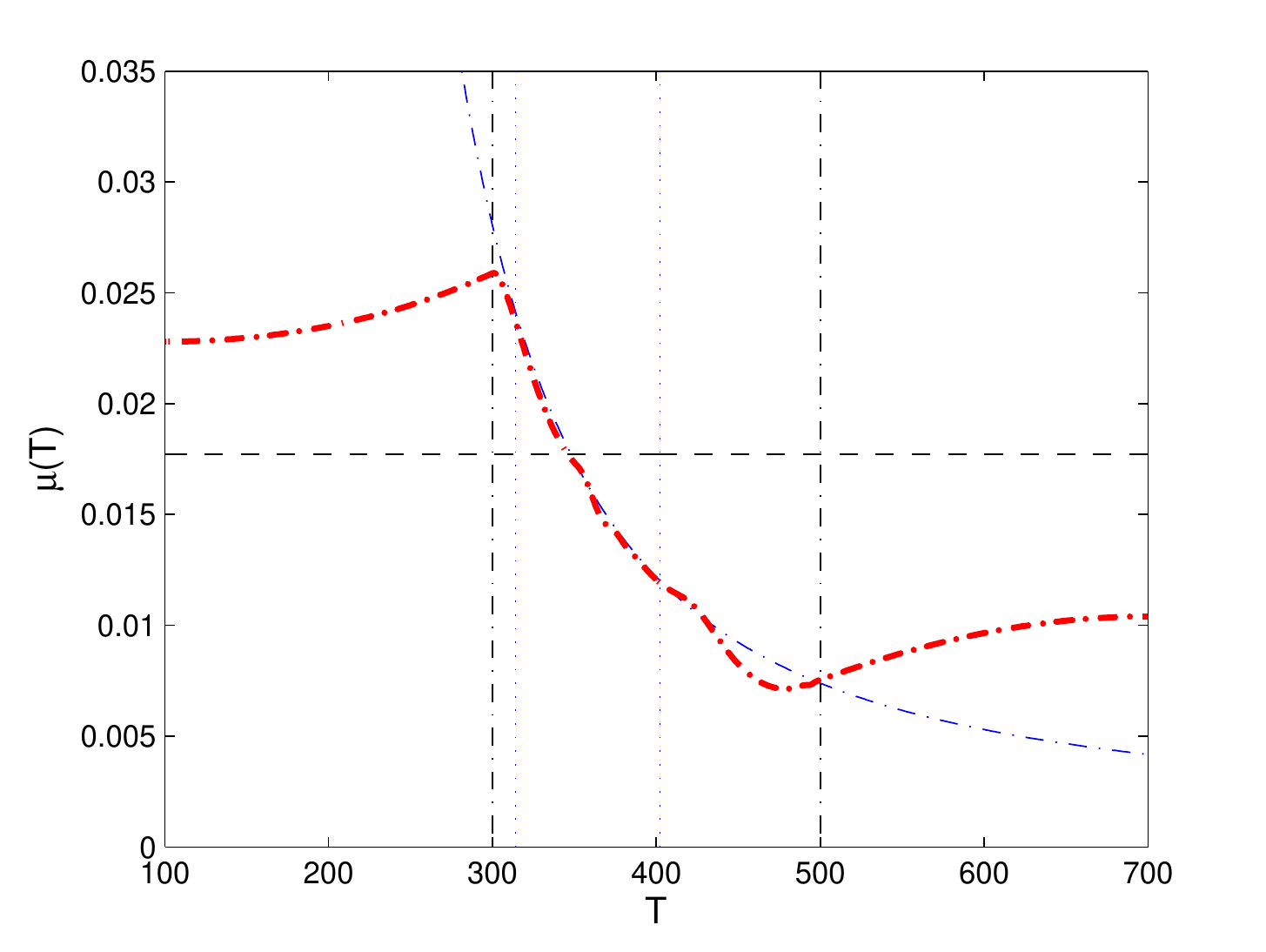}}
\subfigure[]{\includegraphics[width=0.5\textwidth]{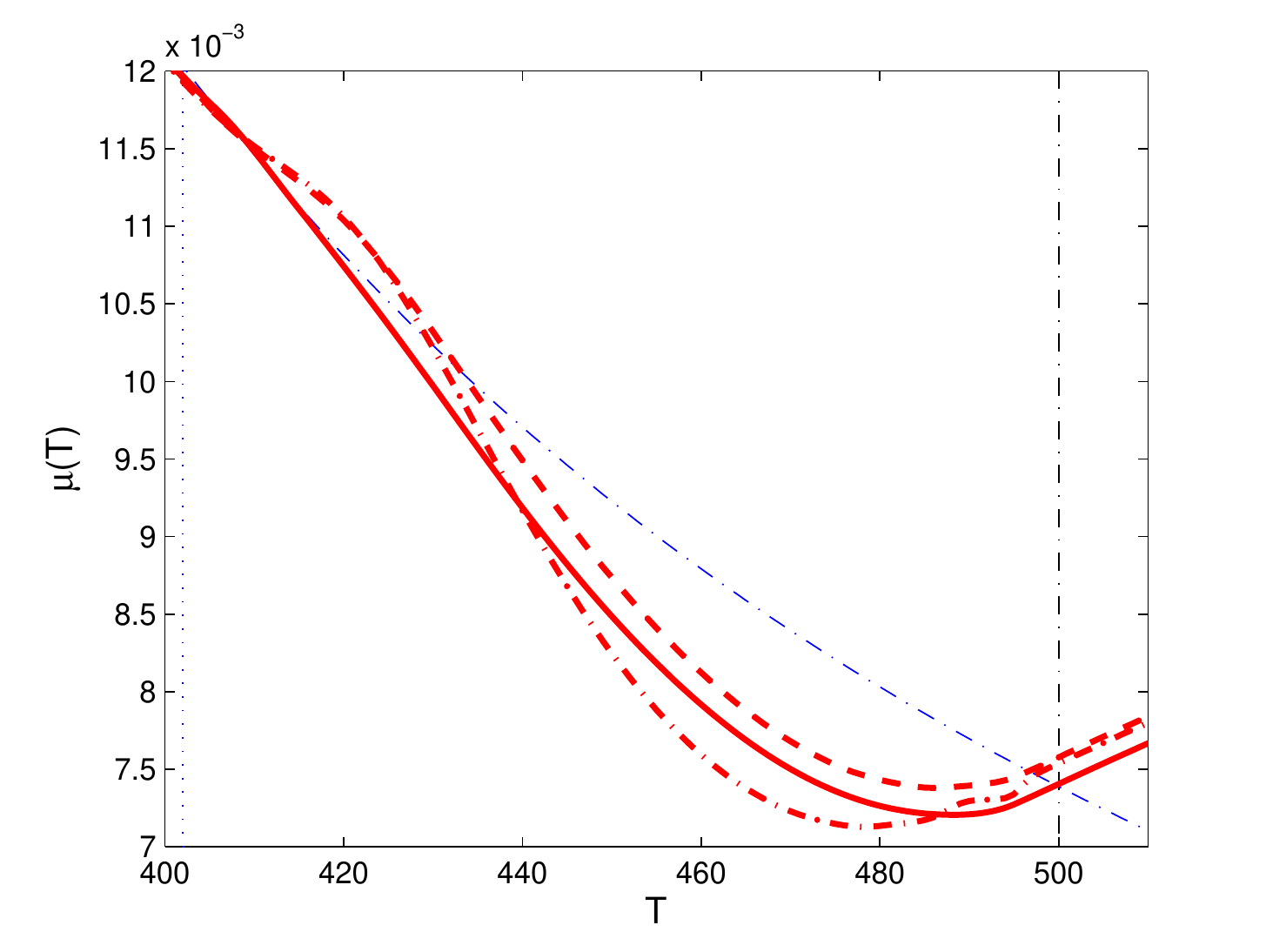}}}
\caption{Reconstruction $\hat{\mu}(T)$ of the material property obtained using
the Sobolev gradients defined in \eqref{eq:helm} on (a,b,c) the
interval $\L$ and (d) close--up view showing the interval outside the
identifiability region $\I$ with the time window $[0, t_f]$, where (a)
$t_f = \frac{1}{4}$ (b) $t_f = \frac{1}{2}$ and (c) $t_f = 1$. The
dash--dotted line represents the true material property
\eqref{eq:Andrade}, the thick solid, dashed and dash dotted lines are
the reconstructions for (a,d) $t_f = \frac{1}{4}$, (b,d) $t_f =
\frac{1}{2}$ and (c,d) $t_f = 1$, respectively, whereas the
dashed line represents the initial guess $\mu_{0} = 0.0177$; the
vertical dash--dotted and dotted lines represent, respectively, the
boundaries of the identifiability interval $\I$ and the measurement span
$\M$.}
\label{fig:base_T-4}
\end{center}
\end{figure}

\begin{figure}
\begin{center}
\mbox{
\subfigure[]{\includegraphics[width=0.5\textwidth]{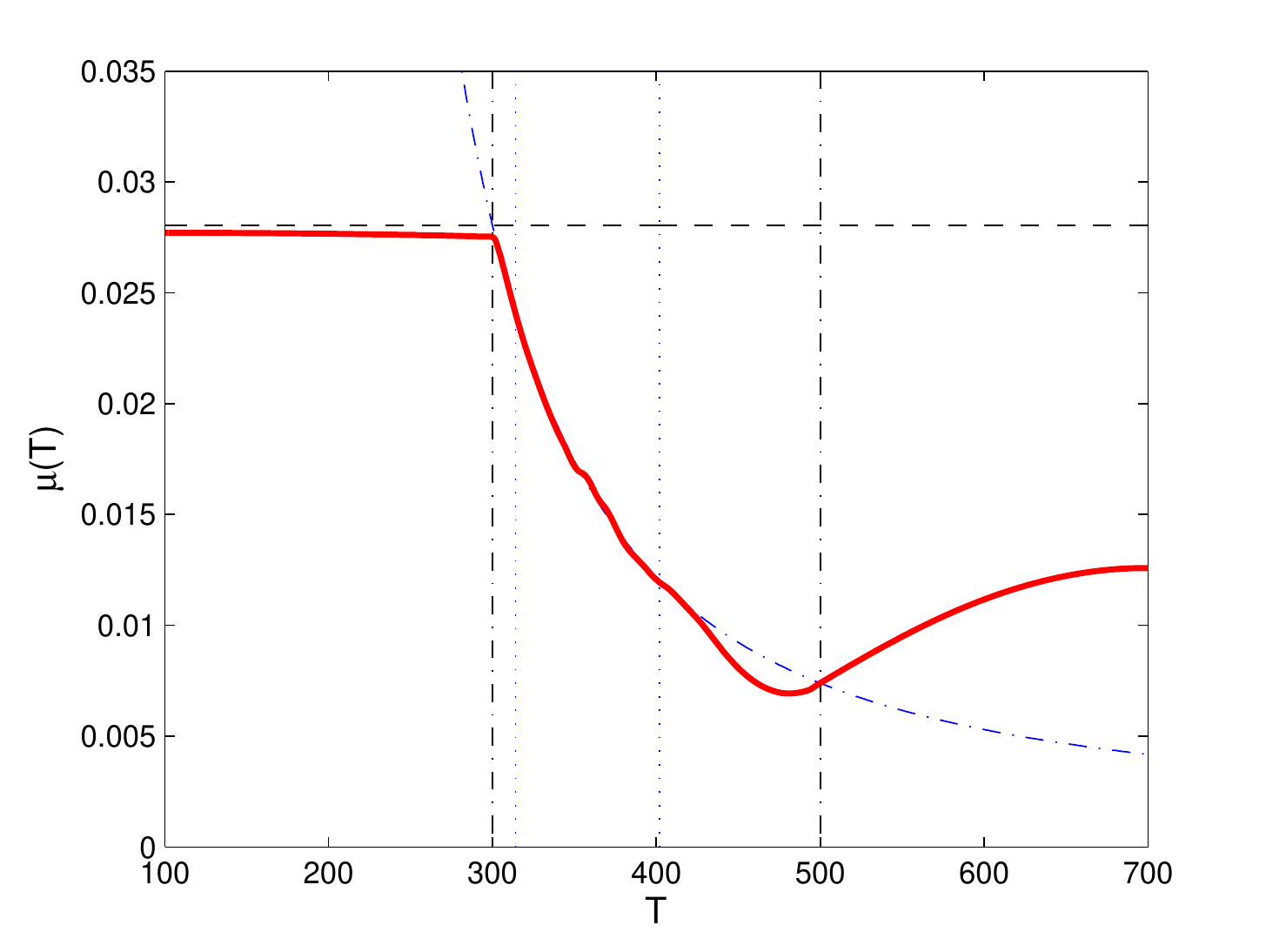}}
\subfigure[]{\includegraphics[width=0.5\textwidth]{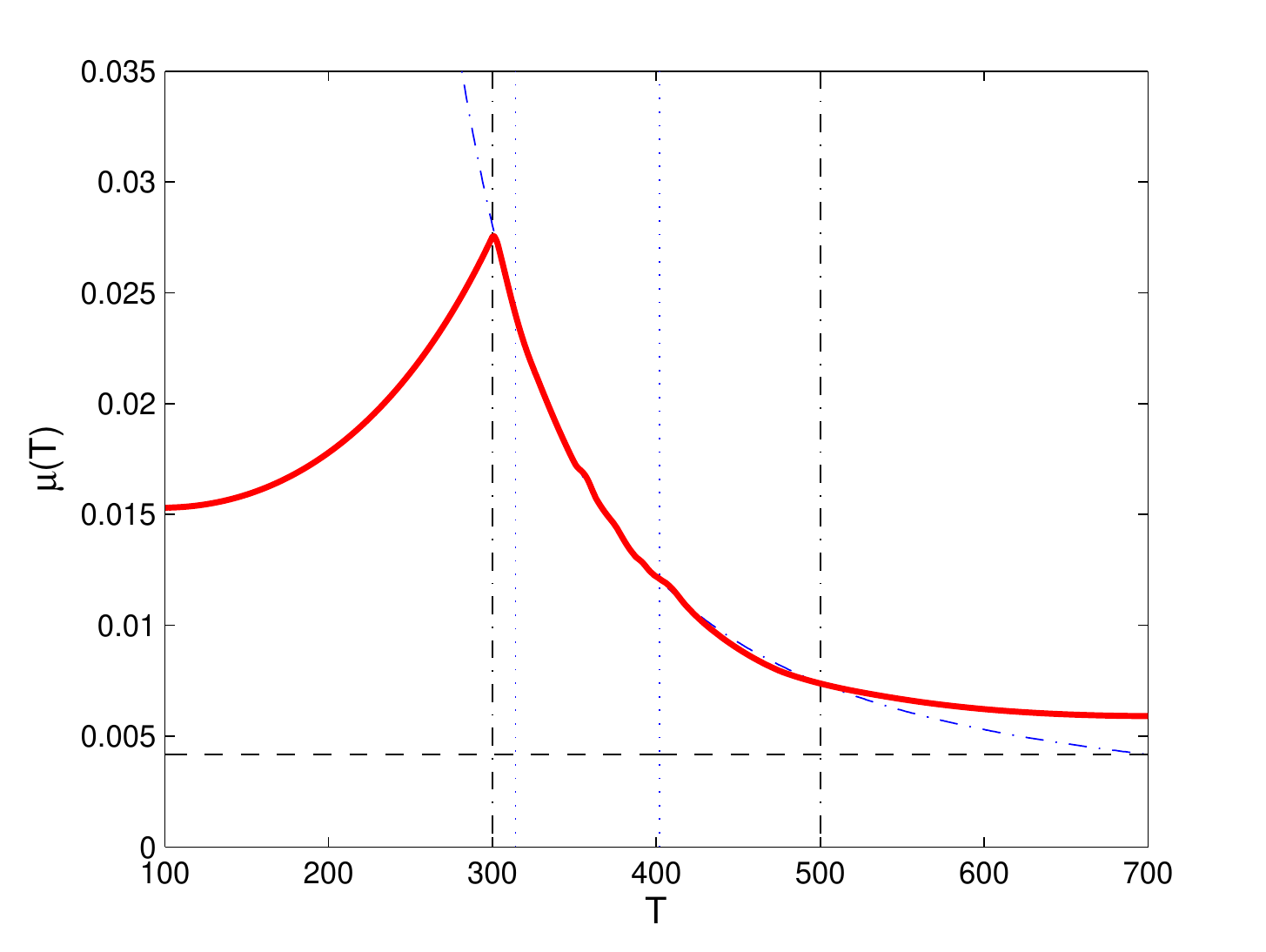}}}
\mbox{
\subfigure[]{\includegraphics[width=0.5\textwidth]{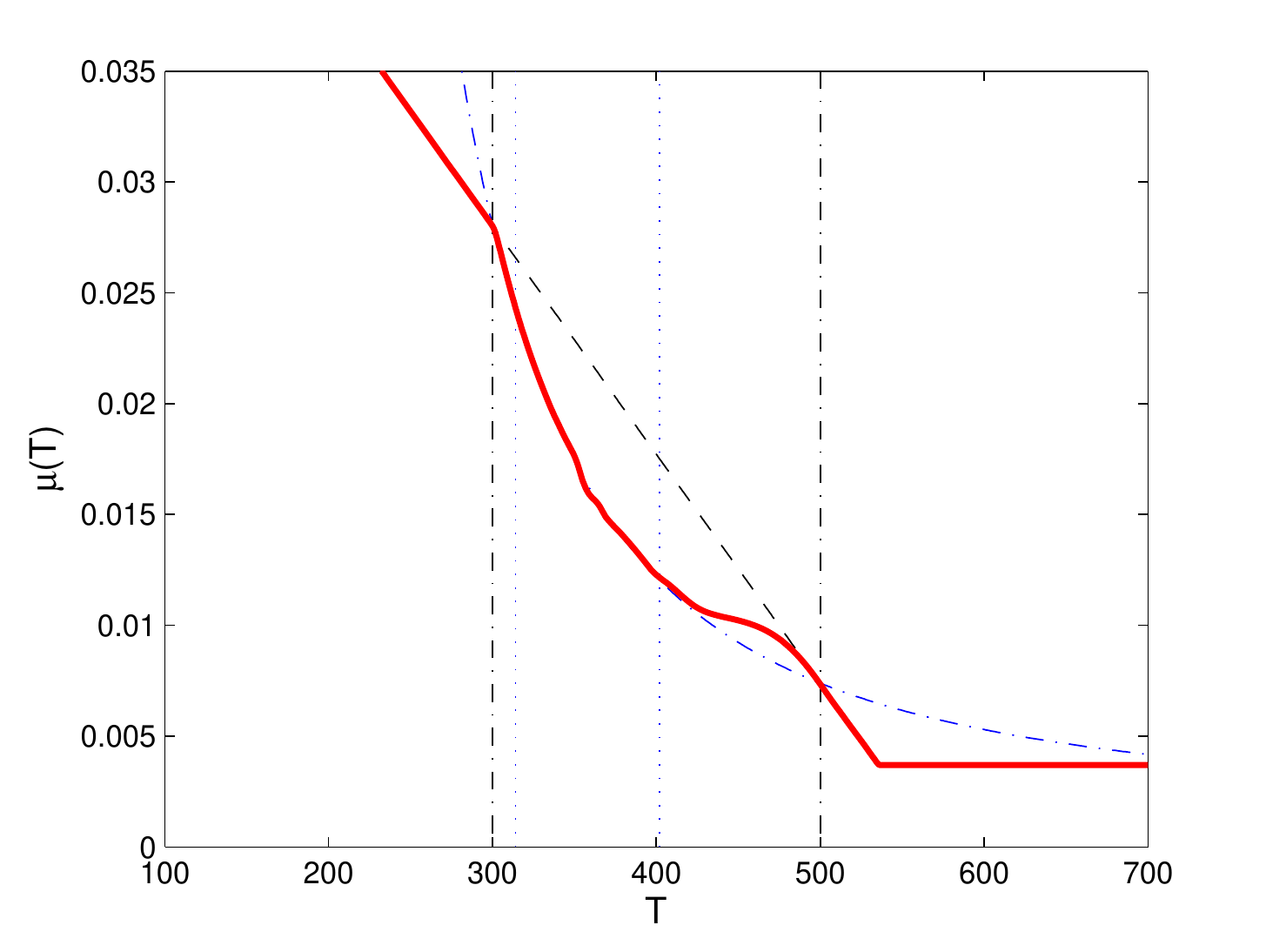}}
\subfigure[]{\includegraphics[width=0.5\textwidth]{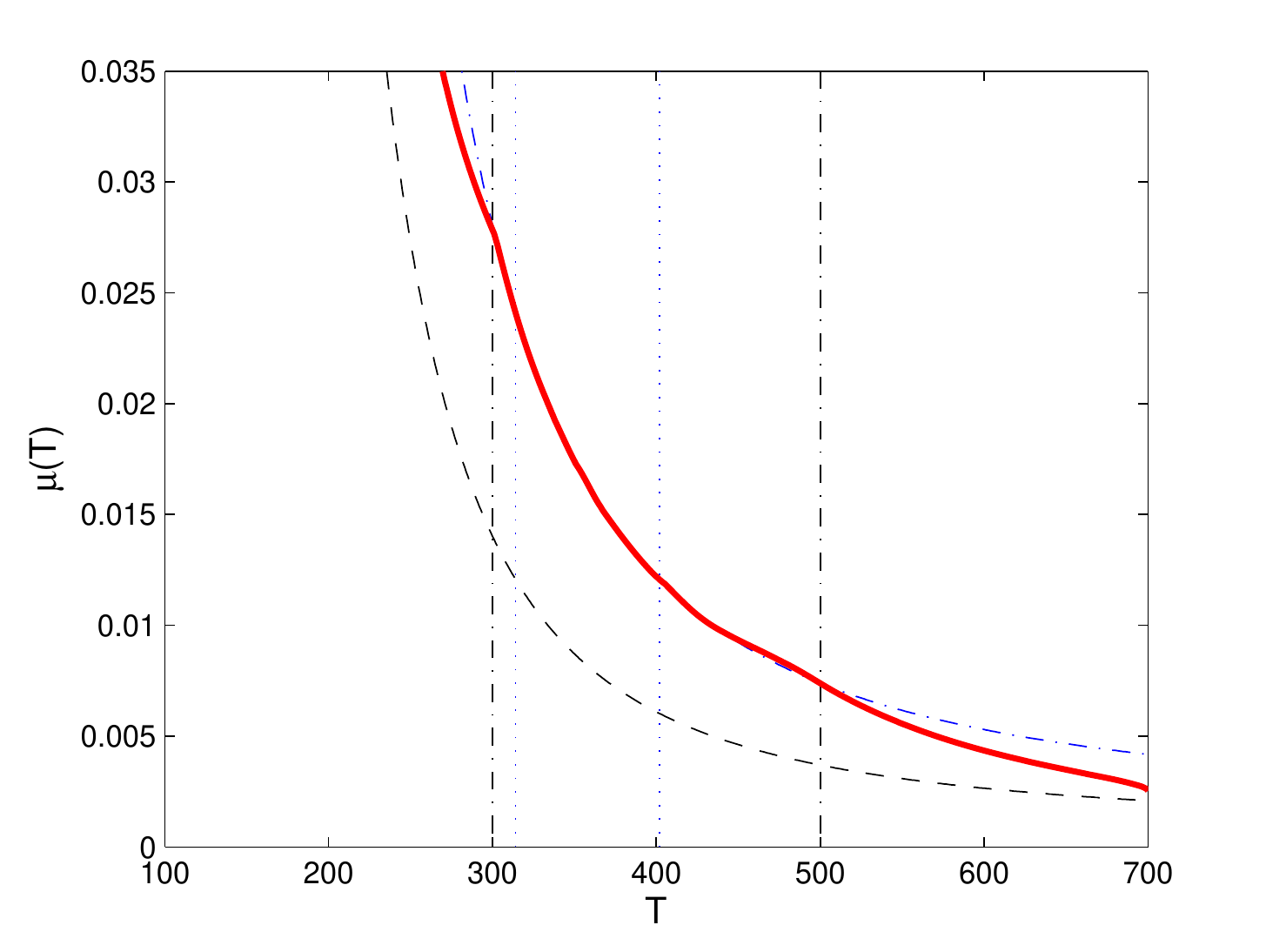}}}
\caption{Reconstruction $\hat{\mu}(T)$ of the material property
  obtained using different initial guesses (a) $\mu_0 = \tilde
  {\mu}(T_{\alpha}) = 0.0280$, (b) $\mu_0 = \tilde {\mu}(T_b) =
  0.0042$, (c) $\mu_0(T)$ varying linearly between $\tilde
  {\mu}(T_{\alpha})$ and $\tilde{\mu}(T_{\beta})$ and (d) $\mu_0 =
  \frac{1}{2} \tilde {\mu}(T)$, and the Sobolev gradients defined in
  \eqref{eq:helm} on the interval $\L$.  The dash--dotted line
  represents the true material property \eqref{eq:Andrade}, the solid
  line is the reconstruction $\hat{\mu}(T)$, whereas the dashed line
  represents the initial guess $\mu_{0}$; the vertical dash--dotted
  and dotted lines represent, respectively, the boundaries of the
  identifiability interval $\I$ and the measurement span $\M$.}
\label{fig:new_guess}
\end{center}
\end{figure}

\begin{figure}
\begin{center}
\mbox{
\subfigure[]{\includegraphics[width=0.5\textwidth]{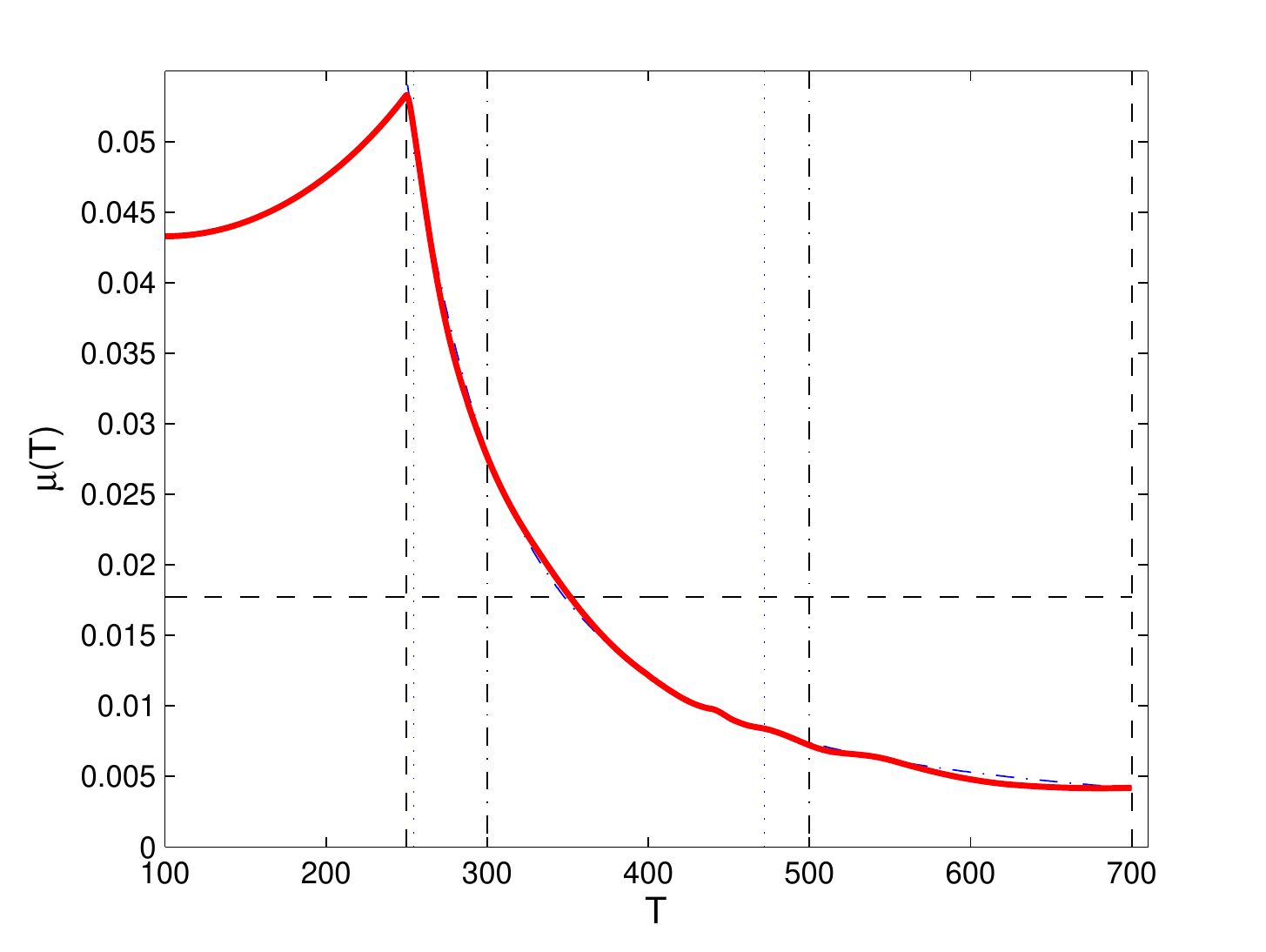}}
\subfigure[]{\includegraphics[width=0.5\textwidth]{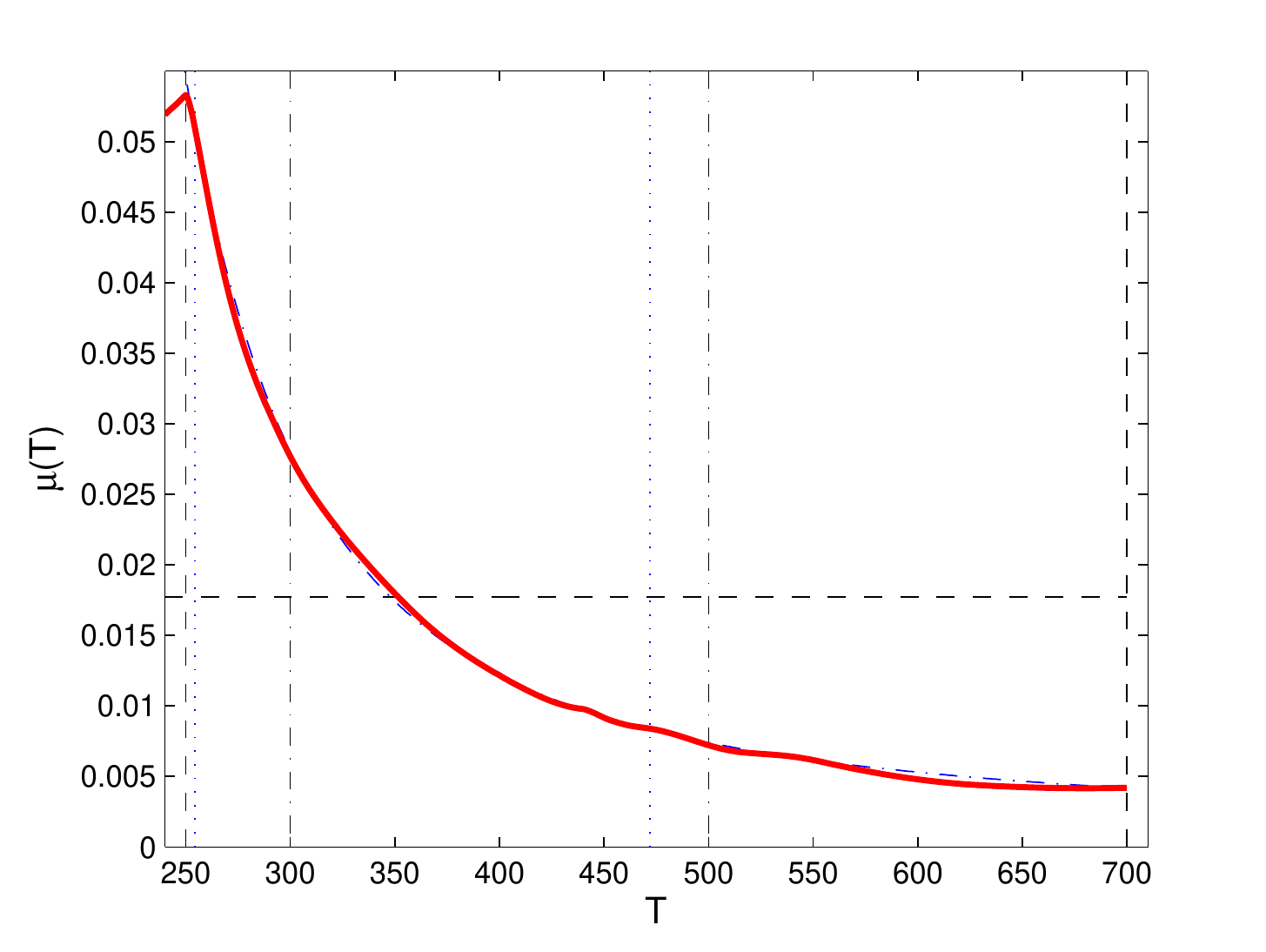}}}
\caption[]{Reconstruction $\hat{\mu}(T)$ of the material property using an extended
identifiability region $\I = [250.0, 700.0] = \L$ shown on (a) the interval
$[100.0, 700.0]$ and (b) magnification of this new identifiability region. The
dash--dotted line represents the true material property
\eqref{eq:Andrade}, the solid line is the reconstruction
$\hat{\mu}(T)$, whereas the dashed line represents the initial guess
$\mu_{0}$; the vertical dash--dotted and dotted lines represent,
respectively, the boundaries of the identifiability interval
$\I_0 = [300.0, 500.0]$ used previously and the measurement span
$\M$, while the dashed vertical lines show the boundaries of the new identifiability
interval $\I$.}
\label{fig:shift_combo}
\end{center}
\end{figure}

\subsection{Reconstruction in the Presence of Noise}
\label{sec:noise}

In this Section we first assess the effect of noise on the
reconstruction without the Tikhonov regularization and then study the
efficiency of the regularization techniques introduced in Section
\ref{sec:reg}. In Figure \ref{fig:noise_Tikh}a,b we revisit the case
already presented in Figure \ref{fig:base_T-4}a (reconstruction on the
interval $\L = [100.0, 700.0]$ with the identifiability region $\I =
[300.0, 500.0]$), now for measurements contaminated with 0.05\%,
0.1\%, 0.3\%, 0.5\% and 1.0\% uniformly distributed noise and without
Tikhonov regularization. To incorporate noise, say of $\eta\%$, into
the measurements $\{\tilde T_i(t) \}_{i=1}^M$, we replace these
measurements at every discrete time step $t_j \in [0,t_f]$ with a new
set $\{\tilde T_i^{\eta}(t_j) \}_{i=1}^M$, where the independent
random variables $\tilde{T}_i^{\eta}(t_j)$ have a uniform distribution
with the mean $\tilde{T}_i(t_j)$ and the standard deviation
$\Delta\eta = \frac{1}{M} \sum_{i=1}^M \tilde{T}_i(t_j) \cdot
\frac{\eta}{100\%}$. Unless stated otherwise, in order to be able to
directly compare reconstructions from noisy measurements with
different noise levels, the same noise realization is used after
rescaling to the standard deviation $\Delta\eta$. As expected, in
Figure \ref{fig:noise_Tikh}a,b we see that increasing the level of
noise leads to oscillatory instabilities developing in the
reconstructed constitutive relations $\hat{\mu}(T)$. We note that the
reconstructions become poor already for relatively low noise levels,
i.e., on the order of $1\%$. One reason for this seems to be the
time-dependent nature of the problem in which independent noise is
added to the measurements at every (discrete) time instant leading to
accumulation of uncertainty. Indeed, such loss of information was not
observed in the case of the {\em steady} problem studied in
\cite{bvp10} where reliable reconstructions could be obtained with
noise levels an order of magnitude larger than in the present
investigation. In regard to the results shown in Figure
\ref{fig:noise_Tikh}a,b, we add that the pattern introduced by the
noise in the reconstructions depends on the specific noise sample used
(which was the same for all the reconstructions shown in the Figure).
Reconstructions performed using different realizations of the noise
produce distinct patterns in the reconstructed constitutive relations
$\hat{\mu}(T)$. In our computational experiments we also observe that
inclusion of noise in the measurements tends to replace the original
minimizers with perturbed ones (this is evident from the uniform, with
respect to $T$, convergence of the perturbed minimizers to the
noise-free reconstructions as the noise level goes to zero in Figure
\ref{fig:noise_Tikh}a,b).

The effect of the Tikhonov regularization is studied in Figure
\ref{fig:noise_Tikh}c,d, where we illustrate the performance of the
technique described in Section \ref{sec:reg}, cf.~\eqref{eq:regH1}, on
the reconstruction problem with 1.0\% noise in the measurement data
(i.e., the ``extreme'' case presented in Figures
\ref{fig:noise_Tikh}a,b). In terms of the (constant) reference
function we take $\overline{\theta} = \sqrt{\mu_0 - m_{\mu}}$, where
$\mu_0 = 0.0177$. We note that by increasing the values of the
regularization parameter $\lambda$ in \eqref{eq:regH1} from 0 (no
regularization) to 2500 we manage to eliminate the instabilities
caused by the presence of noise in the measurements and obtain as a
result smoother constitutive relations, cf.~Figure
\ref{fig:noise_Tikh}c,d. We add that, while after introducing the
Tikhonov regularization the reconstructed solutions converge in fact
to different local minimizers (in comparison with the reconstructions
without noise), this does not prevent the reconstructions from
capturing the main qualitative features of the actual material
property. Systematic methods for determining the optimal values of
regularization parameters are discussed for instance in \cite{ehn96}.
Finally, in Figure \ref{fig:noise_avr} we present the relative
reconstruction errors $\|\hat{\mu} - \tilde{\mu} \|_{L_1(\I)} \,/
\,\|\tilde{\mu}\|_{L_1(\I)}$ obtained using the approach described
earlier in Section \ref{sec:reg} for data with different noise levels
and averaged over 10 different noise samples. From Figure
\ref{fig:noise_avr} we conclude that larger values of the
regularization parameter $\lambda$ are required for more noisy
measurements. We close this Section by concluding, in agreement with
our earlier results reported in \cite{bvp10}, that Tikhonov
regularization performs as expected in problems with noise present in
the measurement data.

\begin{figure}
\begin{center}
\mbox{
\subfigure[]{\includegraphics[width=0.5\textwidth]{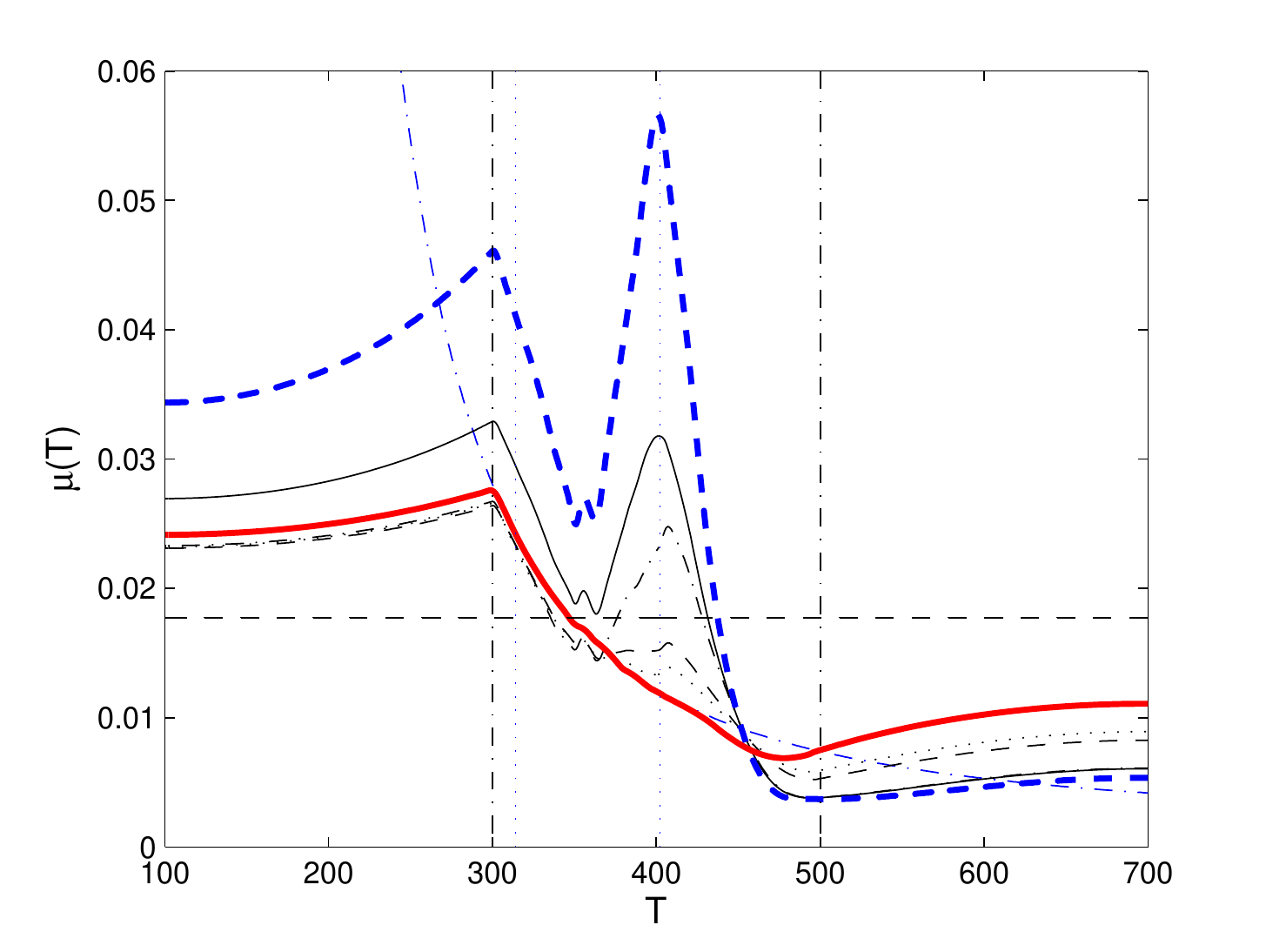}}
\subfigure[]{\includegraphics[width=0.5\textwidth]{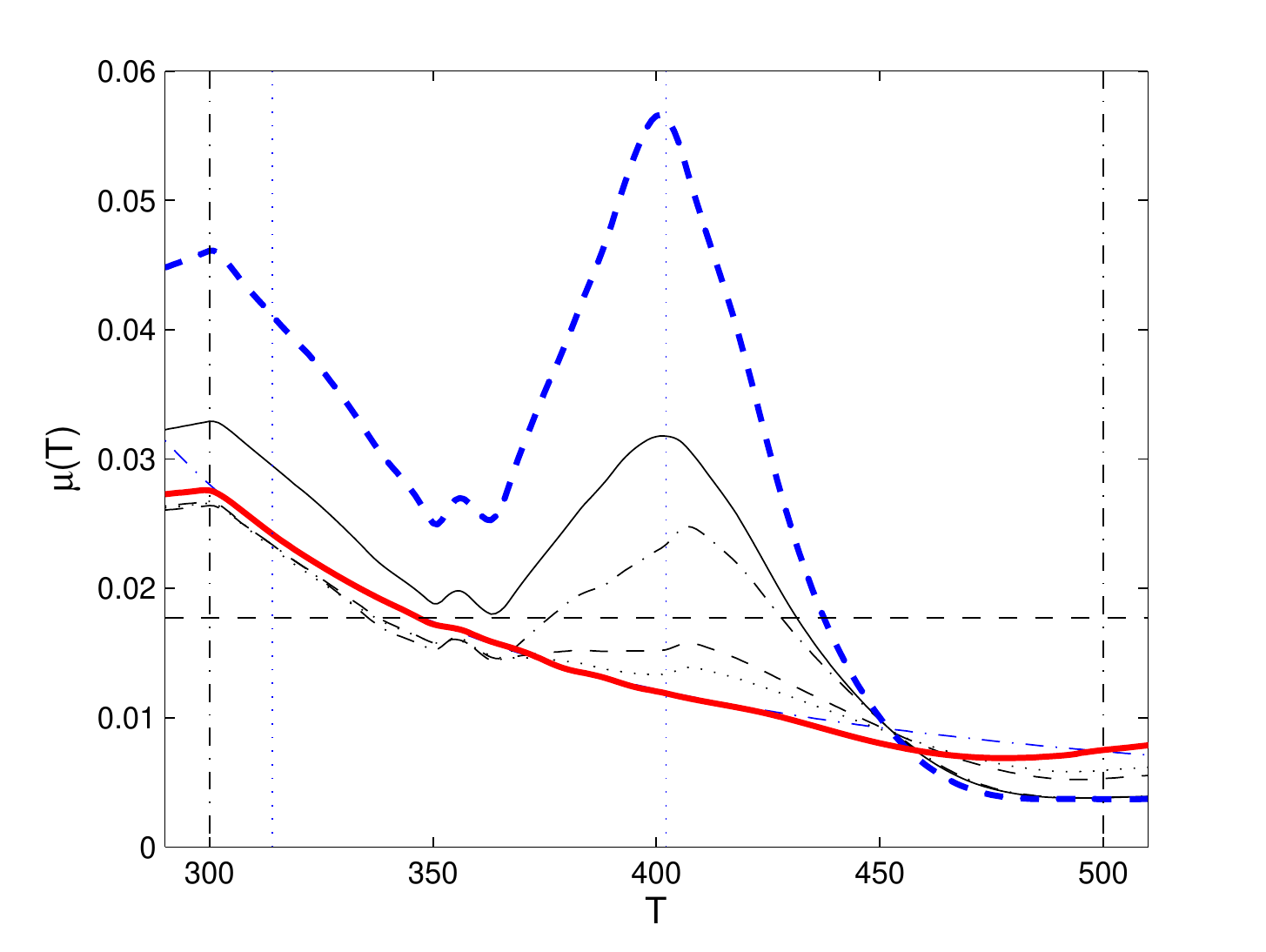}}}
\mbox{
\subfigure[]{\includegraphics[width=0.5\textwidth]{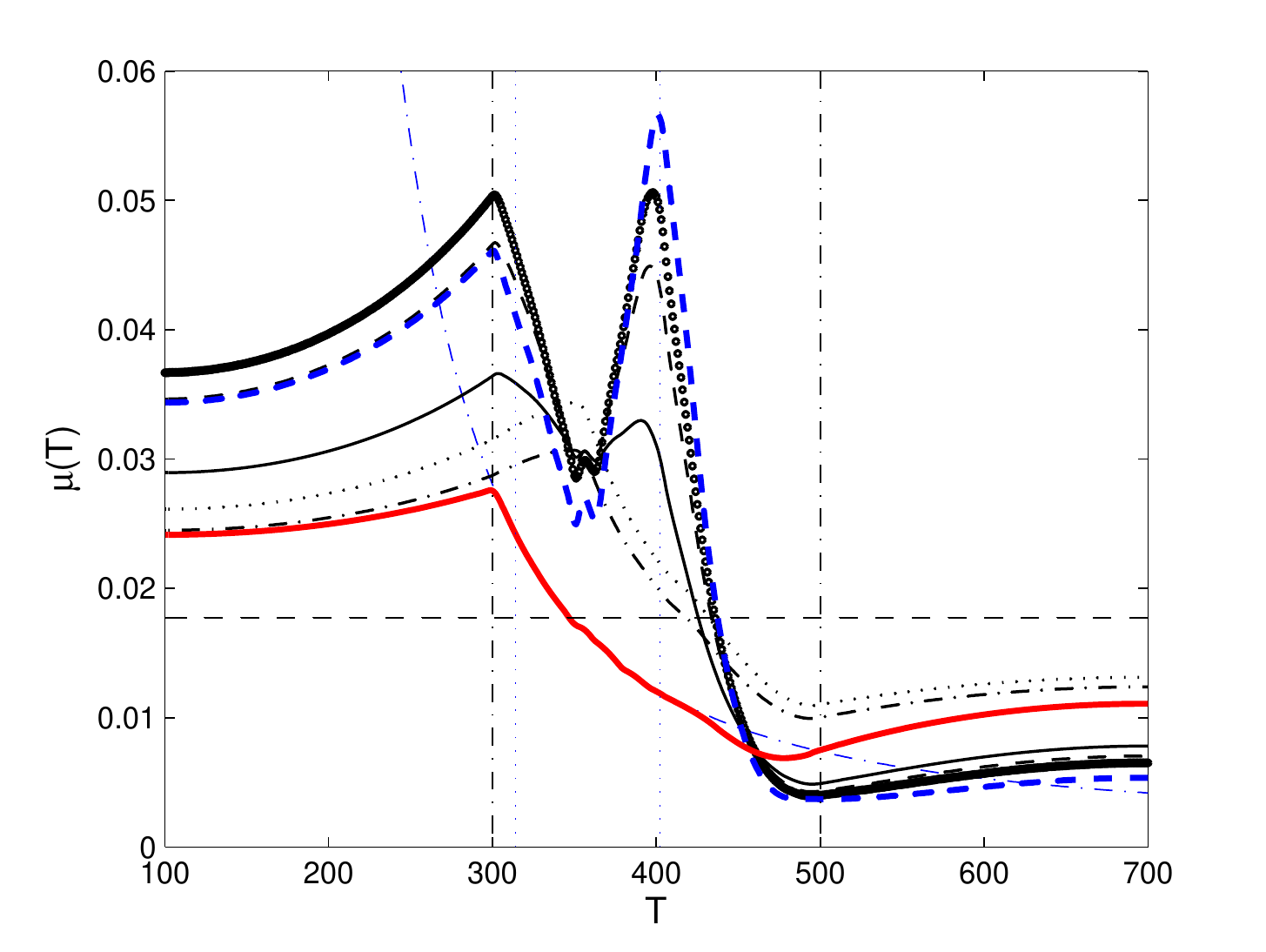}}
\subfigure[]{\includegraphics[width=0.5\textwidth]{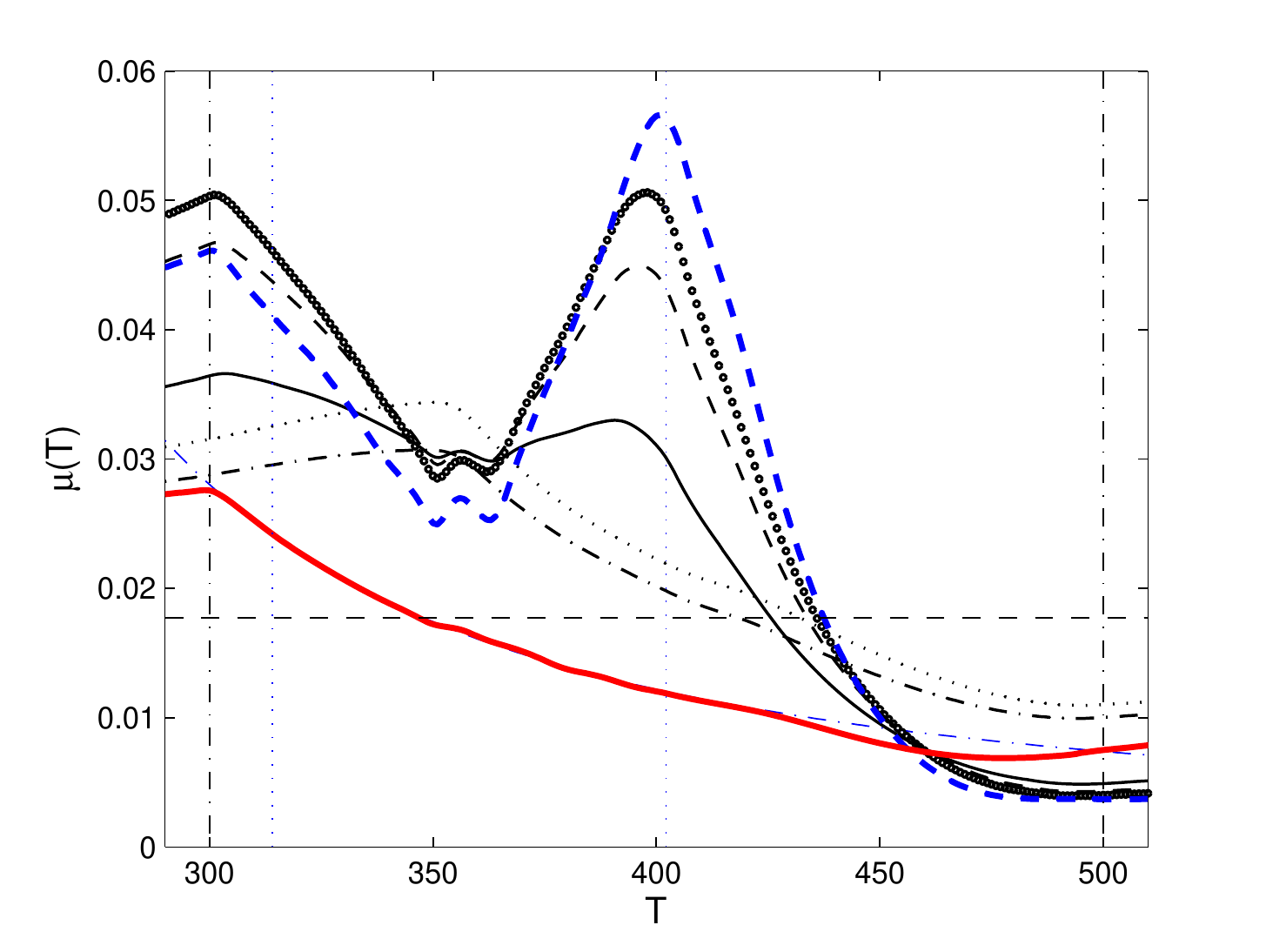}}}
\caption{(a,b) Reconstruction $\hat{\mu}(T)$ of the material property
  obtained in the presence of different noise levels in the
  measurement data: (thick solid line) no noise, (dotted line) 0.05\%,
  (dashed line) 0.1\%, (dash--dotted line) 0.3\%, (thin solid line)
  0.5\%, and (thick dashed line) 1.0\% on (a) the interval $\L$ and
  (b) close--up view showing the identifiability interval $\I$. (c,d)
  Effect of Tikhonov regularization on the reconstruction from the
  measurement data with 1.0\% noise using regularization term
  \eqref{eq:regH1} on (c) the interval $\L$ and (d) close--up view
  showing the identifiability interval $\I$. In both figures (c,d) the
  following values of the regularization parameter were used: (thick
  dashed line) $\lambda = 0$, (circles) $\lambda = 2.5$, (dashed line)
  $\lambda = 6.25$, (thin solid line) $\lambda = 25.0$, (dash--dotted
  line) $\lambda = 250.0$, and (dots) $\lambda = 2500.0$.  For all
  figures the horizontal dashed line represents the initial guess
  $\mu_{0} = 0.0177$; the vertical dash--dotted and dotted lines
  represent, respectively, the boundaries of the identifiability
  interval $\I$ and the measurement span $\M$.}
\label{fig:noise_Tikh}
\end{center}
\end{figure}

\begin{figure}
\begin{center}
\includegraphics[width=0.5\textwidth]{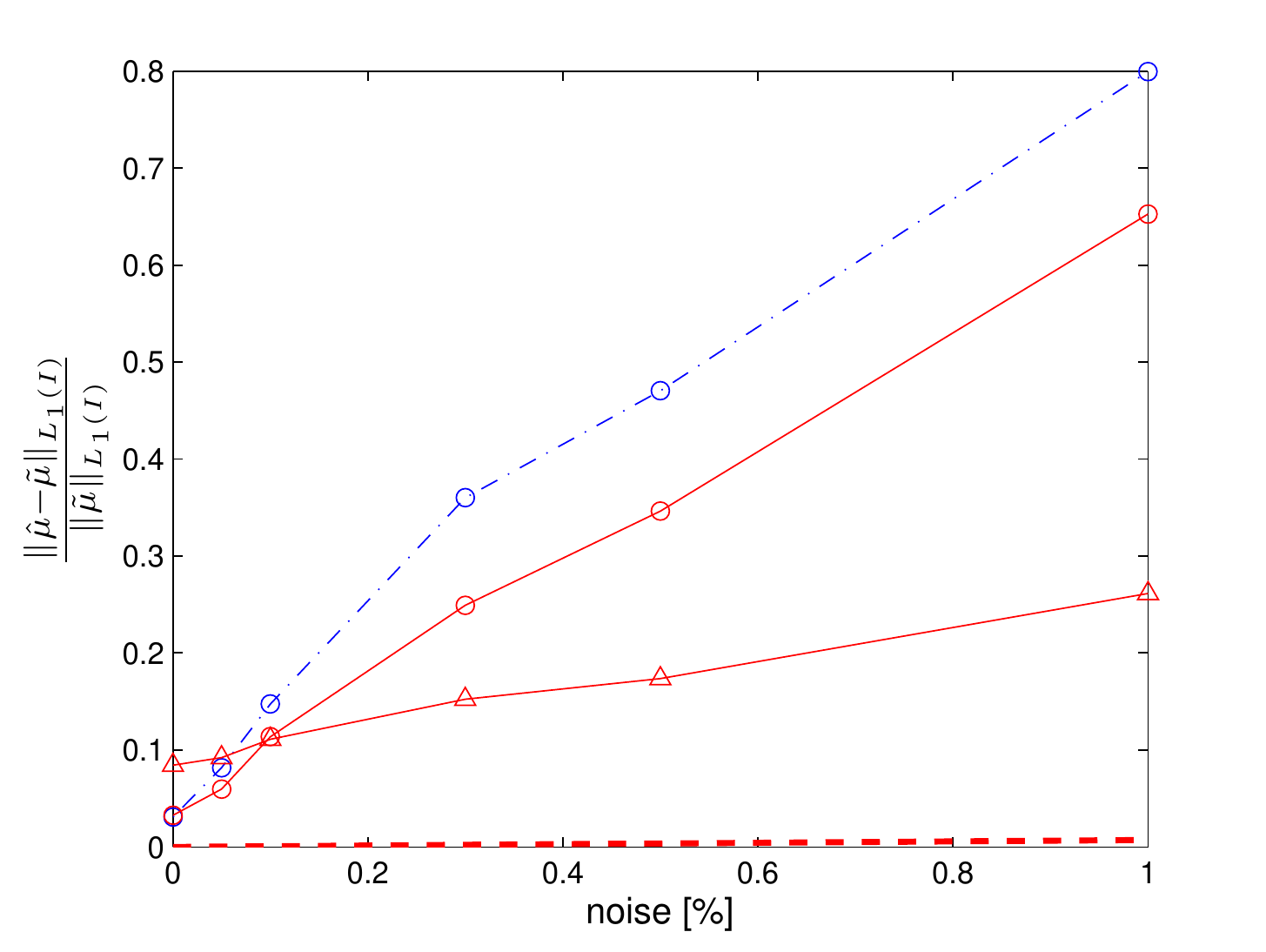}
\end{center}
\caption{Relative $L_1$ reconstruction errors $\|\hat{\mu} -
  \tilde{\mu}\|_{L_1(\I)} \, / \, \|\tilde{\mu}\|_{L_1(\I)}$ obtained
  in the presence of noise with the amplitude indicated and
  averaged over 10 samples: (dash--dotted line) reconstruction
  with Sobolev gradients and without Tikhonov regularization,
  and (solid line) reconstruction with $\dot{H}^1$ Tikhonov
  regularization term \eqref{eq:regH1} [(circles) $\lambda = 2.5$,
  (triangles) $\lambda = 250.0$]. The thick dashed line
  represents the ``error'' in the exact material property
  \eqref{eq:Andrade} obtained by adding noise to $T$
  and averaging over time steps.}
\label{fig:noise_avr}
\end{figure}


\section{Conclusions and Summary}
\label{sec:final}

We developed an optimization--based approach to the problem of
reconstructing temperature--dependent material properties in complex
thermo--fluid systems. As a model problem we considered
two--dimensional unsteady flows in a lid--driven cavity involving also
heat transfer. A key element of this approach is the gradient of the
cost functional which is computed based on a suitably--defined adjoint
system, and is shown to have mathematical structure different than in
most problems involving adjoint--based PDE optimization. We discussed
three different numerical approaches to evaluation of these cost
functional gradients which are given in terms of integrals defined on
the level sets of the temperature field. As compared to earlier work
on the numerical approximation of such expressions
\cite{EngTornTs04,ZahTorn10,Mayo84,Smereka2006,Beale08,MinGib07,MinGib08,Tow09,Tow07},
we also addressed at length the question of the discretization of the
solution space which is specific to our reconstruction problem.
Evidence is shown for the superior performance with respect to the
discretizations of the physical and the solution space, as well as the
computational time, of a new approach to evaluation of gradients which
is proposed in this study.

The reconstruction results obtained demonstrate good performance of
the algorithm, in particular, it is shown that by suitably adjusting
the boundary conditions in the governing system we can extend the
identifiability region. There are two comments we wish to make
concerning these results. As regards the data shown in Figure
\ref{fig:base_T-4}, it might appear somewhat surprising that
reconstructions on longer time windows (which use more information
from the measurements) do not necessarily lead to better results. A
probable reason is that such reconstruction problems defined on longer
time windows tend to be harder to solve numerically, hence the
solutions found may not be global minimizers (or even ``good'' local
ones).  Another comment we have concerns the relatively modest noise
levels for which stable reconstructions could be obtained in Section
\ref{sec:noise}. We note that inclusion of even a small amount of
noise could alter the nature of the reconstructed constitutive
relations. On the other hand, we also remark that reconstructions
performed based on a steady-state problem and reported in \cite{bvp10}
did allow for significantly higher noise levels. We therefore
conjecture that the sensitivity to noise in the present problem is
related to its time-dependent character, as here the effect of
instantaneously small noise levels may be compounded by its continuous
accumulation over the entire time window. Since the flow problem
studied here was admittedly very simple, this issue may certainly
limit the applicability of the proposed method to problems
characterized by higher Reynolds numbers and therefore merits a more
thorough study in the future (this is not the question of the
numerical resolution alone, but rather of the interplay between this
resolution and the complexity of the underlying optimization problem).
Experience with other data assimilation problems may suggest that
acquiring the measurements less frequently in time may actually help
mitigate this effect.

The approach developed in the present study was recently successfully
applied to the problem of system identification involving a dynamical
system in the phase-space representation. More precisely, in
\cite{pnm12} we used this method to reconstruct an invariant manifold
in a realistic reduced-order model of a hydrodynamic instability.  Our
future work will involve extensions of the present approach to more
complicated problems involving systems of coupled PDEs depending on
time and defined on domains in three dimensions.  In the context of
such systems an interesting issue is the reconstruction of anisotropic
constitutive relations. A more challenging problem is related to
reconstruction of constitutive relations in systems involving phase
changes. In addition to the governing PDE in the form of a
free--boundary problem, one would also have to deal with constitutive
relations with distinct functional forms in each of the phases.
Interesting questions also arise in the reconstruction of material
properties defined on the interfaces between different phases, such as
for example the temperature--dependent surface tension coefficient
playing an important role in Marangoni flows. In the study of such
more complicated problems close attention will need to be paid to the
question of ensuring consistency of the reconstructed constitutive
relation with the second principle of thermodynamics, which can be
done by including a form of the Clausius--Duhem inequality
\cite{tpcm08a} in the optimization formulation. In such more general
problems it is not obvious whether this additional constraint can be
eliminated by introducing a slack--variable formulation similar to the
one used in this study, and one may have to perform numerical
optimization in the presence of inequality constraints. These
questions are left for future research.

\section*{Acknowledgments}

The authors wish to acknowledge generous funding provided for this
research by the Natural Sciences and Engineering Research Council of
Canada (Collaborative Research and Development Program), Ontario
Centres of Excellence --- Centre for Materials and Manufacturing, and
General Motors of Canada. The authors are also thankful to Professor
J.~Sesterhenn for bringing to their attention the slack--variable
approach to ensuring the positivity of the reconstructed material
property.

\end{document}